\documentclass[noinfoline]{imsart}

\usepackage{amsfonts,amsthm,amsbsy,graphicx,amsmath, tikz, subcaption, pgfplots} 
\usepackage[font=small]{caption}
\usepackage[OT1]{fontenc}

\usetikzlibrary{shapes,arrows,positioning}
\usepackage{adjustbox}
\usepackage[left=3.6cm,right=3.6cm,top=3.5cm,bottom=3.5cm]{geometry}
\usepackage{natbib}
\usepackage{algorithm}
\usepackage{algpseudocode}

\startlocaldefs
\definecolor{mycolor}{RGB}{32,178,170}
\definecolor{mycolor2}{RGB}{218,165,32}
\definecolor{mycolor3}{RGB}{255,165,0}
\definecolor{mycolor4}{RGB}{255,20,147}
\definecolor{mycolor5}{RGB}{160,32,240}
\definecolor{mycolor6}{RGB}{138,206,234}

\numberwithin{equation}{section}

\newtheorem{theorem}{Theorem}
\newtheorem{assumption}{Assumption}
\newtheorem{corollary}{Corollary}

\newtheorem{lemma}{Lemma}

\newtheorem{remark}{Remark}

\theoremstyle{remark}
\newtheorem{example}{Example}

\newtheorem*{continuedex}{Example \continuedexref\space (cont.)}
\newenvironment{examcont}[1]
  {\newcommand{\continuedexref}{\ref{#1}}\continuedex}
  {\endcontinuedex}

\newcommand{\ind}{\mbox{$\perp \kern-5.5pt \perp$}}

\newcommand{\mycheck}[1]{\hat{#1}}
\newcommand{\g}{g}
\newcommand{\m}{m}

\newcommand{\dee}{d}
\DeclareMathOperator{\sign}{sign}
\DeclareMathOperator{\Var}{Var}
\newcommand{\vertiii}[1]{{\left\vert\kern-0.25ex\left\vert\kern-0.25ex\left\vert #1 
    \right\vert\kern-0.25ex\right\vert\kern-0.25ex\right\vert}}

\begin{document}

\begin{frontmatter}

\title{
Estimation of High-Dimensional Graphical Models Using Regularized Score Matching
}

\begin{aug}
\author{\fnms{Lina} \snm{Lin}\ead[label=e1]{linlina@uw.edu}}, 
\author{\fnms{Mathias} \snm{Drton}\ead[label=e2]{md5@uw.edu}} 
\address{Department of Statistics, University of Washington, Seattle,
  WA  98195, U.S.A.\\
\printead{e1,e2}}
\and
\author{\fnms{Ali} \snm{Shojaie}\ead[label=e3]{ashojaie@uw.edu}}
\address{Department of Biostatistics, University of Washington,
  Seattle, WA 98195, U.S.A.\\
\printead{e3}}

\runauthor{Lin et al.}

\affiliation{University of Washington}

\end{aug}

\begin{abstract}
  Graphical models are widely used to model stochastic dependences
  among large collections of variables.  We introduce a new method of
  estimating undirected conditional independence graphs based on the
  score matching loss, introduced by \citet{hyv2005}, and subsequently
  extended in \citet{hyv2007}.  The \textit{regularized score
    matching} method we propose applies to settings with continuous
  observations and allows for computationally efficient treatment of
  possibly non-Gaussian exponential family models.  In the
  well-explored Gaussian setting, regularized score matching avoids
  issues of asymmetry that arise when applying the technique of
  neighborhood selection,
  and compared to existing methods that directly yield symmetric
  estimates, the score matching approach has the advantage that the
  considered loss is quadratic and gives piecewise linear solution
  paths under $\ell_1$ regularization.  Under suitable
  irrepresentability conditions, we show that $\ell_1$-regularized
  score matching is consistent for graph estimation in  sparse
  high-dimensional settings.  Through numerical experiments and an
  application to RNAseq data, we confirm that regularized score
  matching achieves state-of-the-art performance in the Gaussian case
  and provides a valuable tool for computationally efficient estimation
  in non-Gaussian graphical models.
\end{abstract}



\end{frontmatter}
\maketitle

\section{Introduction}\label{intro}

\textit{Undirected graphical models}, also known as \textit{Markov
  random fields}, are important tools for summarizing dependency
relationships between random variables and have found application in
many fields, including bioinformatics, language and speech processing,
and digital communications.  Each such model is associated to an
undirected graph $G = (V, E)$, with vertex set $V$ and edge set
$E \subset V \times V$.  For a random vector $X = (X_j : j \in V)$
indexed by the nodes of $G$, the graphical model given by $G$ requires
that $X_j$ and $X_k$ be conditionally independent given all other
variables whenever nodes $j$ and $k$ are not joined by an edge in $G$
\citep{Lauritzen1996}.  If $G$ is the smallest graph such that $X$
satisfies this requirement, we term $G$ the \textit{conditional
  independence graph} of $X$.  In this case, $X_j$ and $X_k$ are
conditionally independent given all other variables if and only if $j$
and $k$ are non-adjacent in $G$.  We will always take the vertex set
to be $V = \{1, \ldots , \m\}$, so $m$ is the number of observed
variables in $X$.

Specific models are obtained from additional distributional
assumptions. Particularly, an assumption of multivariate normality
gives Gaussian graphical models, for which estimation of conditional
independence graphs is equivalent to \textit{covariance selection}
\citep{Dempster1972}.  If $X$ is jointly multivariate normal with mean
vector $\mu$ and covariance matrix $\boldsymbol{\Sigma}$---in symbols,
$X \sim N(\mu, \boldsymbol{\Sigma})$---then the conditional
independences among the random variables, and hence edges between nodes in the graph, are determined by the
entries of the inverse covariance, or concentration matrix
$\mathbf{K} = ( \kappa_{jk} ) = \boldsymbol{\Sigma}^{-1}$. More
precisely, $\kappa_{jk} = 0$ for $j \neq k$ if and only if $X_j$ and
$X_k$ are independent given all other variables.

There is a large literature on selection of conditional
independence graphs; see the references in
\citet[Chap.~6]{Edwards2000} or \cite{drton:2007}.  In the last
decade, attention has shifted to high-dimensional settings with the
number of variables $\m$ comparable to or larger than the sample size
$n$.  This scenario arises, for instance, in microarray experiments.
Fortunately, high-dimensional problems may remain tractable in the
presence of 
structural constraints such as \textit{sparsity}, i.e., if each node in the graph
is incident to a small number of edges.  This is of interest for
microarray data
as
gene regulatory
networks are intrinsically sparse \citep{Leclerc2008}.

Gaussian models have been the primary tool for graphical modeling of
data comprising continuous variables, such as gene expression data,
and a large number of methods have been proposed for statistical
estimation in high-dimensional Gaussian graphical models.  
A common strategy involves augmenting a loss function
with a sparsity-inducing penalty such as an $\ell_1$, or lasso
penalty.  Two widely-used approaches are the \textit{graphical lasso}
or \textit{glasso} \citep{YuanLin2007} and \textit{neighborhood
  selection} \citep{meinshausen2006}.  In glasso, an $\ell_1$ penalty
on the entries of the inverse covariance matrix is added to the
negative Gaussian log-likelihood.  Neighborhood selection, on the
other hand, is an $\ell_1$-penalized pseudo-likelihood approach that
leverages the fact that the node-wise full conditional distributions
from a Gaussian graphical model form $\m$ linear regression models.
\cite{meinshausen2006} treat these separate regression models
as having their parameters unrelated, but as we discuss below, methods
that
account for the symmetry in a
concentration matrix have been proposed in subsequent work.

Methods for high-dimensional data have also been developed for
non-Gaussian settings.  \cite{MiyamuraKano2006},
\cite{finegold:drton:2011}, \cite{VogelFried2011} and \cite{SunLi2012}
address robustness to outliers.  \cite{LiuEtAl2009}, \cite{liu2012}
and \cite{Dobra2011} treat Gaussian copula models.
Neighborhood selection/pseudo-likelihood procedures can also be
applied to models for categorical models where the node-wise
regression is logistic or multinomial \citep{LeeEtAl2006,
  HoflingTibshirani2009, RavikumarEtAl2010,JalaliEtAl2011}.
\cite{AllenLiu2012} and \cite{yang:etal:2012} discuss 
extensions using node-wise generalized linear models,
and semi-/nonparametric methods were proposed by
\cite{FellinghauerEtAl2013} and
\cite{VoormanEtAl2014_Biometrika}.

In this paper, we propose a different approach to
high-dimensional graphical model selection.  Addressing the case of
\textit{continuous} but not necessarily Gaussian observations, the
proposed method is based on the \textit{score matching}
loss, first introduced by \citet{hyv2005} in the setting of image
analysis. Recently, \citet{Forbes2014} studied score matching in
Gaussian graphical models with symmetry constraints, and demonstrated
that, when the number of variables $\m$ is fixed, the estimators
derived from the score matching loss are asymptotically efficient in
some special cases, but not in general.  Our focus is instead on the
use of score matching in high-dimensional problems, for which we
consider regularization with an $\ell_1$ penalty.  We will refer to
this graphical model selection technique as \textit{regularized score
  matching}.

Regularized score matching is computationally very convenient for any
exponential family comprising continuous distributions.  Indeed, the
score matching loss is 
a positive
semi-definite quadratic function.  
It follows that the solution path for the regularized score matching problem is
piecewise linear and can be computed in 
entirety.  Moreover,
theoretical analysis 
can be based on familiar
techniques.  Most importantly, as we demonstrate for Gaussian
graphical models, regularized score matching exhibits state-of-the-art
statistical efficiency in high-dimensional settings.  The method also
performs well in our applications to non-Gaussian models, which
include models that seem rather difficult to handle via other methods.

In the Gaussian setting, regularized score matching is structurally
closest to pseudo-likelihood methods with symmetry constraints, such
as \textit{SPACE} \citep{PengEtAl2009}, \textit{symmetric lasso}
\citep{friedman2010tech} and \textit{SPLICE} \citep{rocha2008}.  A
thorough discussion of these different methods is given by
\citet{KhareEtAl2013} who also reformulate the {SPACE}
objective function to ensure convergence
of coordinate descent algorithms.  They abbreviate their method as
\textit{CONCORD}.  For brevity, we refer to these algorithms
collectively as
SPACE. 
We note that in contrast to regularized score matching, the SPACE
methods do not have
piecewise linear solution paths.  Furthermore, as remarked before, the
computational convenience of regularized score matching carries over
to non-Gaussian settings.


A limitation of the original score matching introduced by
\citet{hyv2005} is that it requires the data to be generated from a
distribution whose density is twice differentiable on $\mathbb{R}^\m$.
\citet{hyv2007} proposed a generalization of the approach to the
important case of non-negative data.  For exponential families, the
non-negative score matching loss is again a semidefinite quadratic
function.  We explore regularization of the non-negative score
matching loss as a tool for estimation of conditional independence
graphs from high-dimensional non-negative data, and we establish
consistency of the method.

The remainder of the paper is organized as follows.
Section~\ref{background} provides the needed background on score
matching and its applications. In Section~\ref{sec:rsme}, we describe
the proposed method, \textit{regularized} score matching. Implementation details are given in Appendix A. 
In Section~\ref{simulations}, we present results of numerical
experiments to compare the performance of the procedure with existing
approaches.  An application to RNAseq data is given in
Section~\ref{real}.  Section~\ref{theory} provides sparsistency theory
for both basic and non-negative regularized score matching. Proofs are given in Section~\ref{sec:proofs} with details deferred to Appendix B and C.  We end
with a discussion in Section~\ref{discussion}.  

\subsection*{Notation}
The following notational conventions are used throughout the paper:
\begin{enumerate}
\item Random variables/vectors are denoted by upper case letters; lower case letters are used for observed values.  So, $x\in\mathbb{R}^m$ is an observed value of the random vector $X$.  Similarly,
  $\mathbf{x}=(x_{ij})\in\mathbb{R}^{n\times m}$ is a matrix of
  observed values, which will typically hold the realizations of $n$
  i.i.d.~copies of $X$ in its rows.  We index the columns of a matrix
  with subscripts, so $x_j$ refers to the $j$th column of
  $\mathbf{x}$.  Superscripts in parentheses are used to refer to the
  rows of a matrix, so $x^{(i)}$ is the $i$th row of $\mathbf{x}$.
\item For a matrix $\mathbf{U}=(u_{ij})\in\mathbb{R}^{m\times m}$, we
  denote the vectorization obtained by stacking columns by
  \[
  \mbox{vec}(\mathbf{U})=\begin{pmatrix} u_{11}, u_{21}, \ldots ,
    u_{\m1},  \ldots , u_{1\m}, \ldots , u_{\m\m} \end{pmatrix}^T.  
  \]
\item Let $a,b \in [1, \infty]$.  We denote the $\ell_a$ norm of
  a vector $u \in \mathbb{R}^\m$ by
  \[
  \| u \|_a \;=\; \bigg(\sum_{i=1}^m |u_i|^a\bigg)^{1/a}
  \]
and write $\vertiii{\mathbf{U}}_{a, b} = \max_{\|\boldsymbol{x}\|_a = 1}
\|\mathbf{Ux}\|_b$ for the $\ell_a/\ell_b$ operator norm
of a matrix $\mathbf{U} \in \mathbb{R}^{\m \times \m}$.
We let
$\vertiii{\mathbf{U}}_\infty=\vertiii{\mathbf{U}}_{\infty, \infty}$
and    $\|\mathbf{U}\|_a=\|\mbox{vec}(\mathbf{U})\|_a$.

%

\end{enumerate}

\section{Score Matching}\label{background}

We begin with an overview of Hyv\"arinen's score matching, discussing
first random vectors supported on all of $\mathbb{R}^\m$ and then
random vectors supported on the nonnegative orthant.  We also review
the convenient form of the score matching estimating equations in exponential families.

\subsection{Basic score matching}
\label{sec:basic-score-matching}

Suppose $X$ is a continuous random vector taking values in
$\mathbb{R}^\m$, with joint distribution $P$.  Suppose further that
$P$ belongs to the family $\mathcal{P}$ that comprises all probability
distributions with support equal to $\mathbb{R}^\m$ and a twice
differentiable density with respect to Lebesgue measure.  We emphasize
that in a statistical context the differentiability requirement is
with respect to data.  We write $p$ to denote the density of $P$ and
adopt the usual notation for the gradient and Laplacian
\begin{align*}
  \nabla f(x) &= \left\{\frac{\partial}{\partial x_j} f(x)\right\} \in
              \mathbb{R}^\m, &
\Delta f(x) &= \sum_{j=1}^\m \frac{\partial^2}{\partial x_j^2} f(x) \in \mathbb{R}, 
\end{align*}
of a function $f:\mathbb{R}^\m\to\mathbb{R}$.

For a distribution $Q \in \mathcal{P}$ with density $q$, define
the divergence function
\begin{equation}\label{divergence}
J(Q) = \int_{\mathbb{R}^\m} p(x)\left[\| \nabla \log q(x) -\nabla \log
  p(x) \|_2^2\right] \,dx 
\end{equation}
as the expected squared distance between the gradients of the
log-densities of the two distributions $Q$ and $P$.  By choosing $Q$
to minimize (\ref{divergence}), we are matching `scores' with respect
to the data vector $x$.  Hence, (\ref{divergence}) has been referred
to as the \textit{score matching loss}.  It is evident from
(\ref{divergence}) that the score matching loss is uniquely minimized
when $Q = P$.

Upon initial inspection, optimization of $J(Q)$ seems to require knowledge of
$P$ in an important way.  However, \citet{hyv2005} showed that, under
mild regularity conditions, the score matching loss (\ref{divergence})
can be rewritten as:
\begin{equation}\label{scorematch}
J(Q) = \int_{\mathbb{R}^\m} p(x)\left[\Delta \log q(x) + \frac{1}{2}\|\nabla \log q(x)\|_2^2\right]dx \; + \; \mbox{const},
\end{equation}
where `const' refers to a term independent of $Q$.  The key term in
the integrand in~(\ref{scorematch}) is the so-called Hyv\"arinen
scoring rule
\[
S(x, Q) = \Delta \log q(x) + \frac{1}{2}\|\nabla \log q(x)\|_2^2.
\]
The integral in~(\ref{scorematch}) admits an empirical version in
which the integration with respect to $P$ is replaced by an average
over an observed sample, which we arrange into a data matrix
$\mathbf{x}\in\mathbb{R}^{n\times m}$.  This leads to the
\textit{empirical score matching loss}
\begin{equation}
  \label{eq:emp-sm-loss}
\hat{J}(\mathbf{x}, Q) = \frac{1}{n}\sum_{i=1}^n S(x^{(i)}, Q),
\end{equation}
and the \textit{score matching estimator} (SME)
\[
\mycheck{Q} =  \mbox{arg} \ \underset{Q}{\mbox{min}} ~ \hat{J}(\mathbf{x}, Q).
\]

The score matching loss $J(Q)$ was motivated by problems
involving models whose distributions have an intractable normalization
constant.
Indeed, evaluating (\ref{scorematch}) and computing the SME
$\mycheck{Q}$ requires no knowledge of the normalization constant,
which is eliminated upon taking logarithmic derivatives with respect
to $x$.
 Besides the imaging
problems considered by \cite{hyv2005}, score matching has been applied
to spatial statistics \citep{dawid2013} and neural networks
\citep{koster2007, vincent2011, le2011}.

The statistical properties of SMEs in classical large sample settings
have been investigated by \citet{hyv2005, hyv2007} and
\cite{Forbes2014}.  In particular, it has been shown that, under the
usual regularity conditions, SMEs are asymptotically consistent and
normal in large-sample theory. However, SMEs are not necessarily
asymptotically efficient.

\subsection{Extension to non-negative data} 
\label{sec:extens-non-negat}

The partial integration arguments underlying (\ref{scorematch}) may
fail to apply when considering distributions $Q$ that are not
supported on all of $\mathbb{R}^\m$.  In particular, when $Q$ is taken
to be from $\mathcal{P}_+$, i.e. the family of
distributions that are supported on
$ \mathbb{R}^\m_{+} =[0,\infty)^\m$ with Lebesgue densities that are
twice differentiable on $(0,\infty)^\m$, then partial integration
may not be possible due to discontinuities at points with zero
coordinates.  We thus consider the non-negative score matching loss,
\begin{equation}\label{nndivergence}
J_+(Q) = \int_{\mathbb{R}^\m_{+}} p(x)\left[ \bigg\|\nabla \log q(x) \circ x - \nabla \log p(x) \circ x \bigg\|_2^2 \right]dx,
\end{equation}
as proposed in \citet{hyv2007}. Here, `$\circ$' stands for the
Hadamard product, that is, element-wise multiplication.

The score matching loss (\ref{divergence}) can be thought of as a
function of the Euclidean distance between the gradients of the model
density $q$ and true density $p$ with respect to a hypothetical
location parameter $\mu$, evaluated at $0$.  That is, we may write
(\ref{divergence}) as
\[ 
J(Q) = \int_{\mathbb{R}^\m}  p(\mathbf{x})\left[\big\| \nabla_{\mu = 0} \log q(x + \mu) -\nabla_{\mu = 0} \log p(x + \mu) \big\|_2^2\right]dx.
\]
Likewise, the non-negative score matching loss compares the gradient
of the model density $q$ and true density $p$ with respect to a
hypothetical scale parameter $\sigma$ evaluated at $1$,
\[ 
J_{+}(Q) = \int_{\mathbb{R}^\m_+}  p(\mathbf{x})\left[\big\| \nabla_{\sigma = 1} \log q(x \circ \sigma) -\nabla_{\sigma = 1} \log p(x \circ \sigma) \big\|_2^2\right]dx.
\] 
Under suitably adjusted regularity conditions, \cite{hyv2007} showed
that the non-negative score matching loss from (\ref{nndivergence})
can be simplified into
\begin{equation}\label{nnscorematch}
J_{+}(Q) =  \int_{\mathbb{R}^\m_{+}} p(x) S_+(x, Q)\, dx \;+\; \mbox{const}
\end{equation}
with scoring rule
\begin{equation}
S_+(x, Q) = {\sum_{j=1}^\m \left[2x_j \frac{\partial \log q(x)}{\partial x_j} + x_j^2\frac{\partial^2 \log q(x)}{\partial x_j^2} + \frac{1}{2}x_j^2 \left(\frac{\partial \log q(x)}{\partial x_j}\right)^2 \right]}.
\end{equation}
For a data matrix $\mathbf{x}\in\mathbb{R}^{n\times m}$, one obtains
the \textit{empirical non-negative score matching loss}
\begin{equation}
  \label{eq:emp-sm-loss+}
\hat{J}_+(\mathbf{x}, Q) = \frac{1}{n}\sum_{i=1}^n S_+(x^{(i)}, Q),
\end{equation}
and the \textit{non-negative score matching estimator}
($\mbox{SME}_+$)
\[
\mycheck{Q}_+ = \mbox{arg} \ \underset{Q}{\mbox{min}}
~\hat{J}_+(\mathbf{x}, Q). 
\]
Again, under the usual regularity conditions, the estimator
$\mycheck{Q}_+$ is asymptotically consistent and normal in traditional
large-sample theory.

\subsection{Score matching in exponential families} 

\citet{hyv2007} and \citet{Forbes2014} have shown that the SME has a
convenient closed form as a rational function of the data when
$\mathcal{P}$ is an exponential family. \citet{hyv2007} showed the
same for $\mbox{SME}_+$ for the example of truncated normal
distributions.  As they provide the basis for our later work, we
revisit these results for both SME and $\mbox{SME}_+$.

Let $\mathcal{P}=(Q_\theta:\theta\in\Theta)$ be an exponential family
with natural parameter space $\Theta$.  Suppose that the distributions
$Q_\theta$ have their common support equal to either
$\mathcal{X}=\mathbb{R}^\m$ or $\mathcal{X}=\mathbb{R}^\m_+$, and that
$\mathcal{P}$ is dominated by Lebesgue measure on $\mathbb{R}^\m$.
Assuming that the sufficient statistics $t(x)$ take values in
$\mathbb{R}^s$, the log-densities of the distributions $Q_\theta$ 
have the form
\begin{equation}\label{exponentialfamily}
  \log q(x | \theta) = \theta^T t(x) - \psi(\theta) +
  b(x), \quad x\in\mathcal{X},
\end{equation}
and
\begin{equation}\label{exponentialpar}
\Theta = \left\{ \theta \in \mathbb{R}^s \::\: \psi(\theta) = \log
  \int_{\mathcal{X}} e^{\theta^T t(x)}dx < \infty
\right\}. 
\end{equation}

\begin{lemma}\label{quadraticlemma}
  Let $\mathbf{x}\in\mathbb{R}^{n\times m}$ be a data matrix, and
  suppose $\mathcal{P}=(Q_\theta:\theta\in\Theta)$ is an exponential
  family characterized by (\ref{exponentialfamily}) and
  (\ref{exponentialpar}).  If $\mathcal{P}$ has support
  $\mathcal{X}=\mathbb{R}^\m$, then the empirical score matching loss
  $\hat{J}(\mathbf{x}, Q_\theta)$ is a quadratic function in
  $\theta$ with
  \begin{equation}\label{quadratic}
    \hat{J}(\mathbf{x}, Q_\theta) = \frac{1}{2}~\theta^T
    \boldsymbol{\Gamma}(\mathbf{x}) \theta + 
    \g(\mathbf{x})^T\theta + c(\mathbf{x}), 
  \end{equation}
  where $\boldsymbol{\Gamma}(\mathbf{x})$ is a positive semidefinite
  $s \times s$ matrix, and $\g(\mathbf{x})$ is an $s$-vector.  The
  same is true for $\hat{J}_+(\mathbf{x}, Q_\theta)$ when
  $\mathcal{P}$ has support $\mathcal{X}=\mathbb{R}_+^\m$.  
\end{lemma} 
\begin{proof}
  For $j=1,\dots,m$ and $x\in\mathbb{R}^\m$, define the $s$-vectors
  \begin{align*}
    h_j(x)&= \frac{\partial}{\partial x_j} t(x), &h_{jj}(x)&=
                             \frac{\partial^2}{\partial
                             x_j^2} t(x).
  \end{align*}
  It then
  follows from~(\ref{exponentialfamily}) that
  $\hat{J}(\mathbf{x}, Q_\theta)$ can be expressed in the claimed
  form with
  \begin{align}
    \boldsymbol{\Gamma}(\mathbf{x}) 
    &= \frac{1}{n}\sum_{i=1}^n
      \sum_{j=1}^m h_j(x^{(i)})h_j(x^{(i)})^T,
      \label{gammaform} \\ 
    \g(\mathbf{x}) 
    &= \frac{1}{n}\sum_{i=1}^n \sum_{j=1}^m \left(\frac{\partial}{\partial
      x_j}b(x^{(i)})\right)h_j(x^{(i)})^T
      + \Delta t(x^{(i)}), 
    \\ 
    c(\mathbf{x}) &= \frac{1}{n}\sum_{i=1}^n \frac{1}{2}\left\|\nabla b(x^{(i)}) \right\|_2^2 + \Delta b(x^{(i)}). 
  \end{align}
  For non-negative score matching,
  $\hat{J}_+(\mathbf{x}, Q_\theta)$ admits the claimed form with
\begin{align}
\boldsymbol{\Gamma}(\mathbf{x}) &= \frac{1}{n}\sum_{i=1}^n
                                  \sum_{j=1}^\m x_{ij}^2
                                  h_j(x^{(i)})h_j(x^{(i)})^T, 
\label{gammaform+}\\  
\g(\mathbf{x}) &= \frac{1}{n}\sum_{i=1}^n \sum_{j=1}^\m \left(\frac{\partial}{\partial x_j}b(x^{(i)})\right)h_j(x^{(i)})^T + x_{ij}^2h_{jj}(x^{(i)})^T + 2x_j^{(i)}h_j(x^{(i)})^T, \\  
c(\mathbf{x}) &= \frac{1}{n}\sum_{i=1}^n\sum_{j=1}^\m \frac{1}{2} x_{ij}^2\left(\frac{\partial}{\partial x_j}b(x^{(i)})\right)^2 + x_{ij}^2\frac{\partial^2}{\partial x_j^2}b(x^{(i)}) + 2 x_{ij} \frac{\partial}{\partial x_j}b(x^{(i)}),
\end{align}
where the $x_{ij}$ are the entries of the $n\times m$ data matrix $\mathbf{x}$.
\end{proof}

Lemma \ref{quadraticlemma} implies that, when working with exponential
families, both score matching objectives are quadratic functions of the
unknown parameter vector $\theta$.  A score matching estimator
$\mycheck{\theta}$ thus satisfies a set of \textit{linear} estimating
equations 
\begin{align}\label{esteq}
  \mycheck{\theta}^T\boldsymbol{\Gamma}(\mathbf{x})+ \g(\mathbf{x}) 
  &= 0.
\end{align}


\subsection{Pairwise interaction models}

The most basic class of exponential families that appear in graphical
modeling are pairwise interaction models with log-densities 
\begin{equation}\label{expfam:pairwise}
  \log q(x | \theta) \;=\; \sum_{1\le j\le k\le \m} \theta_{jk}
  t_{jk}(x_j,x_k) - \psi(\theta) + b(x), \quad x\in\mathcal{X}\subseteq
  \mathbb{R}^\m.
\end{equation}
Here, the $t_{jk}$ are sufficient statistics that depend only on the
$j$th and $k$th coordinate of $x$, and the $\theta_{jk}$ are
interaction parameters.  If $Q_\theta$ denotes the distribution with
density given by~(\ref{expfam:pairwise}), then the Hammersley-Clifford
Theorem implies that an edge between nodes $j$ and $k$ exists in the
conditional independence graph of $Q_\theta$ if and only if
$\theta_{jk}$ is nonzero.  The specific models we consider later
either exactly have the form in~(\ref{expfam:pairwise}) or are closely
related extensions with log-densities
\begin{align}\label{expfam:pairwise:general}
  \log q(x | \theta) \;=\; \sum_{a=1}^A\sum_{j\le k}
  \theta^{(a)}_{jk} t^{(a)}_{jk}(x_j,x_k) 
  +\sum_{l=1}^L\sum_{j=1}^m
  \theta^{(l)}_jt^{(l)}_j(x_j) - \psi(\theta) + b(x), 
\end{align}
where pairwise interactions may be of $A$ different types and we also
include $L$ sets of sufficient statistics $t^{(l)}_j$ depending on the
individual coordinates.  The latter appear, for instance, when
allowing  distributions to vary in location.   The
distribution $Q_\theta$ defined by~(\ref{expfam:pairwise:general}) has
no edge between $j$ and $k$ in its conditional independence graph if
and only if $\theta^{(1)}_{jk}=\dots=\theta^{(A)}_{jk}=0$.

In our study of score matching methods for models of the
type~(\ref{expfam:pairwise}) or~(\ref{expfam:pairwise:general}), it
will be convenient to introduce the symmetric $\m\times \m$
interaction matrix $\boldsymbol{\Theta}$ with entries
\[
\boldsymbol{\Theta}_{jk} =
\begin{cases}
  \theta_{jk} &\text{if } \ j\le k,\\
  \theta_{kj} &\text{if } \ j> k.
\end{cases}
\]

\begin{lemma}\label{blockform}
  Let $\mathcal{P}$ to be the pairwise interaction model given
  by~(\ref{expfam:pairwise}) with symmetric $m\times m$ interaction
  matrix $\boldsymbol{\Theta}$.  If $\mathcal{P}$ has support
  $\mathcal{X}=\mathbb{R}^\m$, then the empirical score matching loss
  $\hat{J}(\mathbf{x}, Q_\theta)$ equals
  \begin{equation}\label{quadratic-block}
  \frac{1}{2} \mbox{vec}(\boldsymbol{\Theta})^T
  \boldsymbol{\Gamma}(\mathbf{x}) \mbox{vec}(\boldsymbol{\Theta}) +  
    \g(\mathbf{x})^T\mbox{vec}(\boldsymbol{\Theta}) + c(\mathbf{x})
  \end{equation}
  for a symmetric $m^2\times m^2$ matrix
  $\boldsymbol{\Gamma}(\mathbf{x})$ that is block-diagonal, with all
  blocks of size $m\times m$.  The same is true for
  $\hat{J}_+(\mathbf{x}, Q_\theta)$ when $\mathcal{P}$ has support
  $\mathcal{X}=\mathbb{R}_+^\m$.
\end{lemma}
\begin{proof}
  By~(\ref{gammaform}) and~(\ref{gammaform+}), it suffices to show
  that there exists a block-diagonal matrix $\boldsymbol{\Gamma}_j(x)$
  such that
  \begin{equation}
    \label{eq:going-to-vec}
    \theta^T h_j(x)h_j(x)^T\theta \;=\; \mbox{vec}(\boldsymbol{\Theta})^T
  \boldsymbol{\Gamma}_j(x) \mbox{vec}(\boldsymbol{\Theta}),
  \end{equation}
  where $\theta=(\theta_{jk}:j\le k)$.
  Now,
  \begin{align*}
    h_j(x)^T\theta 
    &=   \sum_{k\ge j} \frac{\partial}{\partial x_j} t_{jk}(x_j,x_k) \theta_{jk} +
      \sum_{k<j} 
      \frac{\partial}{\partial x_j} t_{kj}(x_k,x_j) \theta_{kj}\\
    &=   \sum_{k\ge j} \frac{\partial}{\partial x_j} t_{jk}(x_j,x_k)
      \boldsymbol{\Theta}_{kj} +    \sum_{k<j} 
      \frac{\partial}{\partial x_j} t_{kj}(x_k,x_j) \boldsymbol{\Theta}_{kj}.
  \end{align*}
  Define a vector $\bar h_j(x)\in\mathbb{R}^{m^2}$, indexed by pairs
  $(k,l)$ with $1\le k,l\le m$, by setting the entries to
  \begin{equation}
    \label{eq:barh}
  \bar h_j(x)_{kl} =
  \begin{cases}
    \frac{\partial}{\partial x_k} t_{kl}(x_k,x_l) &\text{if } \ j=k\le
    l,\\
    \frac{\partial}{\partial x_k} t_{l k}(x_k,x_j) &\text{if } \
    j=k>l,\\
    0 &\text{if } \ j\not=k.
  \end{cases}
  \end{equation}
  Then $h_j(x)^T\theta=\bar h_j(x)\mbox{vec}(\boldsymbol{\Theta})$ and
  (\ref{eq:going-to-vec}) holds with $\boldsymbol{\Gamma}_j(x)=\bar h_j(x)\bar h_j(x)^T$,
  which is block-diagonal as it is zero with the exception of the
  $m\times m$ block indexed by pairs $(k,l)$ with $k=j$.
\end{proof}

\begin{remark}
  \label{rem:blockform-gen}
  \rm When $\mathcal{P}$
  is a model as specified in~(\ref{expfam:pairwise:general}), then the
  empirical (non-negative) score matching loss may still be
  represented as an explicit quadratic form with a block-diagonal
  symmetric matrix $\boldsymbol{\Gamma}(\mathbf{x})$
  as in~(\ref{quadratic-block}).  However,
  $\boldsymbol{\Gamma}(\mathbf{x})$
  is then of size $(Am^2+Lm)\times
  (Am^2+Lm)$, and its $m$ diagonal blocks are of size $(Am+L)\times
  (Am+L)$.  The
  $j$th
  block has its rows and columns corresponding to the $j$th
  columns of each of
  $\boldsymbol{\Theta}^{(1)},\dots,\boldsymbol{\Theta}^{(A)}$
  as well $(\theta^{(1)}_j,\dots,\theta^{(L)}_j)$.  
\end{remark}

\begin{example}\label{example1}
  If the exponential family is taken to be the family of centered
  multivariate normal distributions with precision matrix
  $\mathbf{K}=(\kappa_{jk})$, then the support is
  $\mathcal{X} = \mathbb{R}^\m$ and 
  \begin{align}\label{normal}
    q(x | \mathbf{K}) &\propto \exp\left\{-\frac{1}{2}x^T
                        \mathbf{K} x \right\}, \ \ \ x \in
                        \mathbb{R}^\m. 
  \end{align}
  With
  \begin{align*}
    \nabla \log q(x|\mathbf{K}) 
    &= -\mathbf{K}x, 
    &  \Delta \log q(x|\mathbf{K}) &= -\sum_{j=1}^\m \kappa_{jj}, 
  \end{align*}
  and dropping a term that is constant in $\mathbf{K}$, the empirical
  score matching loss from (\ref{scorematch}) takes the form
  \begin{equation}\label{gaussiansme} 
    -\,\mbox{tr}(\mathbf{K}) +
    \frac{1}{2}\mbox{tr}(\mathbf{K}\mathbf{K}\mathbf{W}), 
  \end{equation}
  where 
  \[
  \mathbf{W} = \frac{1}{n}\sum_{i=1}^n
  x^{(i)} x^{(i)T}
  \]
  is the empirical covariance matrix (under knowledge of zero mean).
  %
  Lemma~\ref{blockform} applies with
  $t_{jk}(x_j,x_k)=x_jx_k$, in which case the matrix
  $\boldsymbol{\Gamma}_j(x)$ constructed in the proof of the lemma
  does not depend on $j$, other than through the location of the
  nonzero block.  Indeed,~(\ref{quadratic-block}) holds with
  $\boldsymbol{\Gamma}(\mathbf{x}) = \mathbf{I}_{\m \times \m} \otimes
  \mathbf{W}$
  and $\g(\mathbf{x})=\mbox{vec}(\mathbf{I}_{\m\times \m})$, where 
  $\mathbf{I}_{\m\times \m}$ is the $\m \times \m$ identity matrix.
  Clearly, $\boldsymbol{\Gamma}(\mathbf{x})$ is positive definite if
  and only if $\mathbf{W}$ is as well.  If $\mathbf{W}$ is invertible then
  SME of $\mathbf{K}$ is $\mycheck{\mathbf{K}} =\mathbf{W}^{-1}$ and
  coincides with the maximum likelihood estimator.
\end{example}

\begin{example}\label{example2}
  Consider truncated normal densities of the form
  \begin{align}\label{truncnorm}
    q(x | \mathbf{K}) &\propto \exp\left\{-\frac{1}{2}x^T
                        \mathbf{K} x \right\}, \ \ \ x \in
                        \mathbb{R}^\m_+. 
  \end{align}
  Using $\kappa_j$ to denote the $j$th column of $\mathbf{K}$, it can
  be shown that the empirical non-negative score matching objective is
  \begin{equation}\label{nngaussiansme}
    \frac{1}{n}\sum_{i = 1}^n \sum_{j = 1}^\m 2x_{ij} x^{(i)T}
    \kappa_j - x_{ij}^2\kappa_{jj} + \frac{1}{2}
    \kappa_j^T\left(x_{ij}^2x^{(i)}x^{(i)T}\right) \kappa_j.  
  \end{equation}
  The loss can be written as in~(\ref{quadratic}) with
  $\boldsymbol{\Gamma}(\mathbf{x})$ a block diagonal
  $\m^2 \times \m^2$ matrix, whose $j$th block is given by
  \[
  \frac{1}{n}\sum_{i=1}^n x_{ij}^2x^{(i)}x^{(i)T}.
  \]
  Moreover, $\g(\mathbf{x}) = 2w + w_{\mbox{\tiny{diag}}}$, where
  $w = \mbox{vec}(\mathbf{W})$ and
  $w_{\mbox{\tiny{diag}}} = \mbox{vec}(\mbox{diag}(\mathbf{W}))$.
  The maximum likelihood estimator for $\mathbf{K}$ has
  no closed form due to intractable normalizing constants.
\end{example}

\begin{example}\label{example2.5}
  Finally, 
  consider the
  family of distributions with densities of the form
  \begin{align}
    q(x | \mathbf{B}^{(2)},\mathbf{B},\mathbf{b}) 
    \label{condnorm}
    &\propto \exp\bigg\{ \sum_{1\le j \not= k\le \m} \beta^{(2)}_{jk}x_j^2x_k^2 +  \sum_{j,k=1}^\m \beta_{jk}x_j x_k +  \sum_{j=1}^\m \beta_{j}x_j  \bigg\} , & x \in \mathbb{R}^\m.
  \end{align}
  Here, $\mathbf{b}=(\beta_1,\dots,\beta_\m)^T$ is an $\m$-vector, and
  $\mathbf{B}=(\beta_{jk})$ and
  $\mathbf{B}^{(2)} = (\beta^{(2)}_{jk})$ are symmetric $\m\times\m$
  interaction matrices, the latter having a zero diagonal.  This
  family is 
  a class of distributions with normal
  conditionals, 
  with densities that need not be unimodal
  \citep{arnold:1999,gelman:meng:1991}.  This family is intriguing
  from the perspective of graphical modeling as, in contrast to the
  Gaussian case, conditional dependence may also express itself in the
  variances.  For conditional independence of $X_j$ and $X_k$ both
  $\beta_{jk}$ and $\beta^{(2)}_{jk}$ need to vanish.

  By Remark~\ref{rem:blockform-gen}, the empirical score matching loss
  for the family from~(\ref{condnorm}) can be written as a quadratic
  function with the quadratic term given by block-diagonal matrix
  $\boldsymbol{\Gamma}(\mathbf{x})$ of size $(2m^2+m)\times (2m^2+m)$.
  The blocks are of size $(2m+1)\times (2m+1)$, and the $j$th block
  has its rows and columns corresponding to the $j$th columns of
  $\mathbf{B}$ and $\mathbf{B}^{(2)}$ and the $j$th entry in $\mathbf{b}$.
\end{example}


\section{Regularized Score Matching}\label{sec:rsme}

In this section, we propose the use of \textit{regularized score
  matching} for graphical model selection in the setting of 
high-dimensional sparse graphical models.  We
begin by discussing the proposed method and its implementation.  Later sections show that, despite the fact that SMEs need not be asymptotically efficient in the sense of traditional large-sample theory, regularized score matching achieves state-of-the-art statistical performance in high-dimensional problems, all the while  allowing seemingly complicated non-Gaussian graphical models to be treated in a computationally efficient manner. 


\subsection{Methodology}\label{method}

Building on the ideas underlying methods such as glasso, neighborhood
selection and SPACE, we augment the score matching loss with a
sparsity-promoting penalty.  Our focus is on the most basic case of an
$\ell_1$ penalty but other regularization schemes could
be considered instead; see also
Example~\ref{example2.5} below.  
%



Using the generic representation given in Lemma~\ref{quadraticlemma},
for an exponential family, the proposed method is based on minimizing
the objective 
\begin{align}\label{regquadratic}
\hat{J}^\lambda(\theta) =
\frac{1}{2}~\theta^T\boldsymbol{\Gamma}(\mathbf{x})\theta  -
\g(\mathbf{x})^T\theta + c(\mathbf{x}) + \lambda\|\theta\|_{1},\quad
\theta\in\mathbb{R}^s,
\end{align}
where $\boldsymbol{\Gamma}(\mathbf{x})$ is positive semidefinite and
$\lambda \ge 0$ is a tuning parameter that controls the sparsity
level.  Larger values of $\lambda$ yield sparser solutions, and
$\lambda = 0$ gives the unregularized SME.  Since
$\boldsymbol{\Gamma}(\mathbf{x})$ is positive semidefinite, the
function $\hat{J}^\lambda(\theta)$ is convex but in the settings of
interest here $\boldsymbol{\Gamma}(\mathbf{x})$ will be singular and
$\hat{J}^\lambda(\theta)$ will not be strictly convex.

The regularized score matching objective from~(\ref{regquadratic}) is
similar to the lasso objective in linear regression
\citep{tibshirani1996}, where the function to be
minimized takes the special form
\begin{equation}\label{lassoreg}
\frac{1}{2}\|y - X\theta\|_2^2 + \|\theta\|_1, 
\end{equation}
for a `response vector' $y$ and a `design matrix' $X$.  In the
applications we have in mind~(\ref{regquadratic}) cannot be written
exactly as in (\ref{lassoreg}) because the vector $g(\mathbf{x})$ is
generally not in the column span of $\boldsymbol{\Gamma}(\mathbf{x})$.
However, we may adapt existing optimization methods for lasso to solve
the regularized score matching problem.
Implementation details are given in Appendix
\ref{implementation}. 

If the considered exponential family is supported on
$\mathcal{X}=\mathbb{R}^\m$ and we use
the loss from~(\ref{eq:emp-sm-loss}), then we call the minimizer of
(\ref{regquadratic}) the regularized score matching estimator (rSME).
If $\mathcal{X}=\mathbb{R}_+^\m$ and we use the loss
from~(\ref{eq:emp-sm-loss+}), then we abbreviate to $\mbox{rSME}_+$.
In specific instances of graphical models, we may apply the $\ell_1$
penalty only to those coordinates of $\theta$ whose vanishing
corresponds to absence of edges in a conditional independence graph.
If the subset $\mathcal{E} \subseteq \{1, \ldots , s \}$
holds the relevant coordinates then we use the penalty
\[
\|\theta\|_{1, \mathcal{E}} \equiv \sum_{j \in \mathcal{E}}
|\theta_j|.
\]

\begin{examcont}{example1}
  For the (centered) Gaussian case considered in Example
  \ref{example1}, the target of estimation is the symmetric precision
  matrix $\mathbf{K}$.  The conditional independence graph corresponds
  to the pattern of zeros in the off-diagonal entries of $\mathbf{K}$
  and the rSME is
  \begin{equation}\label{eq: regscore} 
    \mycheck{\mathbf{K}} = \mbox{arg}\ \underset{\mathbf{K} \in
      \text{Sym}_m}{\mbox{min}} ~ \bigg\{
    -\mbox{tr}(\mathbf{K}) + \frac{1}{2} \mbox{tr}(\mathbf{KKW}) +
    \lambda  \|\mathbf{K}\|_{1, \mbox{\tiny{off}}} \bigg\},
  \end{equation}
  where $W$ is the empirical covariance matrix and
  $\|\mathbf{K}\|_{1, \mbox{\tiny{off}}} =\|\mathbf{K}\|_{1, \mathcal{E}} $
  penalizes only the off-diagonal entries indexed by $\mathcal{E} =
  \{(j, k) : j\neq k\}$.  We emphasize that while in this example the
  natural parameter space is the positive definite cone, we propose
  minimizing simply over the entire space of symmetric $\m \times \m$
  matrices, denoted by $\text{Sym}_m$.  As our interest is primarily in
  graph selection, we do not enforce positive definiteness of
  $\mycheck{\mathbf{K}}$, which is in line with methods such as SPACE
  or neighborhood selection; compare \cite{KhareEtAl2013}.  

  We remark that evaluating the function from~(\ref{eq: regscore}) at
  a nonsymmetric matrix $\mathbf{K}$ as well as its transpose $\mathbf{K}^T$
  gives the same value.  By convexity, minimizing over all
  $\m\times\m$ matrices gives a solution in $\text{Sym}_m$, which then
  must equal $\mycheck{\mathbf{K}}$.
  %
\end{examcont}

\begin{examcont}{example2}
  In the truncated normal family from Example~\ref{example2}, the
  conditional independence graph corresponds again to the zero pattern
  in the off-diagonal entries of the positive definite interaction
  matrix $\mathbf{K}$.  Proceeding in analogy to the Gaussian case, we
  define the $\mbox{rSME}_{+}$ as the minimizer
  $\mycheck{\mathbf{K}}_+$ of the objective given by
  (\ref{nngaussiansme}) with the penalty
  $\lambda \|\mathbf{K}\|_{1, \mbox{\tiny{off}}}$ added on.  Again, we
  ignore the positive definiteness requirement and minimize the
  penalized non-negative score matching loss with respect to
  $\mathbf{K} \in \text{Sym}_m$.
\end{examcont}

\begin{examcont}{example2.5}
  For the family of distributions with normal conditionals from
  Example \ref{example2.5}, we would like a penalty to induce joint
  sparsity in the two symmetric interaction matrices $\mathbf{B}$ and
  $\mathbf{B}^{(2)}$, because an edge between nodes $j$ and $k$ is
  absent from the conditional independence graph if and only both
  $\mathbf{B}$ and $\mathbf{B}^{(2)}$ have their $(j,k)$ entries zero.
  For this purpose, it is natural to adopt the group lasso penalty
  \citep{Yuan2006}.  The rSME is then obtained by minimizing the
  empirical score matching loss augmented by the penalty
  \[
  \lambda \sum_{j\not= k} \sqrt{(\beta_{jk})^2+(\beta^{(2)}_{jk})^2}.
  \]
  Ignoring again any refined constraints from the natural parameter
  space of the family, we propose minimizing the penalized loss with
  respect to $\mathbf{b}\in\mathbb{R}^\m$ and
  $\mathbf{B},\mathbf{B}^{(2)}\in\text{Sym}_m$.   Since
  the group lasso is applied with small groups (of size
  2), the problem would be suitable for application of exact
  block-coordinate descent as discussed in \cite{foygel:grouplasso}.
\end{examcont}


\subsection{Uniqueness of rSME}
\label{sec:unique}


  
In the setup from Lemma~\ref{quadraticlemma}, we may write
\begin{equation}
\boldsymbol{\Gamma}(\mathbf{x}) =
\mathbf{H}(\mathbf{x})^T\mathbf{H}(\mathbf{x})
\end{equation}
for an $nm\times s$ matrix $\boldsymbol{H}(\mathbf{x})$;
recall~(\ref{gammaform}) and~(\ref{gammaform+}).  Based on the
arguments leading to Lemmas 3 and 5 in \citet{ryantibshirani2013}, the
function $\hat{J}^\lambda(\theta)$
from~(\ref{regquadratic}) has a unique minimizer $\mycheck{\theta}$
as long as $\lambda>0$ and the columns of $\boldsymbol{H}(\mathbf{x})$
are in \emph{general position}.  To clarify,
suppose that $\mathcal{U}\subset\mathbb{R}^{nm}$ is a collection of
$|\mathcal{U}|=s$ vectors.  Then $\mathcal{U}$ is in general position
if for all $k<\min\{nm,s\}$, all choices of vectors
$u_1,\dots,u_{k+1}\in\mathcal{U}$ and signs
$\sigma_1,\dots,\sigma_{k+1}\in\{-1,1\}$, the affine span of
$\sigma_1u_1,\dots,\sigma_{k+1}u_{k+1}$ does not contain any vector
$u$ or $-u$ for $u\in\mathcal{U}\setminus\{u_1,\dots,u_{k+1}\}$.

The graphical models we are interested in are pairwise interaction
models that have additional special structure in that the matrix
$\boldsymbol{\Gamma}(\mathbf{x})$ is block-diagonal with $m$ blocks of
equal size; recall Lemma~\ref{blockform} and
Remark~\ref{rem:blockform-gen}.  Denote the diagonal blocks by
$\boldsymbol{\Gamma}_1(\mathbf{x}),\dots,\boldsymbol{\Gamma}_m(\mathbf{x})$,
which in the setup from~(\ref{expfam:pairwise:general}) are of size
$(A\m^2+L\m) \times (A\m^2+L\m)$.  Each block is the sum of $n$
symmetric rank one matrices and we have the decomposition
\begin{equation}
  \label{eq:Hj}
\boldsymbol{\Gamma}_j(\mathbf{x}) =
\mathbf{H}_j(\mathbf{x})^T\mathbf{H}_j(\mathbf{x}), \quad j=1,\dots,m.
\end{equation}
The $n$ columns of each of the matrices $\mathbf{H}_j(\mathbf{x})$
were specified in~(\ref{eq:barh}).  It now holds that the regularized
score matching problem from~(\ref{regquadratic}) has a unique
minimizer provided each one of the $n \times (A\m+L)$ blocks
$\mathbf{H}_1(\mathbf{x}),\dots,\mathbf{H}_m(\mathbf{x})$ defined
in~(\ref{eq:Hj}) has its columns in general position.

\begin{examcont}{example1}
  In the Gaussian case, 
  $\mathbf{H}_1(\mathbf{x})=\dots=\mathbf{H}_m(\mathbf{x})=\mathbf{x}$.
  By the Lemma in \cite{Okamoto:1973}, the set of matrices
  $\mathbf{x}$ that fail to be in general position has measure zero.
  The rSME $\mycheck{\mathbf{K}}$ is unique almost surely when data
  are generated from a continuous joint distribution.
\end{examcont}

\begin{examcont}{example2}
  In the truncated normal case, 
  $\mathbf{H}_j(\mathbf{x})$ is equal to the matrix obtained from
  $\mathbf{x}$ by multiplying each column element-wise with $x_j$, the
  $j$th column of $\mathbf{x}$.  The Lemma in \cite{Okamoto:1973}
  implies that the rSME$_+$ is unique almost surely.
\end{examcont}

For the normal conditionals model from Example~\ref{example2.5},
almost sure uniqueness  would have to be derived by appealing
to results on uniqueness of group lasso \citep{Roth:2008}.

\subsection{Piecewise linear paths}\label{piecewise}


The rSME depends on the regularization parameter $\lambda$.  In this
section we make this explicit and denote it by
$\mycheck{\theta}^\lambda$.  Adopting standard language, we refer to
the set of $\mycheck{\theta}^\lambda$ obtained by varying $\lambda$
as the \emph{solution path} and call this path \emph{piecewise linear}
if there exists
$0 = \lambda_0 < \lambda_1 < \ldots < \lambda_R = \infty$ and
$\xi_0, \ldots , \xi_{R-1} \in \mathbb{R}^\m$ such that
$\mycheck{\theta}^\lambda = \mycheck{\theta}^{\lambda_r} + (\lambda -
\lambda_r)\xi_r$
for $\lambda\in[\lambda_r, \lambda_{r+1}]$. Piecewise linear solution
paths have the appeal that the entire solution path can be found by
calculating the change points $\lambda_r$ and associated slopes
$\xi_r$.

The next lemma is a consequence of the quadratic
nature of the score matching objective for exponential families, and holds
for the lasso problem as well.

\begin{lemma}
  The solution path $\mycheck{\theta}^\lambda$ for the regularized
  score matching problem from~(\ref{regquadratic}) is piecewise linear.
\end{lemma} 
\begin{proof}
  An $s$-vector $z$ belongs to $\partial \|\theta\|_{1}$, the
  subdifferential of the $\ell_1$ norm, if
  \begin{equation}\label{subdiff}
    z_j =  \begin{cases} \mbox{sign}(\theta_j) & \mbox{if $\theta_j \neq 0$}, \\ \in [-1, 1] & \mbox{if $\theta_j = 0$}.   \end{cases} 
  \end{equation}
  The Karush-Kuhn-Tucker (KKT) conditions characterizing optimality in
  (\ref{regquadratic}) are
  \begin{equation}\label{generalkkt}
    \boldsymbol{\Gamma}(\mathbf{x})\mycheck{\theta} - \g(\mathbf{x}) +
    \lambda \mycheck{z} = 0,  \quad \mycheck{z} \in \partial
    \|\mycheck{\theta}\|_{1}. 
  \end{equation}
  The linear relationship between $\mycheck{\theta}$ and $\lambda$ (for ``fixed'' $\mycheck{z}$) implies the claim.
\end{proof}

While straightforward to show, the property of piecewise linear paths
is special to the score matching method we propose.  Other methods
that give symmetric estimates of precision matrices in Gaussian
graphical models, such as glasso or the SPACE-type methods discussed
in \cite{KhareEtAl2013} do not have piecewise linear solution paths.
This said, piecewise linear paths also arise in neighborhood selection
\citep{meinshausen2006}, which, however, is a formulation without
symmetry.  Note also that when using a group lasso penalty as
suggested for Example~\ref{example2.5}, rSME solution paths are no
longer piecewise linear.

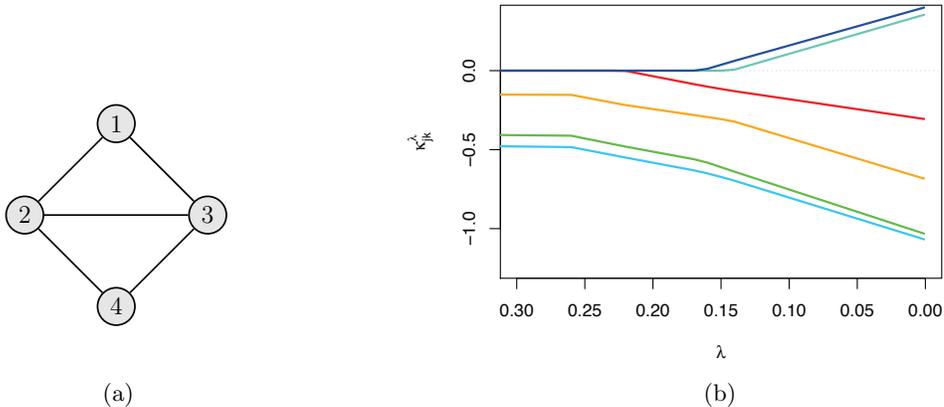
\begin{figure}[t]
  \begin{subfigure}[b]{0.22\textwidth}
    \scalebox{0.8}{
      \begin{tikzpicture}[baseline={(4.base)},circle,auto,thick, main
        node/.style={fill=gray!20,draw,font=\sffamily\large}]
        
        \node[main node] (1) {$1$}; \node[main node] (2) [below left =
        1.5 cm of 1] {$2$}; \node[main node] (3) [below right=1.5cm of
        1] {$3$}; \node[main node] (4) [below left = 1.5 cm of 3]
        {$4$};
        
        \path[every node/.style={font=\sffamily\small}] (1) edge node
        {} (2) (1) edge node {} (3) (2) edge node {} (4) (3) edge node
        {} (4) (2) edge node {} (3);
      \end{tikzpicture}
    } 
    \vskip0.2in
    \caption{}
    \label{graph}
  \end{subfigure}\hspace{2cm}
  \begin{subfigure}[b]{0.6\textwidth}
    \vskip 0.2in
    {\includegraphics[scale = 0.55]{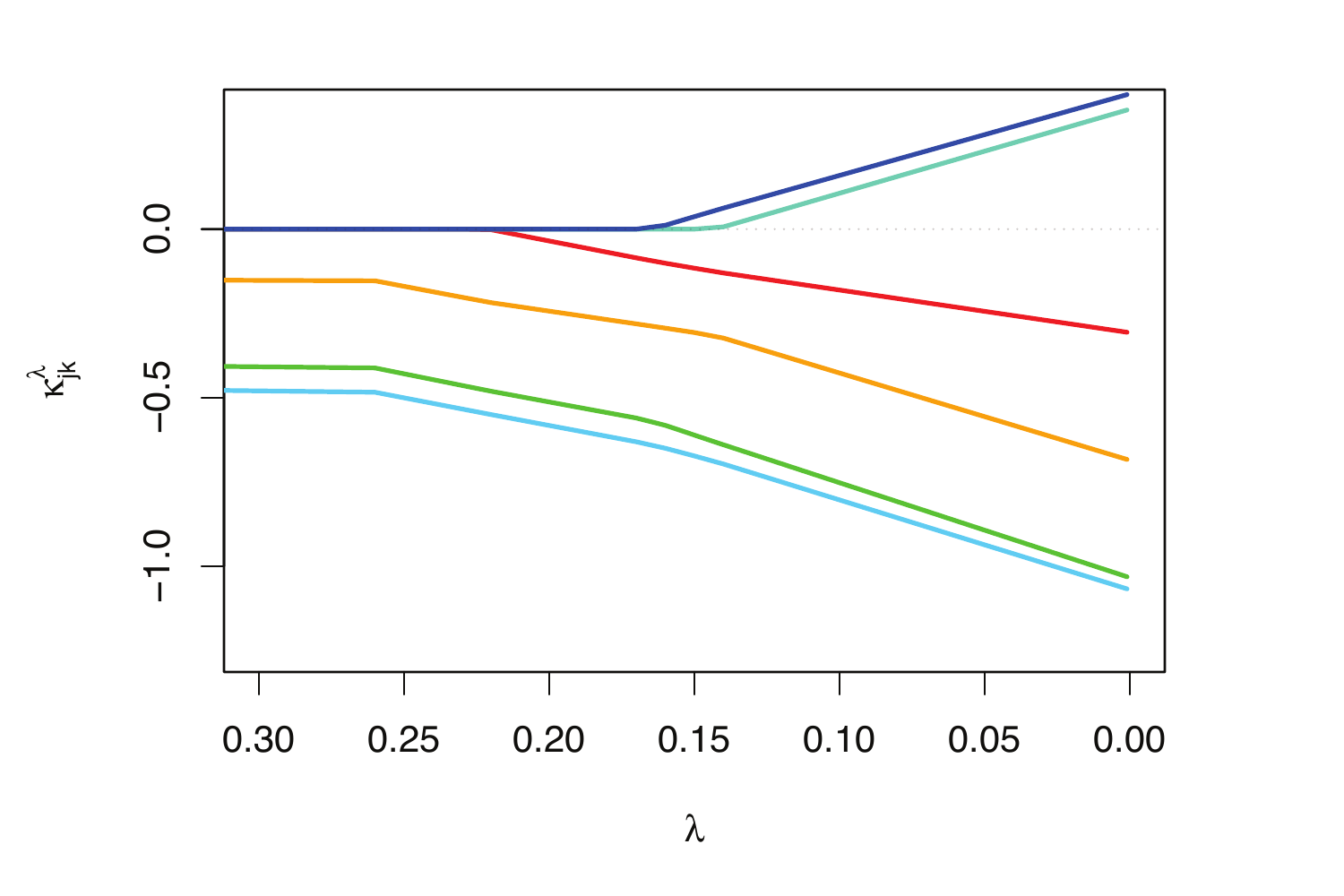}}
    \vskip -0.1in
    \caption{ }
\label{linear}
\end{subfigure} 
\caption{(a) A conditional independence graph with $\m=4$ nodes.  (b) rSME solution path for
  Gaussian graphical modeling ($\m = 4$, $n = 12$).}
\end{figure}

\begin{examcont}{example1}
  In the Gaussian model, the KKT conditions state that
  $\mycheck{\mathbf{K}}$ is a solution to (\ref{regquadratic}) if and
  only if
  \begin{equation}\label{gaussiankkt}
    \left(\mathbf{I}_{\m \times \m} \otimes \mathbf{W}\right)\mbox{vec}(\mycheck{\mathbf{K}}) - \mbox{vec}\left(\mathbf{I}_{\m \times \m}\right) + \lambda \mycheck{z} = 0 
  \end{equation}
  for $\mycheck{z} \in \partial
  \|\mycheck{\mathbf{K}}\|_{1,\text{off}}$, which in slight abuse of
  notation, we take to mean that
  \begin{equation}\label{subdiffoff}
    \mycheck{z}_{jk} =  \begin{cases} 0 & \mbox{if} \  j = k, \\
      \mbox{sign}(\mycheck{\kappa}_{jk}) & \mbox{if} \
        \mycheck{\kappa}_{jk} \neq 0 \ \mbox{and} \ j \neq k, \\ \in [-1, 1]
      & \mbox{if} \ \mycheck{\kappa}_{jk} = 0 \ \mbox{and} \ j\not=
        k.   \end{cases}  
  \end{equation}
  The first case accounts for the fact that the objective is smooth in
  the diagonal entries of the precision matrix, which are not
  penalized.
  Combining (\ref{gaussiankkt}) and (\ref{subdiffoff}), we have that
  \begin{align}
    \label{diag} 
    -1 + \sum_{k = 1}^\m w_{jk}\mycheck{\kappa}_{jk} &= 0, \quad j =
    1, \ldots , \m, \\
    \label{offdiag} 
    \sum_{\ell = 1}^\m w_{j\ell}\mycheck{\kappa}_{\ell k} + \sum_{\ell
      = 1}^\m w_{k\ell}\mycheck{\kappa}_{\ell j} + \lambda
    \mycheck{z}_{jk} &= 0, \quad  1\le j \not= k\le \m. 
  \end{align}
  A Gaussian solution path is shown in Figure~\ref{linear}, with the
  horizontal axis transformed to
  $t(\lambda) = \sum_{j \neq k} |\mycheck{\kappa}^\lambda_{jk}|$.  The
  data were drawn from a multivariate normal distribution with the
  conditional independence graph from Figure~\ref{graph}, with sample
  size $n = 12$.  We note that, as one would hope, the coefficient
  that last enters the solution corresponds to the absent edge
  $(1,4)$.
\end{examcont}

\subsection{Tuning}


A number of methods have been proposed for selecting the
regularization parameter $\lambda$ in $\ell_1$ penalization methods
and can be applied in our context.  On the one hand, a
predictive assessment as in cross-validation can be
considered, but 
the
selected graphs are typically too dense.  Other possibilities include
generalized cross validation (GCV) \citep{tibshirani1996}, Akaike's
Information Criterion (AIC), approaches based on stability under
resampling \citep{meinshausen2010,Samworth2013,Liu2010}, the Bayesian
Information Criterion (BIC) \citep{Schwarz1978} as well as extensions
of BIC proposed to cope with large model spaces \citep{Chen:2008,
  Gao2012,Foygel:2010,foygel:2011}.  The latter come with some consistency
guarantees.

As a demonstration, for the Gaussian case from Example \ref{example1},
we may consider an extended BIC criterion based on the basic
score matching loss (\ref{scorematch}), defined as
\begin{equation}\label{EBIC} 
\text{BIC}(\lambda) =  -2\mbox{tr}(\mathbf{\mycheck{K}}^\lambda) +  \mbox{tr}(\mathbf{\mycheck{K}}^\lambda \mathbf{\mycheck{K}}^\lambda\mathbf{W}) + |\mycheck{E}^\lambda|\log n + 4|\mycheck{E}^\lambda| \gamma \log \m, 
\end{equation}
where $\mycheck{E}^\lambda = \{(j, k): \mycheck{\kappa}^\lambda_{jk}
\neq 0, j < k \}$ and $\gamma$ is typically taken to be $1/2$ or $1$.
Alternatively, we could refit, that is, replace $\mathbf{K}^\lambda$
by an unregularized SME computed in the submodel given by constraining
all $\kappa_{jk}$ with $(j,k)\not\in\mycheck{E}^\lambda$ to be zero.  
In either case, we choose $\lambda$ to minimize
(\ref{EBIC}).

\section{Numerical Experiments}\label{simulations}

We perform numerical experiments comparing regularized score matching
to existing methods when data is simulated from (i) a multivariate
normal distribution, (ii) a multivariate truncated normal
distribution, and (iii) a distribution with normal conditionals.
The comparison is made against three methods for estimation of
Gaussian graphical models, namely, glasso, neighborhood selection
(both implemented in the \verb|R| packages \verb|huge|) and SPACE (in
its CONCORD formulation, with \verb|R| package \verb|gconcord|).
In addition, we
consider the \textit{nonparanormal SKEPTIC}, which applies glasso to a
matrix of rank correlations (Kendall's $\tau$ or Spearman's $\rho$)
and can be motivated by a Gaussian copula model \citep{liu2012}.  We utilize the version based on Kendall's $\tau$.  Finally, we compare to
\textit{SPACEJAM} \citep{VoormanEtAl2014_Biometrika}, which is based
on additive modeling of conditional means and implemented in the
\verb|R| package \verb|spacejam|. We conclude this section with brief investigations on the robustness of regularized score matching when data is not generated under the assumed model. All results in this section are based on averaging over 100
independently generated
datasets.  

\subsection{Gaussian data}\label{gaussiancase} 

We consider a graph with $\m = 1000$ nodes, composed of $10$ connected
components, each $100$ nodes in size and structured as a
$10 \times 10$ 2-D lattice (4 nearest neighbors).  Each connected
component also features three hubs with node degree 20, randomly
selected from the subset of nodes in the component.

We follow a procedure similar to the one from
\citet{PengEtAl2009} to convert the adjacency matrix of the graph  into a sparse diagonally dominant partial
correlation matrix. For each non-zero element of the adjacency matrix,
we sample a draw from a uniform distribution on $[0.5, 1]$. Each row of this new matrix is then rescaled by 1.5 times the sum of the absolute values of the off-diagonal entries in the row. We average this matrix with its transpose to ensure symmetry, and set its diagonal elements to 1. This matrix is inverted and converted into a correlation matrix to form $\boldsymbol{\Sigma}^*$.

Data is then generated from a multivariate normal distribution with
mean zero and a covariance matrix $\boldsymbol{\Sigma}^*$.
We choose sample size $n=600$ and $1000$.  The
setup agrees with that in \citet{PengEtAl2009}, except that the number
of nodes has been scaled up.

Figure \ref{gauss} shows the ROC curves obtained under both sample sizes.  
Since the truth is Gaussian, we do not report results for SKEPTIC or SPACEJAM.
For both sample sizes, the curve for regularized score matching almost
perfectly aligns with those for neighborhood selection, SPACE, and glasso.
The results indicate that regularized score
matching estimators achieves state-of-the-art statistical efficiency
in Gaussian models.

\begin{figure}[t]
\captionsetup[subfigure]{justification=centering}
 \centering
        \begin{subfigure}[b]{0.45\textwidth}
                {\label{gauss1} \includegraphics[width=\textwidth]{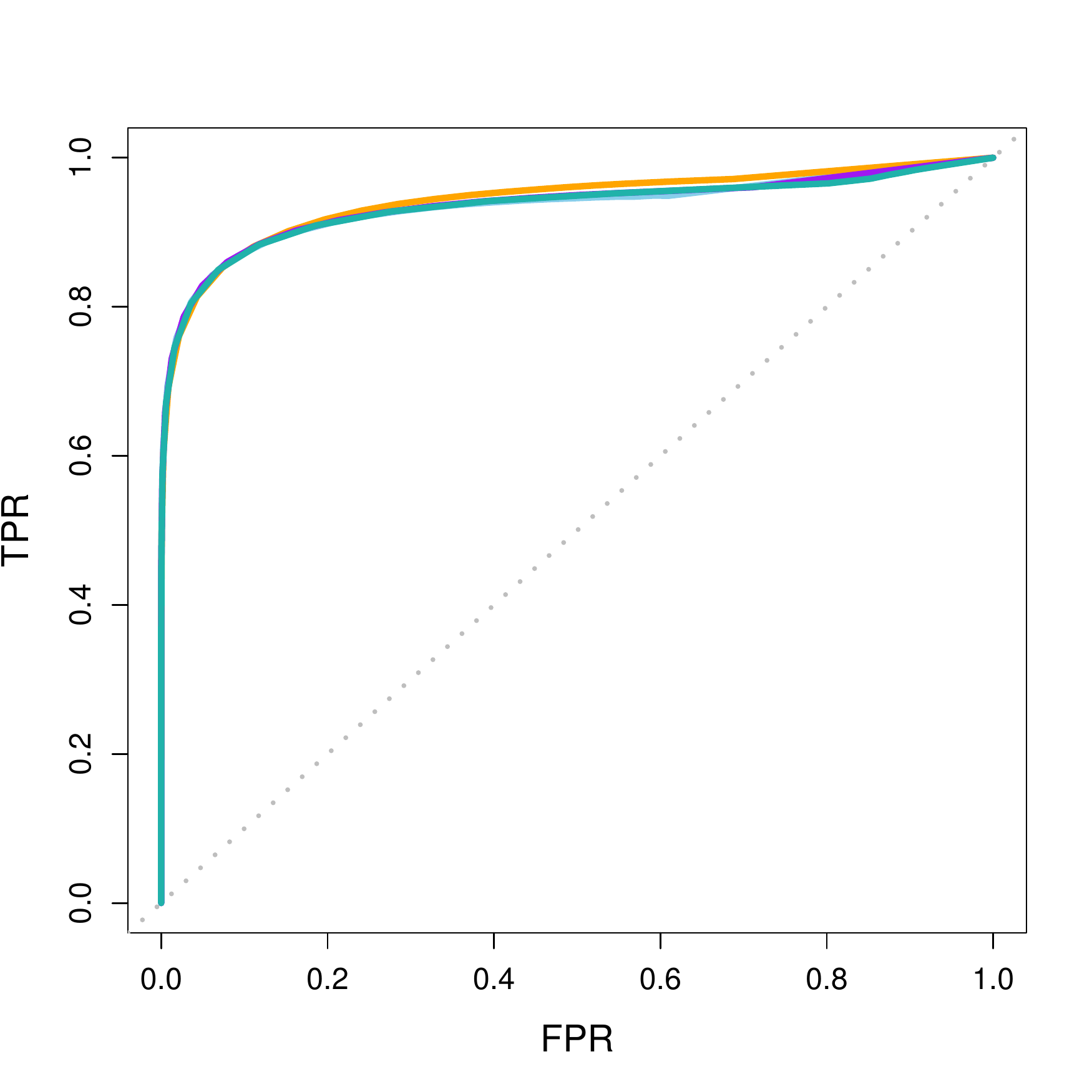}}
                \caption{$n = 600$}
        \end{subfigure}
        \begin{subfigure}[b]{0.45\textwidth}
                {\label{gauss2} \includegraphics[width=\textwidth]{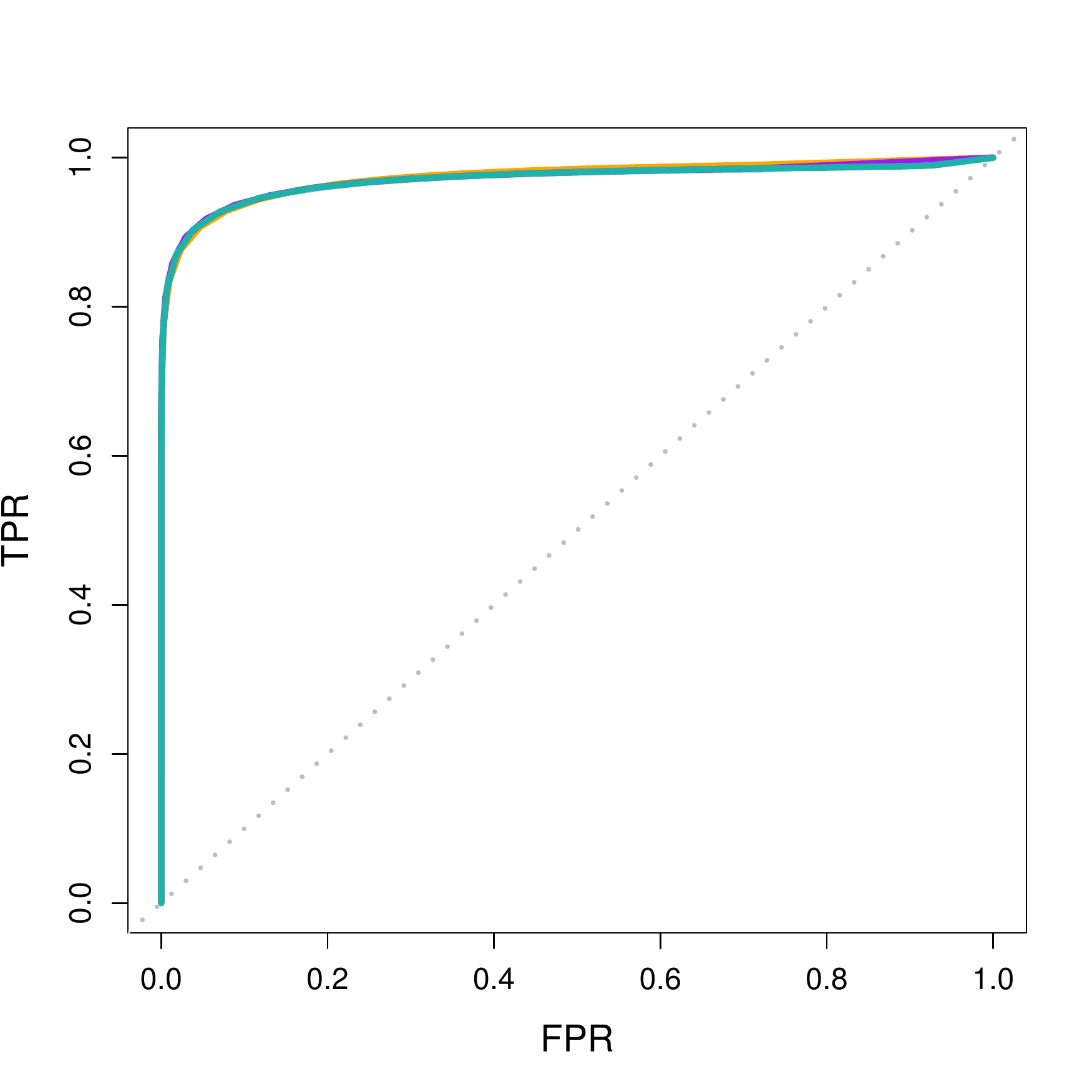}}
                \caption{$n = 1000$}
        \end{subfigure}
        \caption{ROC curves for the Gaussian case. The dashed grey line represents random selection of edges. The color to method correspondence is as follows: regularized score matching (\ref{scorecolour}), neighborhood selection (\ref{mbcolour}), glasso (\ref{glassocolour}), and SPACE (\ref{spacecolour}). The curves are almost perfectly aligned.}
        \label{gauss}
\end{figure}



\subsection{Non-negative Gaussian data} 
\label{sec:non-negat-gauss}

Glasso, SPACE, neighborhood selection and SKEPTIC all presume some
form of underlying Gaussianity.  In this and the next subsection, we
demonstrate the application of regularized score matching in scenarios
where these assumptions do not hold to highlight the versatility of
the proposed appraoch.

Similar to the Gaussian setting, we consider a graph with $\m = 100$ nodes, composed of 10 disconnected subgraphs with equal number of nodes. Using the lower triangular elements adjacency matrix of each $10$ node subgraph, we construct ten matrices, where in each matrix, the element is drawn independently to be 0 with probability 0.2, and from a uniform distribution on $[0.5, 1]$ with probability 0.8. 
The matrices, after symmetrization, are combined into a $100 \times 100$ block matrix. The diagonal elements are set to a common positive number such that the minimum eigenvalue is 0.1 to form the precision matrix of the pre-truncated normal, $\mathbf{K}^*$. 

Data was then generated from a truncated centered multivariate
normal, left-truncated at $0$ and with
$\boldsymbol{\Sigma}^* = (\mathbf{K}^*)^{-1}$ as normal
covariance.  We used the Gibbs sampler from the \verb|tmvtnorm|
package in \verb|R| with a burnin period of 100 samples.  We thinned
out the remaining samples, keeping one in ten.  The sample size $n$ is
taken to be either $2500$ or
$5000$. 
The need for a larger sample size is explained by our theoretical
findings in Section \ref{theory}, specifically Corollary \ref{th2}.

The ROC curves are shown in Figure~\ref{nngauss}, where regularized score matching outperforms all competitors considered.  The closest competitor to regularized score matching are SKEPTIC and SPACEJAM, both of which, objectively, perform well, being capable of capturing some of the non-Gaussianity in the data. %

\begin{figure}[t]
\captionsetup[subfigure]{justification=centering}
 \centering
        \begin{subfigure}[b]{0.45\textwidth}
                {\label{gauss1a} \includegraphics[width=\textwidth]{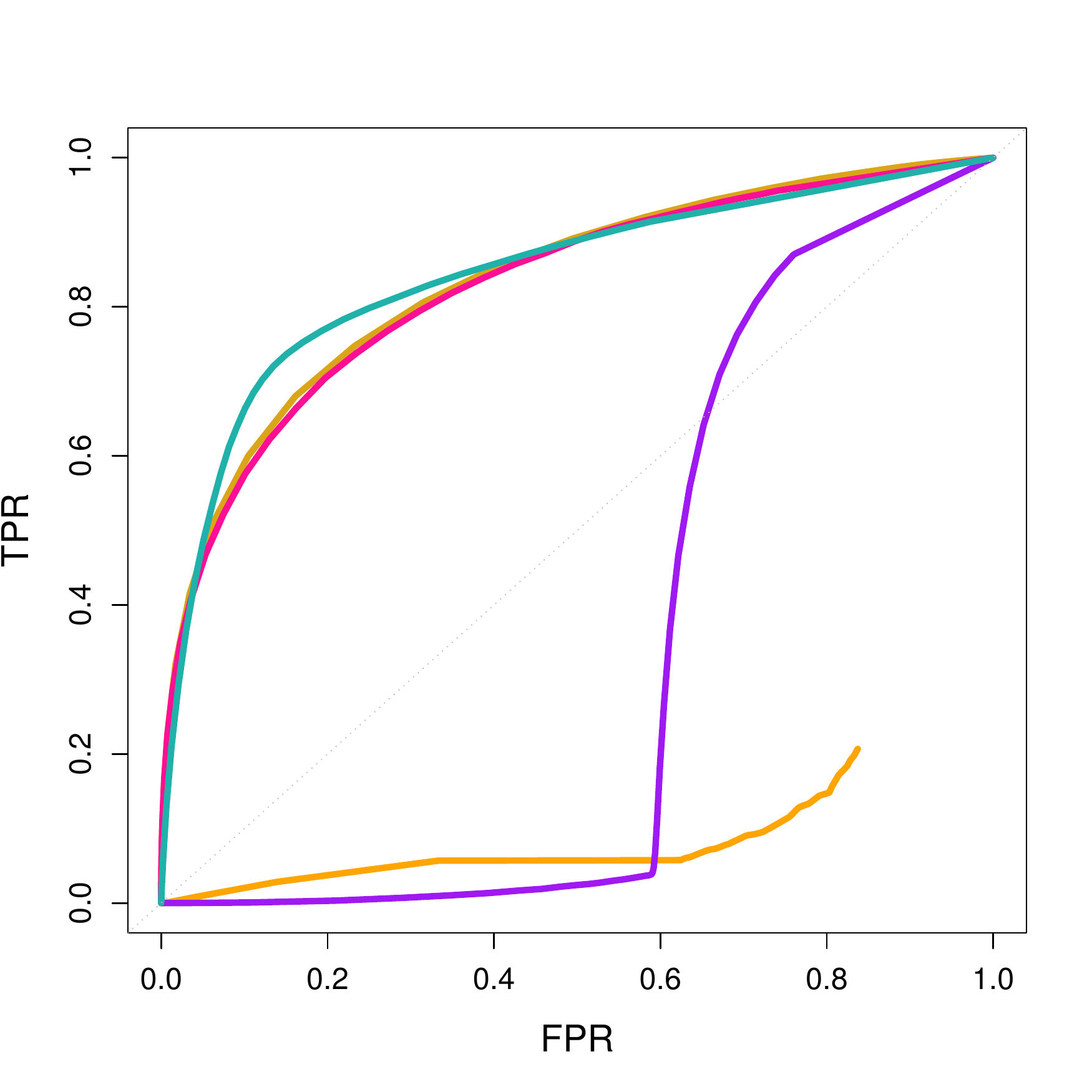}}
                \caption{$n = 2500$}
        \end{subfigure}
        \begin{subfigure}[b]{0.45\textwidth}
                {\label{gauss2a} \includegraphics[width=\textwidth]{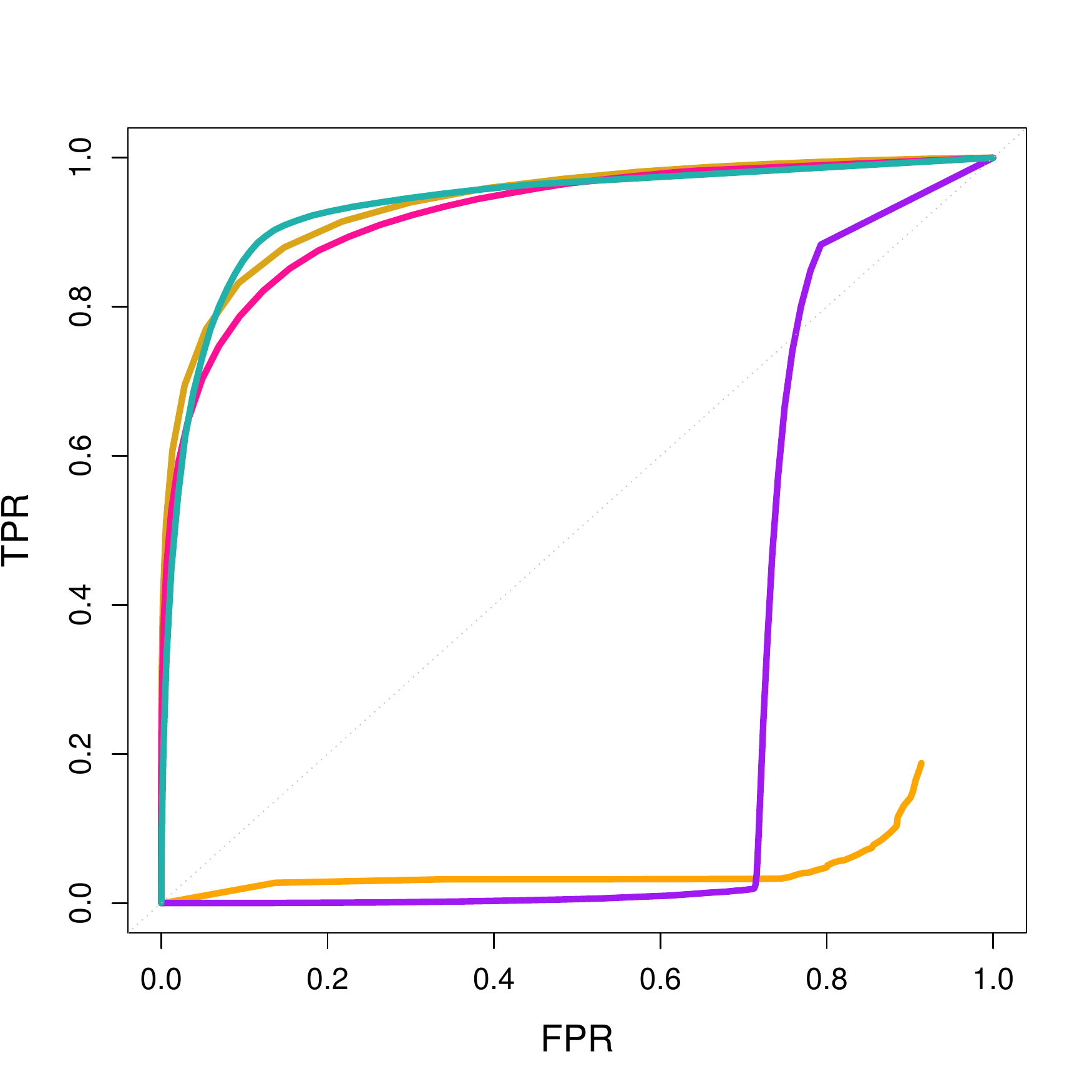}}
                \caption{$n = 5000$}
        \end{subfigure}
        \caption{ROC curves for the non-negative Gaussian case. The dashed line represents random selection of edges.  The color to method correspondence is as follows: regularized score matching (\ref{scorecolour}), glasso (\ref{glassocolour}), SPACE (\ref{spacecolour}), SKEPTIC (\ref{skepticcolour}), and SPACEJAM (\ref{spacejamcolour}).}
        \label{nngauss}
\end{figure}

We emphasize that here score matching was applied in its non-negative
version from Section~\ref{sec:extens-non-negat}.  The basic score
matching procedure from Section~\ref{sec:basic-score-matching} is far less efficient based
on experiments not reported here.

\subsection{Normal conditionals}
\label{sec:normal-conditionals}

Next, we take the data-generating distribution to have a density from
the class
\begin{align}\label{condnormex}
q(x | \mathbf{B}, \mathbf{b}, \mathbf{b}^{(2)}) &\propto \exp\left\{ \sum_{j \neq k} \beta_{jk}x_j^2x_k^2 + \sum_{j=1}^\m \beta^{(2)}_jx_j^2 + \sum_{j=1}^\m \beta_jx_j  \right\}, & x \in \mathbb{R}^\m,
\end{align}
where $\mathbf{B} = \{ \beta_{jk} \}$  is a symmetric matrix
with diagonal entries $0$.  This family is a special case of the
distributions with normal conditionals from
Example~\ref{example2.5}.

We consider the case $\m=625$, with the graph being a $25 \times 25$
2-D lattice (4 nearest neighbors).  The true interaction matrix
$\mathbf{B}^*$ is constructed by multiplying the adjacency matrix by
$-1/25$.  The coefficients for the terms $x_j^2$ are all set equal to
$-1$ and those for the $x_j$ all equal to $8/50$, which makes the
marginal distributions deviate noticeably from Gaussianity.  Data can
be generated by Gibbs sampling using the Gaussian full conditionals.
We discard the first 100 samples and thin out the remaining
samples, keeping one in ten, as in Section~\ref{sec:non-negat-gauss}.
 
\begin{figure}[t]
\captionsetup[subfigure]{justification=centering}
 \centering
        \begin{subfigure}[b]{0.45\textwidth}
                {\label{conditional1} \includegraphics[width=\textwidth]{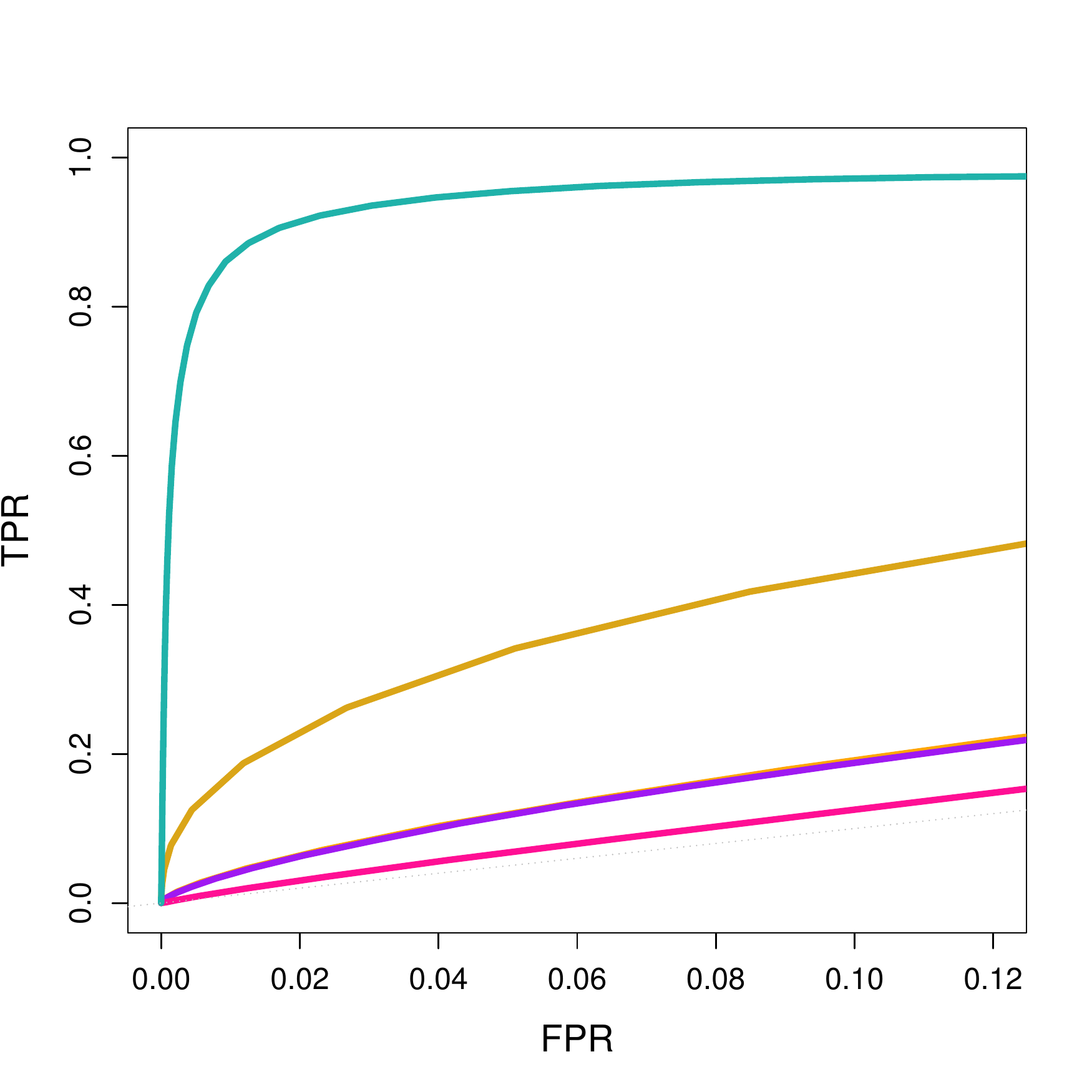}}
                \caption{$n = 750$}
        \end{subfigure}
        \begin{subfigure}[b]{0.45\textwidth}
                {\label{conditional2} \includegraphics[width=\textwidth]{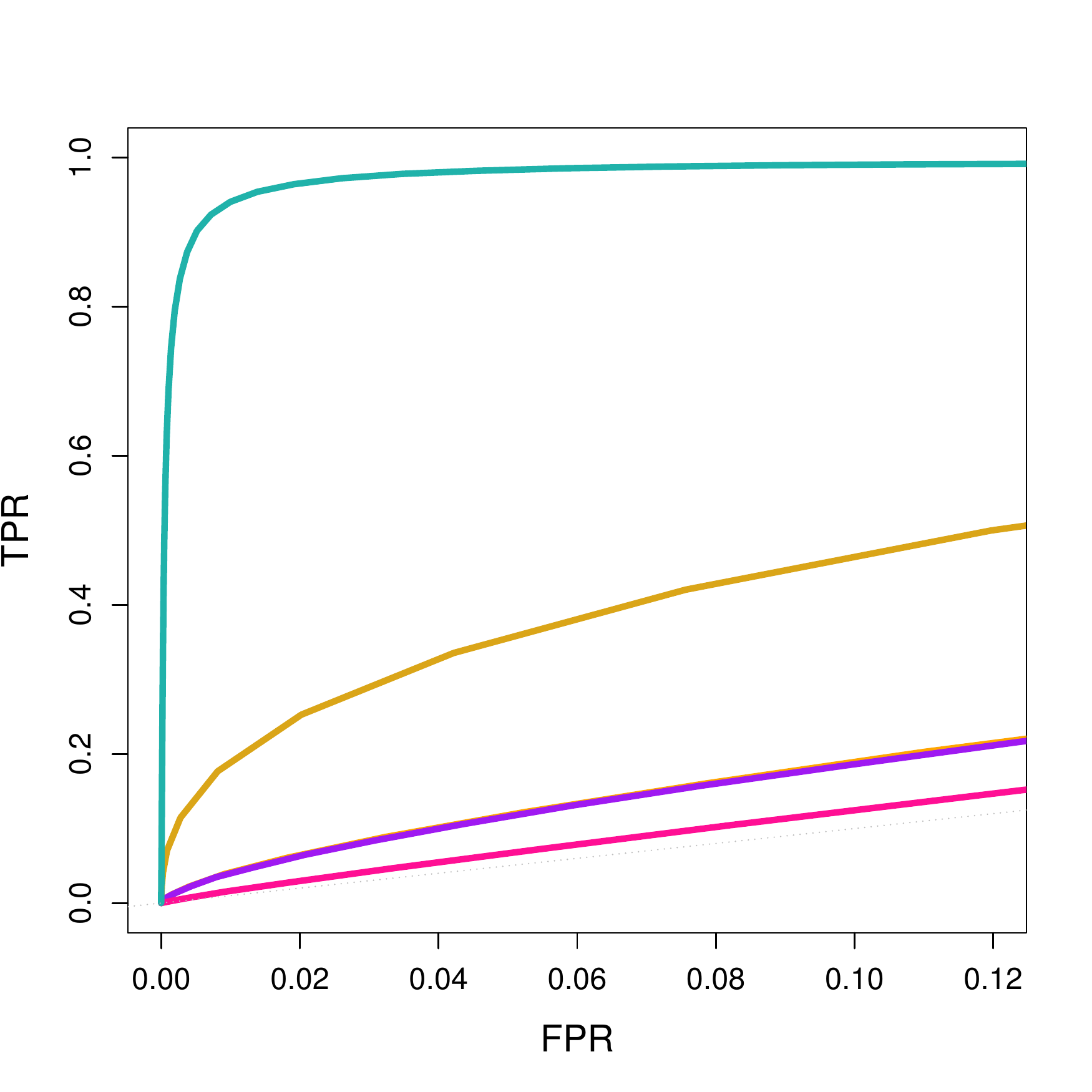}}
                \caption{$n = 1000$}
        \end{subfigure}
        \caption{ROC curves for the normal conditionals case. The dashed line represents random selection of edges. The color to method correspondence is as follows: regularized score matching (\ref{scorecolour}), glasso (\ref{glassocolour}), SPACE (\ref{spacecolour}), SKEPTIC (\ref{skepticcolour}), and SPACEJAM (\ref{spacejamcolour}). The curve for glasso overlaps with the curve for SPACE.}
        \label{conditional}
\end{figure}

We plot the ROC curves for conditional normal data in Figure
\ref{conditional}.  Regularized score matching outperforms
its competitors by a clear margin.  This is not surprising, as both glasso and SPACE are
derived under normality.  A Gaussian copula model as underlying
SKEPTIC is of little help.  SPACEJAM does best among the competitors
but cannot fully extract the available signal about the edge structure
as the conditional means are non-additive and the conditional
variances are not constant.

\subsection{A robustness check}
It is of interest to see how score matching performs when the data-generating mechanism is misspecified. We consider two scenarios. First, we apply the Gaussian score matching to a contaminated Gaussian setting similar to that explored in \cite{finegold:drton:2011}. That is, a random subset of Gaussian observations is replaced with Gaussian noise. In the second example, we investigate the performance of the regularized Gaussian score matching when the observations are not Gaussian but rather drawn from a multivariate $t$-distribution. 

\subsubsection{Contaminated Gaussians}
We mimic the setup used in the numerical experiments in \cite{finegold:drton:2011}, who consider these settings to test the robustness of their \textit{tlasso}. Fixing $m = 200$, we construct a sparse precision matrix $\mathbf K^*$ according to the following steps: (1) choose each (strictly) lower triangular element of $\mathbf K^*$ to be independently -1, 0, 1 with probability 0.01, 0.98 and 0.01 respectively, (2) symmetrize the matrix (3) for each row, i.e. for $j = 1, \ldots , m$, set $ \kappa^*_{jj} = 1 + \| \kappa_{j, -j}^*\|_0$ where $\kappa_{j, -j}^*$ refers to the $j$th row of $\mathbf K^*$ with the diagonal element in that row removed. To strengthen partial correlations, the diagonal elements are scaled down by a common positive factor such that the minimum eigenvalue of the resulting matrix is approximately 0.6 (close to 0.62 in our setup). The covariance matrix $\mathbf \Sigma^*$ is obtained by inverting $\mathbf K^*$.

We generate either $n = 150$ or $n = 200$ observations from a multivariate normal distribution with
mean zero and a covariance matrix $\boldsymbol{\Sigma}^*$. We then corrupt 2\% of the observations, substituting them with i.i.d.~ $N(0, 0.2)$ draws. The corrupted observations cannot easily be differentiated from normal observations, and this elevates the difficulty of the estimation problem.  

\begin{figure}[t]
\captionsetup[subfigure]{justification=centering}
 \centering
        \begin{subfigure}[b]{0.45\textwidth}
                {\label{contaminate1} \includegraphics[width=\textwidth]{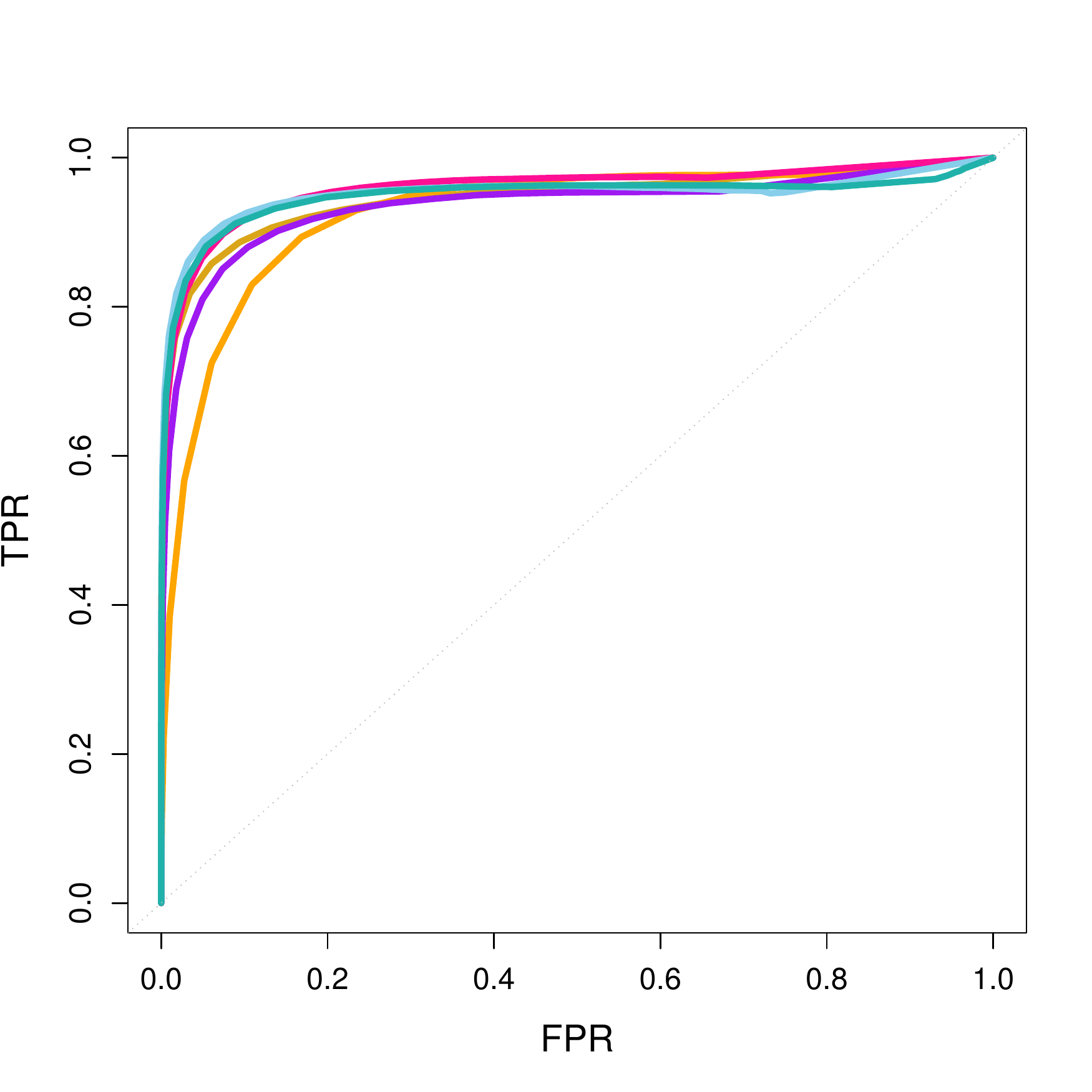}}
                \caption{$n = 150$}
        \end{subfigure}
        \begin{subfigure}[b]{0.45\textwidth}
                {\label{contaminate2} \includegraphics[width=\textwidth]{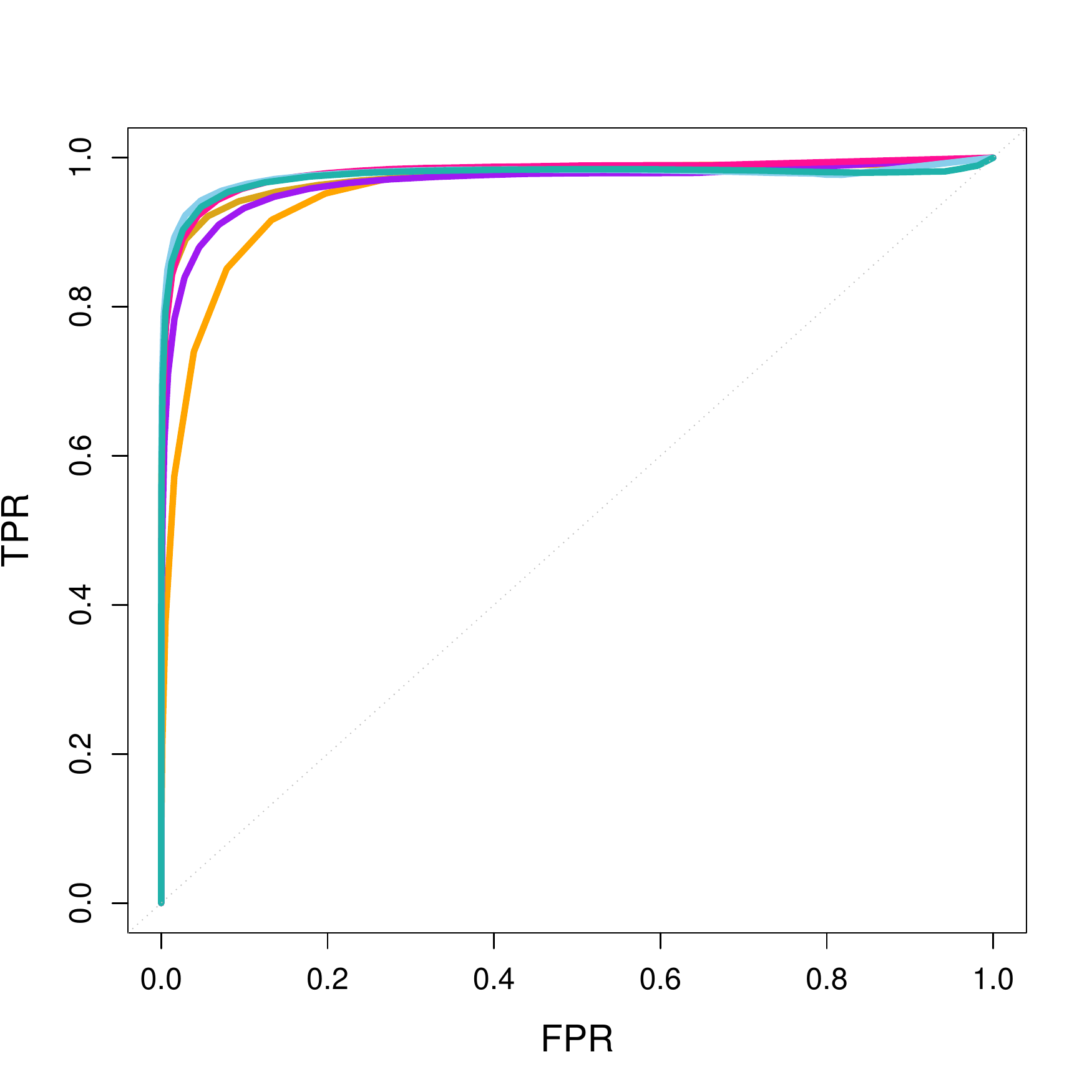}}
                \caption{$n = 200$}
        \end{subfigure}
        \caption{ROC curves for the contaminated Gaussian case. The dashed line represents random selection of edges. The color to method correspondence is as follows: regularized score matching (\ref{scorecolour}), neighborhood selection (\ref{mbcolour}), glasso (\ref{glassocolour}), SPACE (\ref{spacecolour}), SKEPTIC (\ref{skepticcolour}), and SPACEJAM (\ref{spacejamcolour}).}
        \label{contaminated}
\end{figure}

We present the ROC curves in Figure \ref{contaminated}. Interestingly, score matching performs reasonably well, on par with SKEPTIC and neighborhood selection. For both sample sizes, the differences, which are subtle, are most apparent in the regime where the number of false positives detected is small: score matching falls slightly short of neighborhood selection, but it also appears to slightly outperform SKEPTIC. Surprisingly, there is a clear margin of difference between the performances of regularized score matching and SPACE, the former outperforming the latter, despite their noted structural similarities. Glasso, which utilizes the full Gaussian likelihood, performs the worst. Overall, we conclude that regularized score matching is competitively robust when compared to its alternatives in the contaminated Gaussian setting. 

\subsubsection{Multivariate $t$-distributed observations}
In this section, we apply regularized Gaussian score matching to observations arising from a multivariate $t$-distribution with mean $0$ and covariance matrix $\mathbf \Sigma^*$. This corresponds to testing the robustness of regularized score matching under model misspecification. 
Like in the previous section, we consider the case when $m = 200$. To set up $\mathbf \Sigma^*$, we construct a $m \times m$ adjacency matrix based on an Erd\H{o}s-R\'enyi graph with the probability of drawing an edge between any two arbitrary nodes set to 0.01. We then convert the adjacency matrix into $\mathbf \Sigma^*$ using the same procedure as in Section \ref{gaussiancase}. Samples were drawn from a multivariate $t$-distribution with covariance matrix $\mathbf \Sigma^*$ and three degrees of freedom.

\begin{figure}[t]
\captionsetup[subfigure]{justification=centering}
 \centering
        \begin{subfigure}[b]{0.45\textwidth}
                {\label{tdist1} \includegraphics[width=\textwidth]{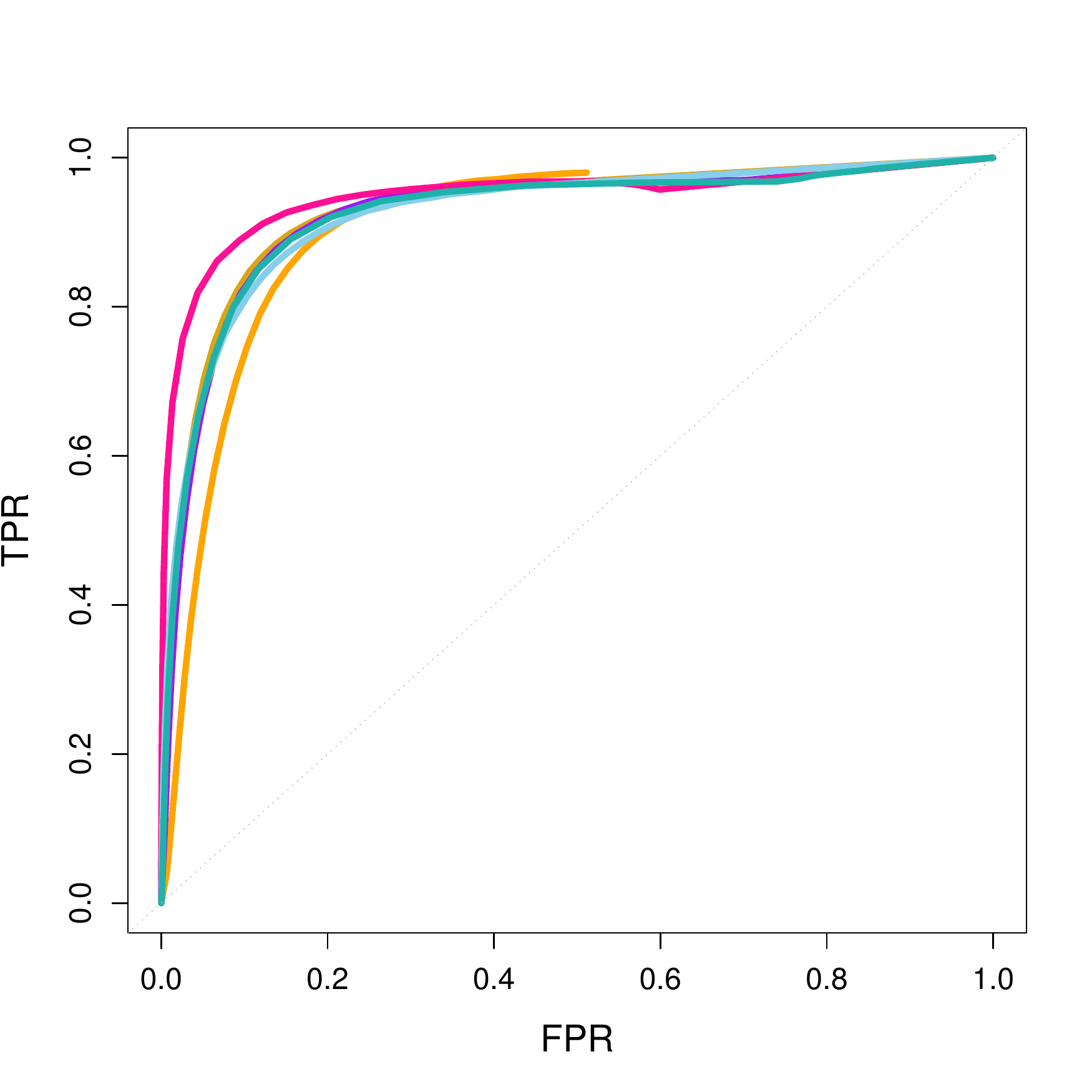}}
                \caption{$n = 100$}
        \end{subfigure}
        \begin{subfigure}[b]{0.45\textwidth}
                {\label{tdist2} \includegraphics[width=\textwidth]{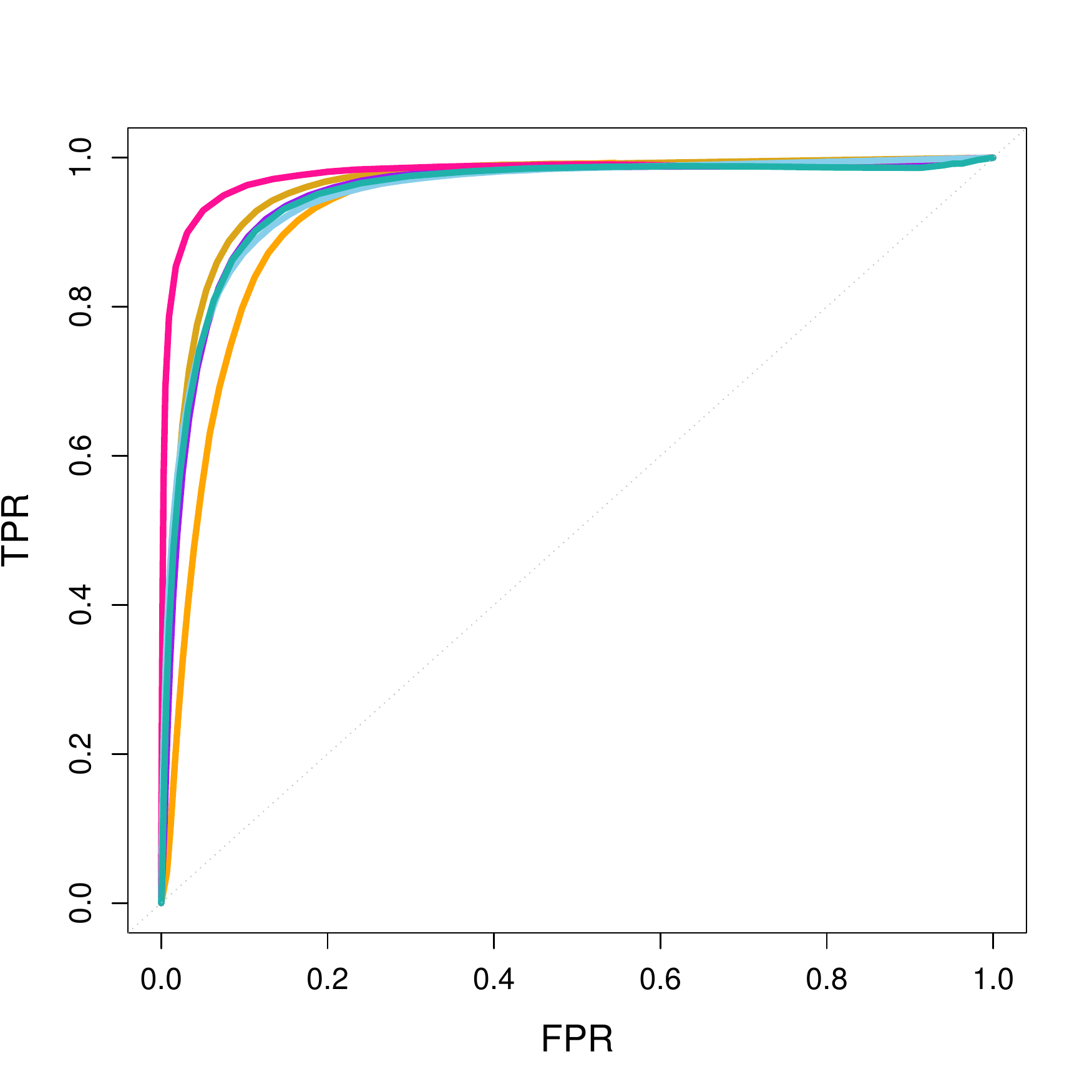}}
                \caption{$n = 150$}
        \end{subfigure}
        \caption{ROC curves for the $t$-distributed case. The dashed line represents random selection of edges. The color to method correspondence is as follows: regularized score matching (\ref{scorecolour}), neighborhood selection (\ref{mbcolour}), glasso (\ref{glassocolour}), SPACE (\ref{spacecolour}), SKEPTIC (\ref{skepticcolour}), and SPACEJAM (\ref{spacejamcolour}).}
        \label{tdist}
\end{figure}
The ROC curves are plotted in Figure \ref{tdist} for $n = 100$ and $n
= 150$. As expected, SKEPTIC outperforms all others, owing to its
flexibility to accommodate outliers, as previously demonstrated in
\cite{liu2012}. In fact, for elliptical distributions, such as the
multivariate $t$-distribution, Kendall's $\tau$ allows for consistent
estimation of $\mathbf \Sigma^*$, so SKEPTIC should perform optimally
\citep{liu2012transelliptical}. Nonetheless, regularized score
matching is reasonably robust under this setting: its performance is
comparable to that of SPACEJAM -- only falling slightly short --
SPACE, and neighborhood selection. Again, glasso yields the poorest
results.

%
%

\section{Application to RNAseq Data}\label{real}

The American Cancer Society estimates that in 2015 there will be
220,800 new cases of prostate cancer and 27,540 deaths.  To understand
how the cancer develops, as well as how it may be treated, it is
necessary to decipher the genetic machinery which drives it.  Since
cancer is such a complex disease, it is insufficient to study a single
gene at a time, as genes may interact with one another in many ways.
Graphical modeling of gene expression data has the potential to aid in
discovery of such interactions.

RNAseq data from next-generation sequencing technology can
be used to identify genes that are activated/transcribed or
suppressed at the time of
measurement.  
However, RNAseq data are non-negative and have skewed marginals, which
presents a challenge for existing methodologies.  Graphical models 
based on truncated Gaussian models are interesting alternatives to
existing approaches that primarily consist of applying Gaussian
methods after transformations.  Whether truncation models are truly useful scientifically deserves a fuller exploration; here we simply illustrate how different estimates can be obtained from the
proposed methodology.

Our case study is based on the RNAseq data from 487 prostate
adenocarcinoma samples available in The Cancer Genome Atlas 
dataset.  We focus on 350 genes that belong to ``known'' cancer
pathways in the Kyoto Encyclopedia of Genes and Genomes. 
Removing genes with more than 10\% missing values, we obtained a
dataset with $\m = 333$ genes.  Remaining missing values were simply
set to zero, adding to the challenge.  (We will comment on the issue
of missing data in the discussion.)  In illustration of the
regularized score matching methodology, we consider an exponential
family of truncated normal distributions with density
\begin{equation*}
q(x | \mu, \mathbf{K}) \;\propto\; \exp\left\{\frac{1}{2}(x -
  \mu)^T\mathbf{K}(x - \mu) \right\},  \quad x\in\mathbb{R}_+^\m.
\end{equation*}
This generalizes the family of distributions considered in
Example~\ref{example2} by allowing the truncated normal distribution to have nonzero mean.  

We compare regularized non-negative score matching, SPACE (using
CONCORD formulation), glasso, SKEPTIC and SPACEJAM.  We apply SPACE
and glasso directly to the standardized data.  We do not consider any
marginal transformations as they are naturally accounted for when
comparing to the rank correlation-based SKEPTIC.  For each method, we
tune the regularization parameter $\lambda$ in order to obtain
$|E| = 333$ (or $334$) edges. Figure~\ref{topology} depicts the estimated networks, with isolated nodes removed, in layouts optimized for each graph. To allow for easier comparison, we also show the estimated networks in fixed layouts in Figure~\ref{topology2}.
Node degree distributions are plotted in Figure~\ref{degdist}.

\begin{figure}[t]
        \centering
            \begin{subfigure}[b]{0.3\textwidth}
                \includegraphics[width=\textwidth]{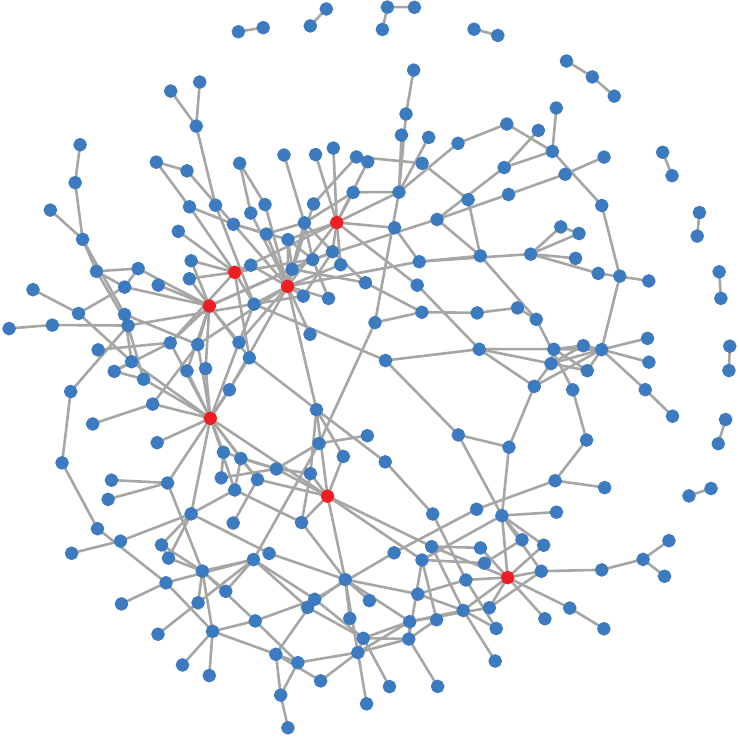}
                \caption{Reg.\ score matching}
                \label{fig:gull2}
                \vspace*{6mm}
        \end{subfigure}
        \hspace*{4mm}
        \begin{subfigure}[b]{0.3\textwidth}
                \includegraphics[width=\textwidth]{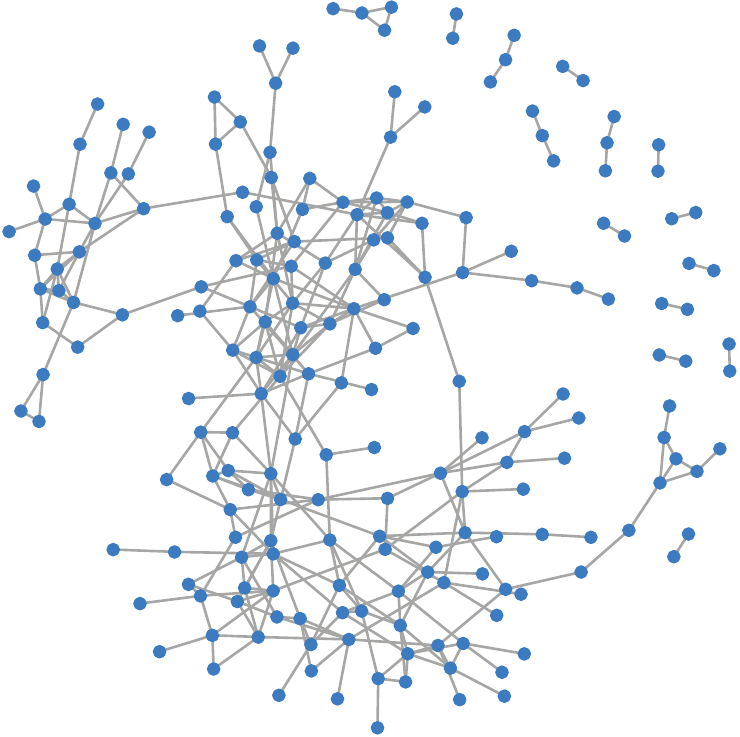}
                \caption{SPACE}
                \label{fig:tiger2}
                \vspace*{6mm}
        \end{subfigure}
        \hspace*{4mm}
        \begin{subfigure}[b]{0.3\textwidth}
                \includegraphics[width=\textwidth]{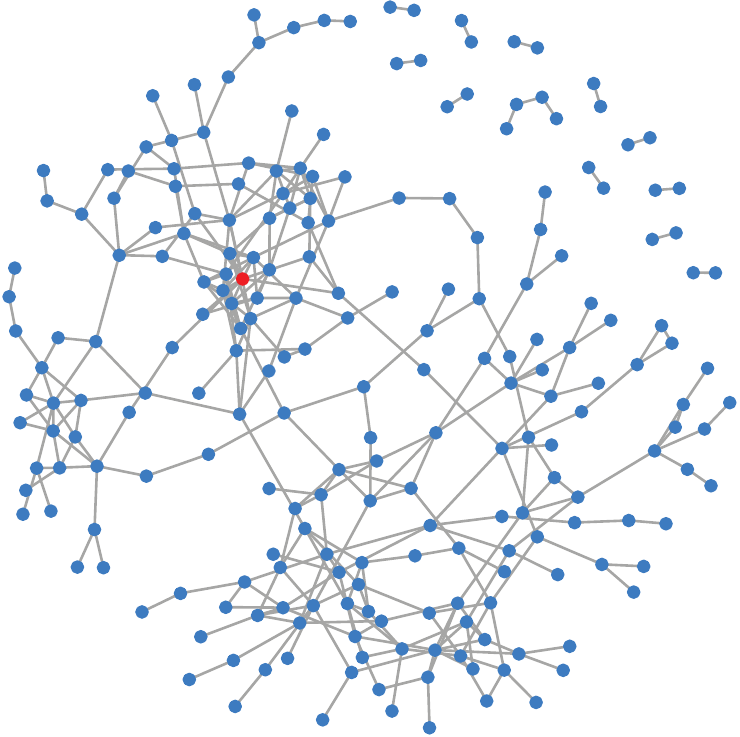}
                \caption{SPACEJAM}
                \label{fig:mouse}
                \vspace*{6mm}
        \end{subfigure}
        \hfill
        \begin{subfigure}[b]{0.3\textwidth}
                \includegraphics[width=\textwidth]{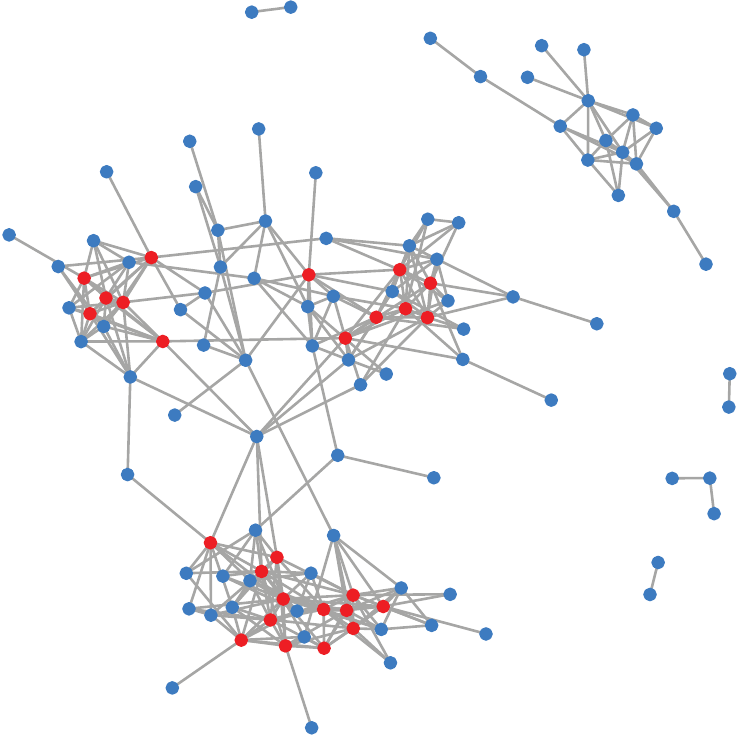}
                \caption{Glasso}
                \label{fig:mousea}
                \vspace*{3mm}
        \end{subfigure}
        \hspace{10mm}
        \begin{subfigure}[b]{0.3\textwidth}
                \includegraphics[width=\textwidth]{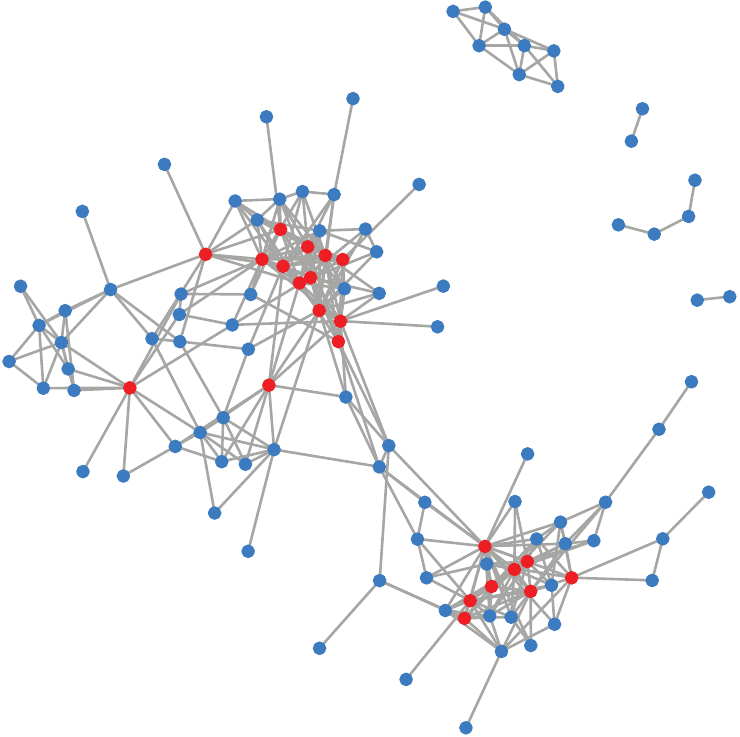}
                \caption{SKEPTIC}
                \label{fig:mouseb}
                \vspace*{3mm}
        \end{subfigure}        \caption{Topology of inferred networks of $|E| = 333$ or $334$
          edges for all considered methods. The layout has been optimized for each graph. Isolated nodes are
          not shown.  Red colored nodes have degree greater or equal to 10.}\label{topology}
\end{figure}

\begin{figure}[t]
        \centering
        \begin{subfigure}[b]{0.3\textwidth}
                \includegraphics[width=\textwidth]{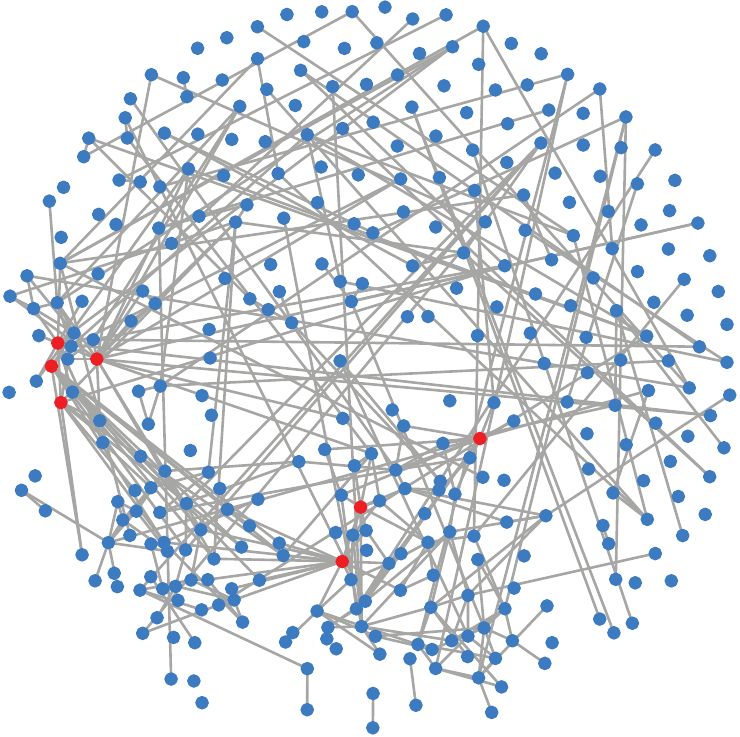}
                \caption{Reg.\ score matching}
                \label{fig:gull}
                \vspace*{6mm}
        \end{subfigure}
        \hspace*{4mm}
        \begin{subfigure}[b]{0.3\textwidth}
                \includegraphics[width=\textwidth]{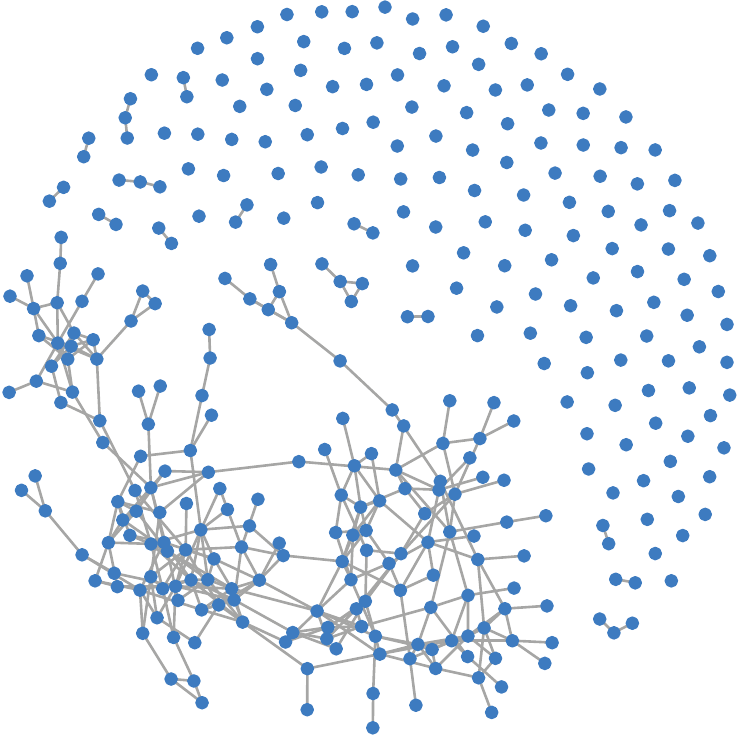}
                \caption{SPACE}
                \label{fig:tiger}
                \vspace*{6mm}
        \end{subfigure}
        \hspace*{4mm}
        \begin{subfigure}[b]{0.3\textwidth}
                \includegraphics[width=\textwidth]{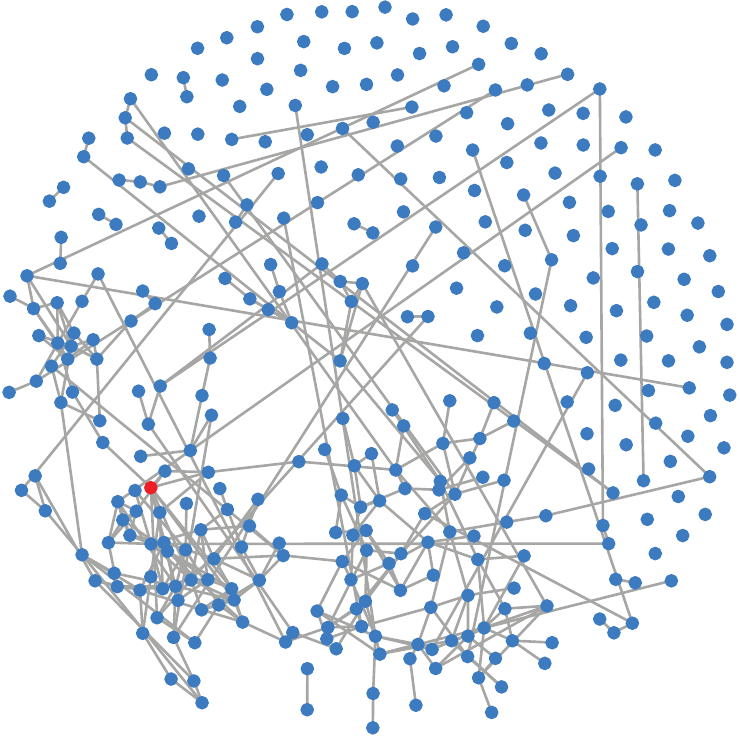}
                \caption{SPACEJAM}
                \label{fig:mouse2}
                \vspace*{6mm}
        \end{subfigure}
        \hfill
        \begin{subfigure}[b]{0.3\textwidth}
                \includegraphics[width=\textwidth]{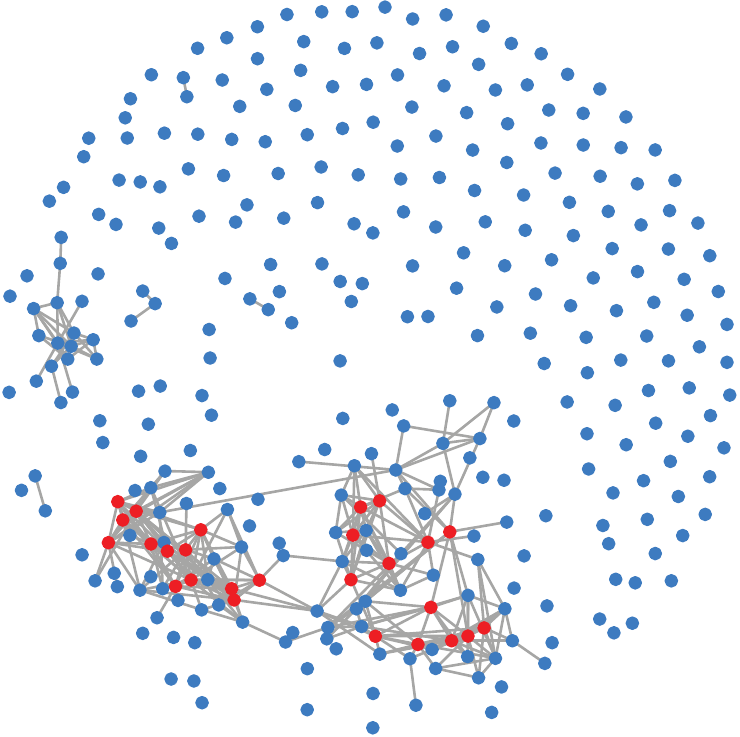}
                \caption{Glasso}
                \label{fig:mousea2}
                \vspace*{3mm}
        \end{subfigure}
        \hspace{10mm}
        \begin{subfigure}[b]{0.3\textwidth}
                \includegraphics[width=\textwidth]{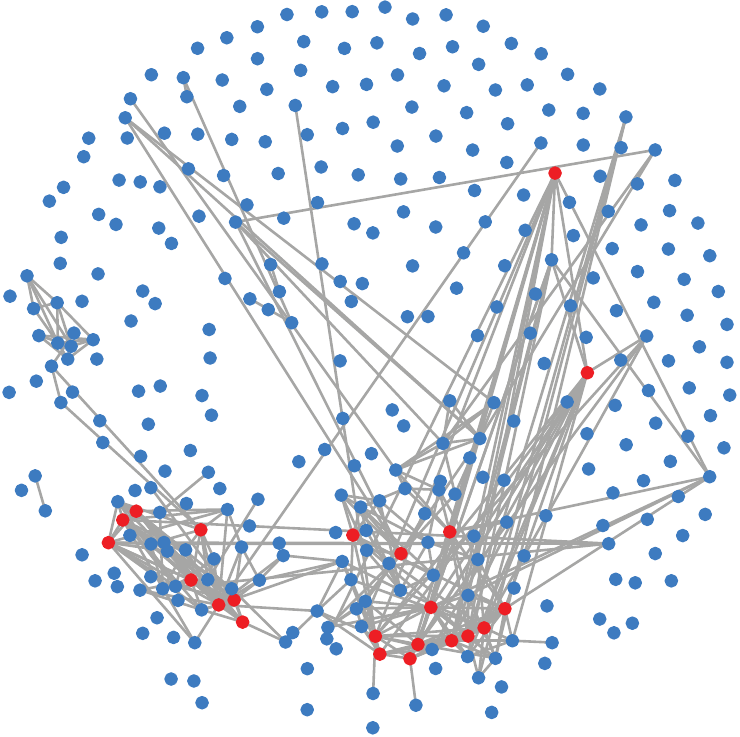}
                \caption{SKEPTIC}
                \label{fig:mouseb2}
                \vspace*{3mm}
        \end{subfigure}
        \caption{Topology of inferred networks of $|E| = 333$ or $334$
          edges for all considered methods. Layout of nodes is fixed across graph estimates and was optimized for the SPACE estimate. Isolated nodes have now been included.  Red colored nodes have degree greater or equal to 10.}\label{topology2}
\end{figure}

\begin{figure}[t]
\vskip -0.2in
\begin{center}
\centerline{\includegraphics[scale = 0.67]{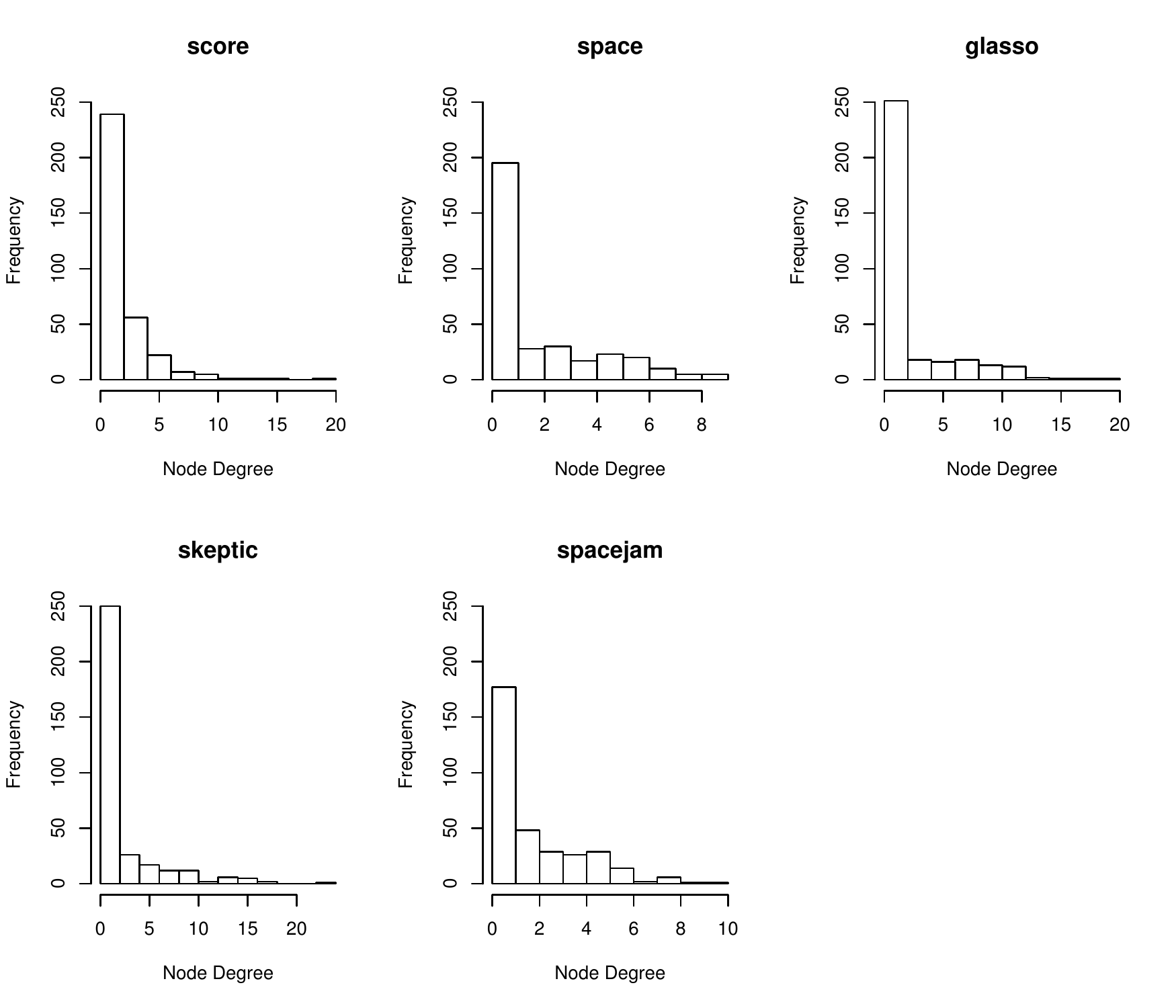}}
\caption{Node degree distributions for inferred networks of $|E| = 333$ or $334$ edges for all considered methods.}
\label{degdist}
\end{center}
\vskip -0.2in
\end{figure} 

By visual inspection, glasso and SKEPTIC give similar topologies,
which can be explained by the fact that both are derived from the full
Gaussian likelihood. Interestingly, we observe that SPACEJAM and SPACE
likewise yield similar graphs, which reinforces findings from
\cite{Shojaie2016}. Regularized non-negative score matching yields a
graph that is fairly different from the rest.

\begin{table}[t]
    \begin{tabular}{l|l|l|l|l}
    Reg. score matching & Glasso &  SKEPTIC  & SPACE & SPACEJAM  \\
    \hline
    CCNE2 (19) & EP300 (20) & PIK3CA (23) & TRAF6 (9) & BHX (10) \\
    PIK3CG (16) & SOS1 (17) & FZD7 (18) & TPR (9) & SOS2 (9) \\
    BRCA2 (13) & BAD (16) & PDGFRB (17) & SOS1 (9) & TRAF6 (8) \\
    BIRC5 (12) & TPR (13) & TGFBR2 (16) & JAK1 (9)& TGFBR2 (8) \\
    SKP2 (10) & RBX1 (13) & TCEB2 (16) & EP300 (9) &  SOS1 (8) \\
    PIK3CD(10) & PIK3CD (12) & MMP2 (16) & SOS2 (8) &  RRM2 (8) \\
    LAMB3 (10) & LAMA4 (12) & LAMA4 (16) & EGFR (8) &  PDGFRB (8) \\
    STAT5B (9) &  HRAS (12) & GLI2 (15) & CBL (8) &  EP300 (8) \\
    HRAS (9) & GLI2 (12) & SOS1 (14) & BAX (8) &  PIK3CA (7) \\
    PDGFRB (8) & TRAF6 (11) & PDGFRA (14) & APPL1 (8) &  ARNT (7) \\
   GSTP1 (8) & TGFBR2 (11) & MITF (14) &  &   \\
    & TCEB2 (11) & EP300 (14) &  &   \\
   &  SPI1 (11) & & &   \\
   & SOS2 (11) & &  &  \\
   & PDGFRB (11) &  & &   \\
   & MAP2K2 (11) &  & &   \\
   & APPL1 (11) &  & &  
    \end{tabular}
    \caption{The most densely connected genes according to the estimated graphs generated via nonnegative regularized score matching, glasso, SKEPTIC, SPACE and SPACEJAM. The number in parenthesis corresponds to the estimated degree of the gene.} 
    \label{table: topten}
\end{table}

While the usefulness of these models remains to be further explored, our case study demonstrates that regularized score matching can provide estimates that differ in interesting ways to the estimates generated by other methods. We compile a list of the top ten most highly connected genes in each of the estimated graphs in Table \ref{table: topten} (some lists have more than ten genes due to ties), as there is strong evidence that highly connected nodes play important roles in biological networks \citep{carter2004gene, jeong2001lethality, han2004evidence}. There are slight overlaps between the lists. Upon further inspection, we observe that six of the ten genes listed under regularized score matching have been previously linked to prostate cancer, five of which have not been identified by the competing methods:

\begin{itemize}
\item \textit{CCNE2} (cyclin E2):  a protein which is required for transition of the $G_1$ to $S$ phase of the cell cycle, which determines cell division. Regulated by PTEN, a tumor suppressor, it is over-expressed in metastatic prostate tumor cells \citep{wu2009cdc6}. 
\item \textit{BRCA2} (breast cancer 2): mutations in the BRCA2 gene have been associated with early-onset prostate cancer in men; men carrying mutations have a predisposition to more aggressive phenotypes \citep{gayther2000frequency, mitra2008prostate, tryggvadottir2007prostate, fan2006brca1}. 
\item \textit{BIRC5} (survivin): a protein which prevents cell death, or apoptosis, and regulates cell division. Heightened expression has been found to be associated with higher final Gleason score, i.e., more aggressive cancer and worse prognosis \citep{kishi2004expression, shariat2004survivin}. 
\item \textit{SKP2} (S-phase kinase-associated protein 2, E3 ubiquitin protein ligase): a positive regulator of the $G_1$ to $S$ phase of the cell cycle, which determines cell division. SKP2 labelling frequency in cancer was positively correlated with the Gleason score, and shown to be a significant predictor of reduced recurrence-free survival time after radical prostatectomy \citep{yang2002elevated, wang2008ar}. It has been proposed elsewhere as a promising therapeutic target for prostate cancer \citep{wang2012skp2}. 
\item \textit{STAT5B} (signal transducer and activator of transcription 5B): a transcription factor that encourages metastatic behavior of human prostate cancer cells. Its inhibition has been shown to induce apoptosis in human prostate cancer cells \citep{gu2010stat5, ahonen2003inhibition, moser2012stat5b}. 
\end{itemize}

Furthermore, via the Kolmogorov-Smirnov test, we fail to reject the hypothesis that the degrees of the nodes for the regularized score matching graph estimate follow a power law distribution, with significance level of 0.05. On the other hand, we reject this hypothesis for all other generated estimates at the same significance level. There is evidence that genetic networks are `scale-free', which implies that their degree distribution can be approximated by a power law distribution \citep{albert2005scale, barabasi1999emergence, jeong2001lethality}. In this aspect, the topology of regularized score matching estimate is most similar to the hypothesized structure of gene networks. 

Finally, we would like to emphasize that we do not intend to claim
that regularized score matching provides the \textit{best} estimate of
the underlying gene network, as the truth is unknown to us. What we
can posit is that truncated Gaussian may be a useful model that
provides potentially valid targets for therapy which may be missed by
other methods.

\section{Theory}\label{theory}

This section establishes high-dimensional model selection
consistency (sparsistency) of regularized score matching.
We focus on pairwise interaction models
as in (\ref{expfam:pairwise}), 
although our results could be extended to more
general models.  
Theorem~\ref{generalthm} below identifies general deterministic
conditions on data that yield sparsistency of regularized
(non-negative) score matching.  Two subsequent corollaries make
probabilistic statements about sparsistency in the Gaussian and the
non-negative Gaussian case.  Proofs are given in
Section~\ref{sec:proofs}.  Experiments that corroborate the
theoretical findings are shown in Appendix~\ref{sec:experiments}.

Before stating the main results, we describe a key assumption for
model selection consistency of $\ell_1$-penalized estimators, the
irrepresentability assumption, and highlight differences between
various estimators of Gaussian graphical models with respect to this
assumption. 

\subsection{Setup and notation}

We consider a continuous pairwise interaction model as given
by~(\ref{expfam:pairwise}) with symmetric $m\times m$ interaction
matrix $\boldsymbol{\Theta}=(\theta_{jk})$.  We let
$\theta=\mbox{vec}(\boldsymbol{\Theta})$.  Then the regularized score
matching estimator, in its basic or non-negative version, is
\begin{equation}
  \mycheck{\theta} = \arg\min_\theta
  \frac{1}{2} \theta^T
  \boldsymbol{\Gamma}(\mathbf{x}) \theta +  
  \g(\mathbf{x})^T\theta + c(\mathbf{x})  + \lambda\|\theta\|_{1}.
\end{equation}
By Lemma~\ref{blockform}, $\boldsymbol{\Gamma}(\mathbf{x})$ is a
symmetric $\m^2\times \m^2$ matrix that is block-diagonal, with 
blocks of size $\m\times \m$.  For notational convenience, we drop the
explicit reference to the data matrix $\mathbf{x}$ and denote
$\boldsymbol{\Gamma}(\mathbf{x})$ and $\g(\mathbf{x})$ as
$\boldsymbol{\Gamma}$ and $\g$.  

The true data-generating distribution is assumed to belong to the
considered model.  We denote the true interaction matrix by
$\boldsymbol{\Theta}^*=(\theta_{jk}^*)$ and its vectorization by
$\theta^*$.  We define $\boldsymbol{\Gamma}^*$ and $\g^*$ to be the
expected values of $\boldsymbol{\Gamma}$ and $\g$.  
The support of $\theta^*$, that is,
\[
S \equiv S(\theta^*) = \{(j,k) : j\not=k,\;\theta^*_{jk} \neq 0 \}
\]
is the edge set of the true conditional independence graph.
Similarly,
\[
\mycheck{S} \equiv S(\mycheck{\theta}) = \{(j,k) :
j\not=k,\;\mycheck{\theta}_{jk} \neq 0 \}
\]
determines the graph inferred by regularized score matching.  Finally,
we write $\dee$ for the maximum degree of the $\m$ nodes of the
conditional independence graph.  In other words, $\dee$ is the maximum
number of nonzero off-diagonal entries in any row (or column) of
$\boldsymbol{\Theta}^*$.

\subsection{Irrepresentability}

We say that the irrepresentability (or mutual incoherence) condition holds
with incoherence parameter $\alpha$ if the following assumption holds. 
\begin{assumption}\label{assumption}
There exists an $\alpha \in (0, 1]$ such that 
\begin{equation}
\vertiii{\boldsymbol{\Gamma}^*_{S^cS}(\boldsymbol{\Gamma}^*_{SS})^{-1}}_\infty \leq (1- \alpha). 
\end{equation} 
\end{assumption}

Irrepresentability conditions play a key role in the analysis of
$\ell_1$ regularization techniques \citep{Buhlmann2011}.  For
neighborhood selection in Gaussian graphical models, it has been
formulated in terms of the covariance matrix $\boldsymbol{\Sigma}^*$
\citep{meinshausen2006}.
In the theoretical analysis of the glasso, the constraint is placed on
the Hessian of the log-determinant of the precision matrix
$\mathbf{K}^*$,
i.e., $(\mathbf{K}^*)^{-1}
\otimes (\mathbf{K}^*)^{-1}$ \citep{RavikumarEtAl2011}.


In order to highlight the differences in conditions required for
sparsistency of glasso, neighborhood selection, SPACE and regularized
score matching, we revisit the Gaussian graphical model example
in \citet{Meinshausen2008}.  Let $\rho\in(0,1/\sqrt{2})$, and let
$\boldsymbol{\Sigma} = \begin{pmatrix} \sigma_{ij} \end{pmatrix} $ be
the $4\times 4$ covariance matrix with ones along the diagonal,
$\sigma_{23}=\sigma_{32}=0$, $\sigma_{14}=\sigma_{41}=2\rho^2$ and all
other off-diagonal entries equal to $\rho$. The precision matrix
$\mathbf{K} =(\boldsymbol{\Sigma})^{-1}$ then has
$\kappa_{14}=\kappa_{41}=0$.  The conditional independence graph $G$
is as in Figure~\ref{graph}.

Meinshausen showed that for samples drawn from
$N(0, \boldsymbol{\Sigma})$, glasso can consistently recover $G$ only
if $\rho \le \sqrt{3/2} -1 \approx 0.23$. For neighborhood selection,
the corresponding necessary condition is $\rho \le 0.5$. If these
conditions fail, then for large sample size, the probability of
erroneously including the edge $(1, 4)$, i.e.,
$P\left( \hat{\kappa}_{14} \neq 0 \right)$ can be shown to be at least 0.5.
It turns out that for regularized score matching, the analogous
necessary condition gives a bound that falls in between 0.23 and 0.5,
specifically, $\rho \le \sqrt{2} - 1 \approx 0.41$.

We observe that glasso, which yields positive definite estimates,
requires the most stringent condition.  When working with symmetric
matrices as in regularized score matching, the condition is markedly
relaxed.  Allowing non-symmetric matrices in neighborhood selection
leads to further relaxation of the condition.  Interestingly, the
pseudo-likelihood methods classified under SPACE have the same
necessary condition as score matching.

Assumption~\ref{assumption} should be seen as sufficient for
consistency of regularized score matching.  For Meinshausen's example,
it can be shown to amount to
$ \rho < \frac{1}{2}(\sqrt{3} - 1) \approx 0.37$.  The analogous
sufficient condition for glasso from \citet{RavikumarEtAl2011}
requires that $\rho < \frac{1}{2}(\sqrt{2} - 1) \approx 0.21$.  For
neighborhood selection, the condition is $\rho<0.5$.

\subsection{Main Results}

We define
 \begin{align}\label{quantities}
c_{\boldsymbol{\Gamma}^*} = \vertiii{ (\boldsymbol{\Gamma}_{SS}^*)^{-1}}_\infty , ~ \mbox{and} ~ c_{\mathbf{\Theta}^*} = \vertiii{ \boldsymbol{\Theta}^*}_\infty.
\end{align}
Moreover, let
\begin{align}
  \label{eq:R1r2r3}
  \mathbf{R}_1&=(\boldsymbol{\Gamma} - \boldsymbol{\Gamma}^*), 
  &r_2 &=\g^* - \g,
  &r_3&=\boldsymbol{\Gamma}^*\theta^* - \g^*,
\end{align}
such that the KKT conditions from~(\ref{generalkkt}) can be written as
\begin{equation}\label{largekkt} 
\boldsymbol{\Gamma}^*(\mycheck{\theta} - \theta^*) +
R_1\mycheck{\theta}  + {r_2} + {r_3} + \lambda \mycheck{z}= {0},  \quad \mycheck{z} \in \partial
    \|\mycheck{\theta}\|_{1}.
\end{equation}

\begin{theorem}\label{generalthm}
  Assume that $\boldsymbol{\Gamma}^*_{SS}$ is invertible and the
  irrepresentability condition holds with incoherence parameter
  $\alpha \in (0, 1]$ (Assumption \ref{assumption}).  Furthermore,
  assume that 
\begin{align}
\|\mathbf{R}_1\|_\infty &< \epsilon_1,  & 
\|r_2\|_\infty &< \epsilon_2, \label{bigassumptions2}
\end{align}
with $\dee\epsilon_1 \leq \alpha/(6c_{\boldsymbol{\Gamma^*}})$.  If
\begin{equation}
\lambda > \frac{3(2-\alpha)}{\alpha}\max\{c_{\boldsymbol{\Theta}^*}\epsilon_1, \epsilon_2\},
\end{equation} 
then the following statements hold:
\begin{enumerate}
\item[(a)] The rSME $\mycheck{\theta}$ is unique, has its support
  included in the true support ($\hat S\subseteq S$), and satisfies
  \[
  \| \mycheck{\theta} - \theta^*  \|_\infty < \frac{c_{\boldsymbol{\Gamma}^*}}{2-\alpha}\,\lambda.
  \]
\item[(b)] If
  \[ \underset{1 \leq j < k \leq \m}{\min} |\theta^*_{jk}| > \frac{c_{\boldsymbol{\Gamma}^*}}{2-\alpha}\,\lambda,
\]
then $\mycheck{S}=S$ and $\sign(\mycheck{\theta}_{jk})=\sign(\theta^*_{jk})$
for all $(j,k)\in S$. 
\end{enumerate}
\end{theorem}

Theorem \ref{generalthm} imposes deterministic conditions on the data, namely, the
bounds in (\ref{bigassumptions2}).  In the following corollaries, we will
consider specific distributional assumptions and impose population
conditions that 
imply bounds of the form (\ref{bigassumptions2}) with high probability.

First, we provide a result for regularized score matching for the
Gaussian case (Example \ref{example1}), which has
$\boldsymbol{\Gamma} = \mathbf{I}_{\m \times \m} \otimes \mathbf{W}$
with $\mathbf{W}$ being the sample covariance matrix, and
$\g = \mbox{vec}(\mathbf{I}_{\m \times \m})$.  When the data is
generated from a normal distribution with covariance matrix
$\boldsymbol{\Sigma}^*$ then
$\boldsymbol{\Gamma} = \mathbf{I}_{\m \times \m}
\otimes\boldsymbol{\Sigma}^*$
and, of course, $\g^*=\g=\mbox{vec}(\mathbf{I}_{\m \times \m})$.


\begin{corollary}\label{th1}
  Suppose the data is generated from a normal distribution
  $N(0,\boldsymbol{\Sigma}^*)$ such that
  $\boldsymbol{\Gamma}^*_{SS}$ is invertible and irrepresentability
  holds for $\alpha \in (0, 1]$.  Let
  $\mathbf{K}^*=(\kappa^*_{jk})=(\boldsymbol{\Sigma}^*)^{-1}$, 
  \[
  c^* = 3200 \underset{j = 1, \ldots , \m}{\max}
  (\boldsymbol{\Sigma}_{jj}^*)^2  \quad\text{and}\quad
  c_1 = \frac{4}{\alpha}c_{\mathbf{\Gamma}^*}.
  \]
  Take any $\tau_1 > 2$.  If the sample size satisfies
  \begin{equation}\label{n}
    n > c^*c_1^2\dee^2(\log \m^{\tau_1} + \log 4), 
  \end{equation}
  and the regularization parameter is
  \begin{equation}
    \lambda >
    \frac{2c_{\mathbf{K}^*}(2-\alpha)}{\alpha}\sqrt{\frac{c^*(\log
        \m^{\tau_1} + \log 4)}{n}},
  \end{equation}
  then the following statements hold with probability
  $1 - 1/\m^{\tau_1 - 2}$:
\begin{enumerate}
\item[(a)] The rSME $\mycheck{\mathbf{K}}$
  from~(\ref{eq: regscore}) is unique, has its support included in the
  true support ($\hat S\subseteq S$), and satisfies
  \[
  \| \mycheck{\mathbf{K}} - \mathbf{K}^*  \|_\infty < \frac{c_{\boldsymbol{\Gamma}^*}}{2-\alpha}\,\lambda.
  \]
\item[(b)] If
  \[ \underset{1 \leq j < k \leq \m}{\min} |\kappa^*_{jk}| > \frac{c_{\boldsymbol{\Gamma}^*}}{2-\alpha}\,\lambda,
\]
then $\mycheck{S}=S$
and $\sign(\mycheck{\mathbf{K}}_{jk})=\sign(\kappa^*_{jk})$
for all $(j,k)\in S$.
\end{enumerate}
\end{corollary}

The corollary is proven in Appendix~\ref{sec:proof-coroll-refth1}.
Numerical experiments reported in Appendix~\ref{sec:experiments}
suggest that the sample size $n$
indeed needs to scale at least $\Omega(\dee^2\log\m)$
for sparsistency.

From Theorem~\ref{generalthm}, we can also derive an analogous result
for regularized non-negative score matching for the truncated Gaussian
case (Example \ref{example2}).  The result requires the sample size to
be larger than in the Gaussian case, due to the need to control higher
order moments.  Recall that here, $\boldsymbol{\Gamma}(\mathbf{x})$ a
block diagonal $\m^2 \times \m^2$ matrix, with the $j$th block given
by
\[
\frac{1}{n}\sum_{i=1}^n x_{ij}^{2}x^{(i)}x^{(i)T},
\]
and $g = 2w + w_{\mbox{\tiny{diag}}}$, where $w =
\mbox{vec}(\mathbf{W})$ and $w_{\mbox{\tiny{diag}}} =
\mbox{vec}(\mbox{diag}(\mathbf{W}))$.   



\begin{corollary}\label{th2}
  Suppose the data is generated from a non-negative Gaussian
  distribution with parameter $\mathbf{K}^*$, i.e.,
  $N(0, (\mathbf{K}^*)^{-1})$ is truncated to $\mathbb{R}^\m_+$.
  Suppose further
  that $\boldsymbol{\Gamma}^*_{SS}$ is invertible and
  irrepresentability holds for $\alpha \in (0, 1]$. Let
\[
c^{**} =  \max\left\{\left(\frac{L}{2}\right)^4 \sqrt{\max_{j} \Var[X_j^4]}, \left(\frac{L}{2}\right)^2\sqrt{\max_{j} \Var[X_j^2]} \right\}  \quad\text{and}\quad
c_2 = \frac{6}{\alpha}c_{\boldsymbol{\Gamma}^*}
\] 
where $L > 0$ is an absolute constant.
Take any $\tau_2 > 3$. If the sample size satisfies
\begin{align}
n &> c^{**}c_2^2\dee^2(\log \m^{\tau_2} + \log 2)^8, \label{nnn}
\end{align}
and the regularization parameter is
\begin{equation}
\lambda > \frac{3(2-\alpha)}{\alpha}\max\{c_{\mathbf{K}^*}, 1 \}\sqrt{\frac{c^{**}(\log \m^{\tau_2} + \log 2)^8}{n}}, 
\end{equation}
then the following statements hold with probability $1 -
\frac{1}{\m^{\tau_2-3}}$:  
\begin{enumerate}
\item[(a)] The rSME $\mycheck{\mathbf{K}}_+$ based on
  penalizing~(\ref{nngaussiansme}) with
  $\lambda \|\mathbf{K}\|_{1, \mbox{\tiny{off}}}$  is unique, has its support included in the true support
  ($\hat S\subseteq S$), and satisfies
  \[
  \| \mycheck{\mathbf{K}}_+ - \mathbf{K}^*  \|_\infty < \frac{c_{\boldsymbol{\Gamma}^*}}{2-\alpha}\,\lambda.
  \]
\item[(b)] If
  \[ \underset{1 \leq j < k \leq \m}{\min} |\kappa^*_{jk}| > \frac{c_{\boldsymbol{\Gamma}^*}}{2-\alpha}\,\lambda,
\]
then $\mycheck{S}=S$
and $\sign((\mycheck{\mathbf{K}}_+)_{jk})=\sign(\kappa^*_{jk})$
for all $(j,k)\in S$.
\end{enumerate}
\end{corollary}

The proof of the corollary, which is given in
Section~\ref{sec:proof-coroll-refth2}, uses general tail bounds that
apply to log-concave measures.  The lower bound for $n$ given
in~(\ref{nnn}) could well be suboptimal and a lower power of $\log\m$
may be sufficient for sparsistency.  However, the experiments in
Appendix~\ref{sec:experiments} suggest that the exponent for $\log \m$
cannot be taken too much smaller than 8.

We also compared the lower bound we obtained for the non-negative
Gaussian case to a result implied by the work of
\citet{yang:etal:2013} who treat consistency of neighborhood selection
in a general framework that allows node-wise conditional distributions
to arise from exponential families.  Interestingly, when working out
what their general theorem would say about the above non-negative
Gaussian model we found that the sample size $n$ would also be required
to be at least $\Omega(\dee^2(\log \m)^8)$.  Our result from
Corollary~\ref{th2} is thus at least comparable to existing results in the literature.

\section{Proofs}
\label{sec:proofs}
\subsection{Proof of Theorem~\ref{generalthm}}
\label{sec:proof-theor-refg}
  First, we note that claim (b) is an immediate consequence of claim
  (a).  To show (a), we apply the primal-dual witness method (PDW)
  from \citet{Wainwright2009}.  As explained in detail below, PDW
  entails construction of a pair $(\tilde{\theta}, \tilde{z})$, with
  $\tilde{\theta}\in\mathbb{R}^{\m^2}$ and
  $\tilde{z}\in\partial\|\tilde{\theta}\|_{1}$, that satisfies the KKT
  optimality conditions from~(\ref{largekkt}) and has the support of
  $\tilde{\theta}$ included in $S$.  If the construction is successful
  then it ensures that the rSME problem admits a unique solution such
  that the rSME $\mycheck{\theta}$ is equal to $\tilde{\theta}$ and
  inherits all the properties the latter has by definition.  These
  properties include the $\ell_\infty$ bound on estimation error in
  addition to the claim about the support.


  Replacing $\boldsymbol{\Gamma}$ by $\boldsymbol{\Gamma}^*$ and $g$ by $g^*$ in the empirical
  (basic or non-negative) score matching loss recovers the population
  loss which, in the present exponential family context, is quadratic
  and minimized when $\theta = \theta^*$.  (Recall that the score
  matching loss is consistent.)  It follows that $r_3$
  from~(\ref{eq:R1r2r3}) is zero as it is the gradient of the
  population loss.  In block form, (\ref{largekkt}) becomes
\begin{equation}\label{blockkkt}
\begin{bmatrix}\boldsymbol{\Gamma}^*_{SS} &
  \boldsymbol{\Gamma}^*_{SS^c} \\ \boldsymbol{\Gamma}^*_{S^cS} &
  \boldsymbol{\Gamma}^*_{S^cS^c} \end{bmatrix} \begin{bmatrix}
  \mycheck{\theta}_S - \theta^*_S \\ \mycheck{\theta}_{S^c} -
  \theta^*_{S^c} \end{bmatrix} + \begin{bmatrix}\mathbf{R}_{1, SS} &
  \mathbf{R}_{1, SS^c} \\ \mathbf{R}_{1, S^cS}  & \mathbf{R}_{1,
    S^cS^c}  \end{bmatrix}\begin{bmatrix}\mycheck{\theta}_S \\
  \mycheck{\theta}_{S^c} \end{bmatrix} + \begin{bmatrix} r_{2,S}  \\
  r_{2, S^c}  \end{bmatrix} +  \lambda  \begin{bmatrix} \mycheck{z}_S
  \\ \mycheck{z}_{S^c} \end{bmatrix} = \begin{bmatrix} 0  \\
  0 \end{bmatrix}.  
\end{equation}
We construct the PDW pair $(\tilde{\theta}, \tilde{z})$ according to
the following steps:
\begin{enumerate}
\item Take $\tilde{{\theta}}$ to be the unique solution to the
  support-restricted problem, that is,
\begin{equation}\label{restricted}
\tilde{\theta} = \mbox{arg} \underset{\theta_{S^c} = {0}}{\min} ~\frac{1}{2}\theta^T\boldsymbol{\Gamma}\theta - \g^T\theta + \lambda\|\theta\|_{1}. 
\end{equation}
\item Choose \[\tilde{z}_S\in \partial \|\tilde{\theta}_S \|_{1}.\]
\item  Solving (\ref{blockkkt}), set
\begin{align}
  \nonumber
\tilde{z}_{S^c} = \frac{1}{\lambda} \Big[
  -\boldsymbol{\Gamma}^*_{S^cS} &
  (\boldsymbol{\Gamma}^*_{SS})^{-1}\left(\mathbf{R}_{1,SS}\tilde{\theta}_S
    + r_{2, S}\right) 
\\  \label{zsc}
& +\mathbf{R}_{1,S^cS}\tilde{\theta}_S + r_{2, S^c} + \lambda\boldsymbol{\Gamma}^*_{S^cS} (\boldsymbol{\Gamma}^*_{SS})^{-1}\tilde{z}_S \Big]. 
\end{align}
\item Check the \textit{strict
    dual feasibility} condition that
  \begin{equation}\label{eq:strict-dual-feas}
    \|\tilde{z}_{S^c}\|_\infty < 1. 
  \end{equation}
\end{enumerate} 

By step (i), $\tilde{\theta}$ has support contained
in $S$.  By step (iii), $(\tilde{\theta},\tilde{z})$ is guaranteed to
fulfill the equations from~(\ref{blockkkt}).  By step (ii), the
$S$-coordinates of $\tilde{z}$ satisfy `their part' of the subgradient
condition.  Thus, if the strict dual feasibility from step (iv) holds,
then $(\tilde{\theta},\tilde{z})$ satisfies the KKT conditions
from~(\ref{largekkt}).  Having a strict inequality
in~(\ref{eq:strict-dual-feas}) ensures that every solution to the
original rSME problem has support contained in the true support $S$
and since $\boldsymbol{\Gamma}^*_{SS}$ is assumed invertible, there is
then only one solution \citep[Lemma 1]{Wainwright2009}.  The
invertibility of $\boldsymbol{\Gamma}^*_{SS}$ is also what guarantees
the uniqueness in step (i).

If the PDW construction is successful, that is, if the strict dual
feasibility condition can be established, then we may conclude the
rSME $\mycheck{\theta}$ possesses all the desired properties.  Indeed,
$\mycheck{\theta}$ equals $\tilde{\theta}$ which has these properties
by construction.


Let $\tilde{\Delta} = \tilde{\theta} - \theta^*$, where
$\tilde{\theta}$ is the solution to the support-restricted regularized
score matching problem from~(\ref{restricted}).   By
definition,
$\|\tilde{\Delta} \|_\infty = \|\tilde{\Delta}_S
\|_\infty$. Furthermore, by step (iii) in the PDW construction,
\begin{multline}
\tilde{z}_{S^c} =
\frac{1}{\lambda}\bigg[\boldsymbol{\Gamma}^*_{S^cS}(\boldsymbol{\Gamma}_{SS}^*)^{-1}(\mathbf{R}_{1,SS}(\theta^*_S
+ \Delta_S) + r_{2,S}) - \mathbf{R}_{1, S^cS}(\theta^*_S + \Delta_S) -
r_{2, S^c} \bigg]\\
+ \boldsymbol{\Gamma}^*_{S^cS}(\boldsymbol{\Gamma}_{SS}^*)^{-1}\tilde{z}_S.
\end{multline}
By Assumption \ref{assumption}, and the triangle inequality for the
$\ell_\infty$ norm,
\begin{align}
\|\tilde{z}_{S^c}\|_\infty &\leq \frac{1}{\lambda}\bigg[(1-\alpha)\,(\|\mathbf{R}_{1,SS}(\theta^*_S + \Delta_S)\|_\infty + \|r_{2,S}\|_\infty) \nonumber\\
{}&\qquad\qquad + \|\mathbf{R}_{1, S^cS}(\theta^*_S + \Delta_S)\|_\infty + \|r_{2,S^c} \|_\infty \bigg] + (1-\alpha)\nonumber\\
&\leq \frac{(2-\alpha)}{\lambda} \bigg[\|\mathbf{R}_{1,\cdot S}(\theta^*_S + \Delta_S)\|_\infty + \|r_2\|_\infty \bigg] + (1 -\alpha) \nonumber\\
&=  \frac{(2-\alpha)}{\lambda} \bigg[\|\mathbf{R}_{1}\theta^* + \mathbf{R}_{1, \cdot S}\Delta_S\|_\infty + \|r_2\|_\infty \bigg] + (1 -\alpha) \nonumber\\
&\leq \underbrace{\frac{(2-\alpha)}{\lambda} \,
  \|\mathbf{R}_1\theta^*\|_\infty}_{= G_1} +
  \underbrace{\frac{(2-\alpha)}{\lambda} \,\vertiii{\mathbf{R}_{1, \cdot
  S}}_\infty\|\Delta_S\|_\infty}_{= G_2} +
  \underbrace{\frac{(2-\alpha)}{\lambda}\|r_2\|_\infty}_{= G_3} +
  (1-\alpha),  
\nonumber
\end{align}
where the equality in the second to last line follows from the fact that
$\theta^*_{S^c} = 0$.  

We observe that
\begin{align}
G_1 = \frac{(2-\alpha)}{\lambda} \times
  \|\boldsymbol{\Theta}^*_{\mbox{\tiny{wide}}}\mbox{vec}(\mathbf{R}_{1,
  \mbox{{\tiny blocks}}})\|_\infty
\label{eq:wide}
\end{align}
where 
\[ \boldsymbol{\Theta}^*_{\mbox{\tiny{wide}}} = \begin{bmatrix}
  \theta^{*T}_1 & 0 & \ldots & \ldots & 0 \\ 0 & \theta^{*T}_1 & 0 &
  \ldots & \vdots \\ \vdots & 0 & \ddots & \ddots & \ldots \\ \vdots &
  \vdots & \ddots & \theta^{*T}_\m & 0 \\ \vdots & \vdots & \vdots & 0
  & \theta^{*T}_\m \end{bmatrix}
\]
is an $\m^2 \times \m^3$ matrix whose diagonal blocks are given by the
rows of the the interaction matrix $\boldsymbol{\Theta}^*$, each row
being replicated $\m$ times.  Moreover,
$\mbox{vec}(\mathbf{R}_{1, \mbox{{\tiny blocks}}})$ refers to the
vectorization of the $\m$ diagonal blocks of $\mathbf{R}_1$ that are
each of size $\m\times \m$; recall Lemma \ref{blockform}.  More
precisely, if $\mathbf{R}_{1, 1}, \ldots , \mathbf{R}_{1, m}$ are the
diagonal blocks of $\mathbf{R}_1$, then
$\mbox{vec}(\mathbf{R}_{1, \mbox{{\tiny blocks}}})$ is obtained by
concatenating
$\mbox{vec}(\mathbf{R}_{1,1}), \ldots , \mbox{vec}(\mathbf{R}_{1, m})$
in that order.   Equation~(\ref{eq:wide}) is the only argument relying
on the block-diagonality of $\boldsymbol{\Gamma}$ and $\mathbf{R}_1$.

From~(\ref{eq:wide}), we obtain that
\[
G_1 \leq
\frac{(2-\alpha)}{\lambda}\vertiii{\boldsymbol{\Theta}^*_{\mbox{\tiny{wide}}}}_\infty\|\mbox{vec}(\mathbf{R}_1)\|_\infty
<
\frac{(2-\alpha)}{\lambda}\vertiii{\boldsymbol{\Theta}^*_{\mbox{\tiny{wide}}}}_\infty\epsilon_1. 
\]
since we have assumed that
$\|\mbox{vec}(\mathbf{R}_1)\|_\infty=\|\mathbf{R}_1\|_\infty <
\epsilon_1$.
By construction,
$\vertiii{\boldsymbol{\Theta}^*_{\mbox{\tiny{wide}}}}_\infty =
\vertiii{\boldsymbol{\Theta}^*}_\infty = c_{\boldsymbol{\Theta}^*}$.
It follows, from our choice of $\lambda$ that $G_1 < \alpha/3$.

By the assumption that $\|r_2\|_\infty < \epsilon_2$, we
have
\[G_3 < \frac{(2-\alpha)}{\lambda}\epsilon_2 < \frac{\alpha}{3}, \]
and it remains to similarly bound $G_2$.  We treat
$\vertiii{\mathbf{R}_{1,\cdot S}}_\infty$ and
$\|\tilde{\Delta}_S\|_\infty$ separately. 

We note that the rows of $\mathbf{R}_{1,\cdot S}$ have at most $\dee$
non-zero elements.  It follows that
$\vertiii{\mathbf{R}_{1,\cdot S}}_\infty \leq
\dee\|\mathbf{R}_1\|_\infty < \dee\epsilon_1 <
\alpha/6c_{\boldsymbol{\Gamma}}$,
where the last inequality holds by assumption.  Since
$\boldsymbol{\Gamma}_{SS}$ is assumed invertible, we have from the top
block of equations in (\ref{blockkkt}) that
\[
\tilde{\Delta}_S = (\boldsymbol{\Gamma}_{SS})^{-1} (-\mathbf{R}_{1,
  SS}\theta_S^* - \lambda \tilde{z}).
\] 
Note that by assumption, $\boldsymbol{\Gamma}_{SS}$ is invertible.  We
obtain that
\begin{align}
\|\tilde{\Delta}_S \|_\infty &\leq \vertiii{(\boldsymbol{\Gamma}_{SS})^{-1}}_\infty\bigg[ \|\mathbf{R}_{1, SS}\theta^*_S\|_\infty + \|r_2\|_\infty + \lambda \bigg] \nonumber \\
& <   \vertiii{(\boldsymbol{\Gamma}_{SS})^{-1}}_\infty\bigg[ \vertiii{\boldsymbol{\Theta}^*_{\mbox{\tiny{wide}}}}_\infty\|\mbox{vec}(\mathbf{R}_1)\|_\infty   +   \|r_2\|_\infty + \lambda \bigg] \nonumber\\
&\leq \vertiii{(\boldsymbol{\Gamma}_{SS})^{-1}}_\infty \times \frac{(6
  - \alpha)}{3(2 - \alpha)}\lambda .\label{Delta}  
\end{align}
Since $\|\mathbf{R}_1\|_\infty<\epsilon_1$, we have
$\vertiii{\mathbf{R}_{1, SS}}_\infty \leq \dee\epsilon_1 <
1/c_{\boldsymbol{\Gamma^*}}$.  This implies that
\[\vertiii{(\boldsymbol{\Gamma}_{SS}^*)^{-1}\mathbf{R}_{1, SS}}_\infty \leq \vertiii{(\boldsymbol{\Gamma}_{SS}^*)^{-1}}_\infty\vertiii{\mathbf{R}_{1, SS}}_\infty < 1,\]
which gives us the following bound in the error in the inverse in the matrix $\ell_\infty$ norm, 
\begin{align*}
\vertiii{(\boldsymbol{\Gamma}_{SS})^{-1} - (\boldsymbol{\Gamma}_{SS}^*)^{-1}}_\infty &\leq \frac{\vertiii{(\boldsymbol{\Gamma}_{SS}^*)^{-1}\mathbf{R}_{1, SS}}_\infty}{ 1 - \vertiii{(\boldsymbol{\Gamma}_{SS}^*)^{-1}\mathbf{R}_{1, SS}}_\infty} \times \vertiii{(\boldsymbol{\Gamma}_{SS}^*)^{-1}}_\infty \\
&\leq   \frac{\vertiii{(\boldsymbol{\Gamma}_{SS}^*)^{-1}}_\infty\vertiii{\mathbf{R}_{1, SS}}_\infty}{1-\vertiii{(\boldsymbol{\Gamma}_{SS}^*)^{-1}}_\infty\vertiii{\mathbf{R}_{1,SS}}_\infty}   \times \vertiii{(\boldsymbol{\Gamma}_{SS}^*)^{-1}}_\infty . 
\end{align*}
Application of the triangle inequality, along with our definition of $c_{\boldsymbol{\Gamma}^*} = \vertiii{(\boldsymbol{\Gamma}_{SS}^*)^{-1}}_\infty$, yields
\begin{align}
\vertiii{(\boldsymbol{\Gamma}_{SS})^{-1}}_\infty &\leq \vertiii{(\boldsymbol{\Gamma}_{SS}^*)^{-1}}_\infty + \vertiii{(\boldsymbol{\Gamma}_{SS})^{-1} - (\boldsymbol{\Gamma}_{SS}^*)^{-1}}_\infty \nonumber \\
&= \vertiii{(\boldsymbol{\Gamma}_{SS}^*)^{-1}}_\infty \times  \frac{1}{1-\vertiii{(\boldsymbol{\Gamma}_{SS}^*)^{-1}}_\infty\vertiii{\mathbf{R}_{1,SS}}_\infty} \nonumber\\
&\leq   \frac{c_{\boldsymbol{\Gamma}^*}}{1-\dee
  c_{\boldsymbol{\Gamma}^*}\epsilon_1}\nonumber
\\
& \leq \frac{c_{\boldsymbol{\Gamma}^*} }{1 - \alpha/6}, \label{gammainverse}
\end{align} 
where the last inequality uses the assumption that
$\dee\epsilon_1 \leq \alpha/6c_{\boldsymbol{\Gamma}^*}$.
Substituting~(\ref{gammainverse}) into (\ref{Delta}), it is
straightforward to show that $G_2 < \alpha/3$. Therefore,
$G_1+G_2+G_3<\alpha$, which yields that $\|\tilde{z}_{S^c}\| < 1$.

Along the way we have also proven the second part of the claim.
Indeed, from (\ref{Delta}) and (\ref{gammainverse}), we have
\[\|\tilde{\Delta}_S\|_\infty \leq  \frac{c_{\boldsymbol{\Gamma}^*} }{1 - \alpha/6}  \times  \frac{(6 - \alpha)}{3(2 - \alpha)}\lambda  = \frac{2c_{\boldsymbol{\Gamma}^*}\lambda}{2-\alpha} . \]

\subsection{Proof of Corollary \ref{th1}}
\label{sec:proof-coroll-refth1}

We need to show that the conditions in Theorem \ref{generalthm},
specifically those in~(\ref{bigassumptions2}), hold with the claimed
probability.  Since
$r_2 = \g - \g^* = \mbox{vec}(\mathbf{I}_{\m \times \m}) -
\mbox{vec}(\mathbf{I}_{\m \times \m}) = 0$,
the second inequality in~(\ref{bigassumptions2}) can be trivially
satisfied with any $\epsilon_2 > 0$.  Thus, we only need to show that
we can bound $\|\mathbf{R}_1\|_\infty$ by some suitable $\epsilon_1$
with sufficiently large probability.  To do so, we apply a Bernstein-type
concentration inequality for the entries of $W$ that is also used by
\cite{RavikumarEtAl2011}.  Lemma~\ref{lem:ravi-concentration} below
states the inequality, as given in their paper.

The matrix $\mathbf{R}_1$ features only entries in
$\mathbf{W} - \boldsymbol{\Sigma}^*$.  By taking a union bound over
the $\m^2$ entries of $\mathbf{W}$, plugging in our lower bound for
$n$ and observing that $\sigma = 1$ in the Gaussian case, 
 Lemma~\ref{lem:ravi-concentration} yields that
\begin{equation*}
 \Pr\bigg[\|\mathbf{R}_1 \|_\infty \geq   \sqrt{\frac{c^*(\log
     \m^{\tau_1} + \log 4)}{n}} \bigg]  \leq \exp\left\{-\log \m^{\tau_1}  +
   2\log \m \right\} = \frac{1}{\m^{\tau_1 - 2}} . 
\end{equation*}
In addition, each row in $\vertiii{\mathbf{R}_{\cdot S}}_\infty$
features at most $\dee$ entries from the matrix
$\mathbf{W} - \boldsymbol{\Sigma}^*$.  Hence, it follows from another union bound, and choosing $n$ at least
\[
c^*c_1^2\dee^2(\log \m^{\tau_1} + \log 4)
\]
where $c^*$ and $c_1$ are defined in the corollary statement, that
\begin{equation*}
\Pr\bigg[ \vertiii{\mathbf{R}_{\cdot S}}_\infty  > \frac{1}{c_1} \bigg ]  \leq \frac{1}{\m^{\tau_1 - 2}}.
\end{equation*} 
Thus, applying Theorem~\ref{generalthm} with
\[\epsilon_1 =  \sqrt{\frac{c^*(\log \m^{\tau_1} + \log 4)}{n}} \]
shows that our choices for $\lambda$ and $n$ give the high probability
statement in Corollary \ref{th1}.

When looking back at the proof of Theorem \ref{generalthm}, we see
that as a consequence of having $r_2 = 0$, we need only
be concerned with bounding terms $G_1$ and $G_2$.  We may thus
bound $G_1$ and $G_2$ each by $\alpha/2$ instead of $\alpha/3$ and
ignore the $G_3$ term entirely, as it is $0$.  This leads us to
having $c_1 = (4/\alpha)c_{\boldsymbol{\Gamma}^*}$, as opposed to
the expected $(6/\alpha)c_{\boldsymbol{\Gamma}^*}$.

\subsection{Proof of Corollary \ref{th2}}
\label{sec:proof-coroll-refth2}

We proceed as for the proof of Corollary \ref{th1} and use
concentration results to satisfy the bounds
from~(\ref{bigassumptions2}) in Theorem \ref{generalthm}.  However, we
now bound $\|\mathbf{R}_1\|_\infty$
and $\|r_2\|_\infty$ using concentration inequalities for general
log-concave measures (any truncated multivariate normal density is
log-concave).

Let $X^{(i)} = (X_{i1},\dots,X_{i\m})$ be i.i.d.~according to
$N(0, (\mathbf{K}^*)^{-1})$ with truncation to $\mathbb{R}^\m_+$.  
Take
\begin{align}
\epsilon_1 &= \frac{\left[\left(\frac{L}{2}\right)(\log \m^{\tau_2} + \log 2)\right]^4}{\sqrt{n}}\sqrt{\underset{j}{\max}~\Var[X_j^4]}, \label{e1}\\
\epsilon_2 &= \frac{\left[\left(\frac{L}{2}\right)(\log \m^{\tau_2}+ \log 2)\right]^2}{\sqrt{n}}\sqrt{\underset{j}{\max}~\Var[X_j^2]} \label{e2}. 
\end{align}
From Lemma \ref{polynomiallemma} below, we know that for the absolute
constant $L$ specified in Lemma \ref{markovlemma}, we have,
\begin{align*}
\Pr\left[\left|\frac{1}{n}\sum_{i=1}^n X_{ij}X_{ik}X_{i\ell}^2 - E[X_jX_kX_\ell^2]\right| > \epsilon_1\right] &< \exp\left\{ - \frac{2}{L}\bigg(\frac{\sqrt{n}\epsilon_1}{\sqrt{\underset{j, k, \ell}{\max}~\Var[X_jX_kX_\ell^2]}}\bigg)^{\frac{1}{4}} \right\}, \\ 
\Pr\left[\left|\frac{1}{n}\sum_{i=1}^n X_{ij}X_{ik} - E[X_jX_k]\right| > \epsilon_2\right] &< \exp\left\{ - \frac{2}{L}\bigg(\frac{\sqrt{n}\epsilon_2}{\sqrt{\underset{j, k, \ell}{\max}~\Var[X_jX_k]}}\bigg)^{\frac{1}{2}} \right\}
\end{align*}
for all $j, k, \ell=1,\dots,\m$.  By a union bound over no more than
$2\m^3$ events, we have both $\|\mathbf{R}_1\|_\infty<\epsilon_1$ and
$\|r_2\|_\infty<\epsilon_2$ with probability at least $1 - 1/\m^{\tau_2-3}$ as
$\m \to \infty$.  Applying Theorem~\ref{generalthm} with the chosen
$\epsilon_1$ and $\epsilon_2$ thus shows that our choices for
$\lambda$ and $n$ lead to the claim in Corollary \ref{th2}.


\section{Discussion}\label{discussion}

This paper proposes the use of regularized score matching for
estimation of conditional independence graphs in high dimensions.  The
focus is on modifying the score matching loss of
\citet{hyv2005} with an $\ell_1$ penalty to accommodate underlying
sparsity, which is in the spirit of popular existing methods such as
glasso and neighborhood selection.  This said, any other
regularization scheme can be considered instead.  For instance, the
method from \cite{Defazio2012} can be applied to encourage hub
structure in the inferred graph.

Our study of the Gaussian example of \citet{Meinshausen2008} suggests
that $\ell_1$-reg\-ularized score matching falls in between
neighborhood selection and glasso in terms of conditions for required
for graph selection consistency.  Here, the glasso requires the most
stringent conditions, and the score matching approach appears to be
similar to pseudo-likelihood methods that work with symmetric
estimates of precision matrices, such as SPACE \citep{PengEtAl2009}
and subsequent reformulations such as CONCORD \citep{KhareEtAl2013}.
However, regularized score matching is particularly convenient in that
the score matching loss is a quadratic function,  even
for non-Gaussian exponential families.  This brings about
piecewise linear solution paths 
and  allows for a simple theoretical analysis.  We
anticipate that the simple structure of score matching will lead to
further advances in graphical modeling, such as computationally
efficient techniques to deal with corrupted or missing data, in the
spirit of \cite{loh2012}, or new methods to tune regularization
parameters, as in \cite{Lederer}.

Regularized score matching is an interesting method for Gaussian
models, as we showed empirically and theoretically.  In particular,
for consistency (under the usual irrepresentability conditions), the
sample $n$ must be on the order $\Omega(\dee^2\log \m)$, which matches
the conditions for the existing methods mentioned above.  However, as
our simulation study shows, regularized score matching really shines
in the context of non-Gaussian models, where it eliminates the need to
deal with computationally intractable normalization constants in a way
that the loss continues to be a quadratic function of parameters.
This opens a lot of new possibilities for graphical modeling such as
the truncated normal model we applied to RNAseq data.

Score matching applies to continuous data.  While \cite{hyv2007}
discusses a ratio matching method for discrete data, it is not as
computationally convenient as its continuous counterpart.  A different
approach of adding Gaussian noise to discrete data was proposed for
imaging problems by \citet{kingma2010regularized}. Exploring the
merits of their approach for graphical modeling, and supplying
supporting theory, would be an interesting problem for future work.

\appendix

\setcounter{lemma}{0}
\renewcommand{\theequation}{A.\arabic{equation}}
\renewcommand{\thelemma}{A.\arabic{lemma}}

\section{Implementation}\label{implementation}

The piecewise linear solution path for regularized score matching can
be computed using Algorithm~\ref{alg1}, which is an adaptation of the
LARS-Lasso algorithm for linear regression \citep{efron2004}.  It is
also a special case of the algorithm found in \citet{Rosset2007}.  In
our pseudocode, $\mycheck{S}$ is the current active set, i.e.,
$\mycheck{S}=\{j: \theta^\lambda_j \neq 0 \}$ for the currently
relevant value of the regularization parameter $\lambda$.  

\begin{algorithm}[t]
\caption{}\label{alg1}
\renewcommand{\algorithmicrequire}{\textbf{Input:}}
\renewcommand{\algorithmicensure}{\textbf{Output:}}
\begin{algorithmic}[1]
\State \textit{Initialize}  $\theta = 0$
\State \textit{Initialize} $\mycheck{S} = \mbox{arg}~\underset{j}{\max}\left|\left(\boldsymbol{\Gamma}(\mathbf{x})\theta + \g(\mathbf{x})\right)_j\right|$
\State \textit{Initialize} $\xi_{\mycheck{S}} = -\mbox{sign}\left(\left(\boldsymbol{\Gamma}(\mathbf{x})\theta + \g(\mathbf{x})\right)_{\mycheck{S}}\right)$
\State \textit{Initialize} $\xi_{\mycheck{S}^c} = 0$
\While  {$\left\|\boldsymbol{\Gamma}(\mathbf{x})\theta + \g(\mathbf{x})\right\|_\infty > 0$ and $\boldsymbol{\Gamma}_{\mycheck{S}\mycheck{S}}$ is invertible }
\State $\eta_1 \gets \min\{\eta > 0 : |(\boldsymbol{\Gamma}(\mathbf{x})\theta + \g(\mathbf{x})|_j = |(\boldsymbol{\Gamma}(\mathbf{x})\theta + \g(\mathbf{x})|_{\mycheck{S}}, j \notin \mycheck{S}\}$.
\State $\eta_2 \gets \min\{\eta > 0 : (\theta + \eta\boldsymbol{\xi})_j = 0, j \in \mycheck{S} \}$.
\State $\eta \gets \min\{\eta_1, \eta_2\}$. 
\State $\theta \gets \theta + \eta\xi$ 
\If {$\eta = \eta_1$}
    \State Add variable that attains equality to $\mycheck{S}$.  
\Else
     \State Remove variable that attains $0$ from $\mycheck{S}$.  
\EndIf
\State $\xi_{\mycheck{S}} \gets \left(\boldsymbol{\Gamma}(\mathbf{x})_{\mycheck{S}\mycheck{S}}\right)^{-1}\mbox{sign}(\theta_{\mycheck{S}})$
\EndWhile
\end{algorithmic}
\end{algorithm}

In the Gaussian and truncated Gaussian case, the algorithm stops when
the active set has size $|\mycheck{S}| = \min\{n, \m\}\m$.   For larger
active sets the matrix $\boldsymbol{\Gamma}_{\mycheck{S}\mycheck{S}}$
is not invertible.  Finding the step size in Algorithm~\ref{alg1}
requires $\mathcal{O}\left(\min\{n, \m\}\m\right)$ operations, while
the inversion step is at its worst $\mathcal{O}(|\mycheck{S}|^2) =
\mathcal{O}\left(\min\{n, \m\}^2\m^2\right)$. Overall, the complexity
of Algorithm~\ref{alg1} can be found to be $\mathcal{O}\left(\min\{n,
  \m\}^3\m^2\right)$; the heaviest cost comes from the matrix
inversion step.

For large-scale problems, LARS-type algorithms may be slow and
coord\-inate-descent methods are popular alternatives \citep[see
e.g.][]{Friedman2007}.  Algorithm~\ref{alg2} describes a
coordinate-descent algorithm to minimize the regularized score
matching objective from~(\ref{regquadratic}).  It entails updating one
coordinate, or one element in the parameter vector/matrix, such that
it minimizes the objective function while holding all others as
constant, until a convergence criterion is satisfied.  Results in
\citet{Tseng2001} ensure convergence of Algorithm~\ref{alg2}.

\begin{algorithm}[t]
\caption{}\label{alg2}
\renewcommand{\algorithmicrequire}{\textbf{Input:}}
\renewcommand{\algorithmicensure}{\textbf{Output:}}
\begin{algorithmic}[1]
\Require Initial estimate $\mycheck{\theta}^{(0)}$
\Require $t_{max}$, maximum number of iterations
\Require $\epsilon$, the maximal tolerance level
\State \textit{Initialize} $t \gets 1$
\State \textit{Initialize} $C \gets \epsilon+1$ ($C$ stands for convergence criteria)
\While  {$C > \epsilon$ or $t < t_{max}$ }
\State $\mycheck{\theta}^{(t)} \gets \mycheck{\theta}^{(t-1)}$
\For {$j \gets 1, 2, \ldots , s$}
\State    $\mycheck{\theta}_j^{(t)} \gets \mbox{Soft}\left(\frac{-\left(\boldsymbol{\Gamma}(\mathbf{x})_{-j,j}\right)^T\mycheck{\theta}_{-j}^{(t)} - \g(\mathbf{x})_j}{\boldsymbol{\Gamma}(\mathbf{x})_{jj}}, \frac{\lambda}{\boldsymbol{\Gamma}(\mathbf{x})_{jj}} \right)$.
\EndFor
\State $C \gets \|\mycheck{\theta}^{(t)} - \mycheck{\theta}^{(t-1)}  \|_1$ 
\State $t \gets t + 1$
\EndWhile
\end{algorithmic}
\end{algorithm}

\begin{examcont}{example1}
  For the Gaussian case, the coordinate descent procedure alternates
  between updating the diagonal entries and off-diagonal entries, by
  manipulating the estimating equations (\ref{diag}) and
  (\ref{offdiag}) accordingly. The updates are of the form
  \begin{align*}
    \kappa_{jj}^{(t+1)} &\leftarrow \frac{1 - \sum_{j' \neq j} w_{jj'}\kappa_{jj'}^{(t)}}{w_{jj}}, \\
    \kappa_{jk}^{(t+1)}, \kappa_{kj}^{(t+1)} &\leftarrow
    \mbox{Soft}\left(\frac{-\sum_{j' \neq j} w_{jj'}\kappa_{j'k}^{(t)}
        - \sum_{k' \neq k} w_{jk'}\kappa_{k'k}^{(t)}}{w_{jj} +
        w_{kk}}, \frac{2\lambda}{w_{jj} + w_{kk}}\right),
  \end{align*}
  for $j, k \in \{1, \ldots , \m\}$.  The computational complexity of
  this scheme can be shown to be $\min(\mathcal{O}(n\m^2),
  \mathcal{O}(\m^3))$, which is the same as for the methods classified
  under SPACE; the complexity of glasso is $\mathcal{O}(\m^3)$.  We do
  not prove this fact, as it follows directly from reasoning
  elaborated on in \citet{KhareEtAl2013}.
\end{examcont}


\renewcommand{\theequation}{B.\arabic{equation}}
\renewcommand{\thelemma}{B.\arabic{lemma}}

\section{Concentration results}
\label{sec:conc-results}

Corollaries \ref{th1} and \ref{th2} make use of the following
concentration results.  The first lemma is used to prove Corollary
\ref{th1} while the latter two (one is derived from the other) are
used to prove Corollary \ref{th2}.
\begin{lemma}[\citeauthor{RavikumarEtAl2011}, \citeyear{RavikumarEtAl2011}]
  \label{lem:ravi-concentration}
  If $\begin{pmatrix}X_1, \ldots , X_\m \end{pmatrix}$ is a zero-mean
  random vector with covariance matrix $\boldsymbol{\Sigma}^*$ such
  that $X_i/\sqrt{\Sigma_{ii}^*}$ is sub-Gaussian with scale parameter
  $\sigma$, then the sample covariance matrix $\mathbf{W}$, for $n$
  i.i.d.\ samples, satisfies the bound
\begin{align}
\Pr[|\mathbf{W}_{jk} - \boldsymbol{\Sigma}_{jk}^*| > \delta] &\leq 4\exp\left\{-\frac{n\delta^2}{128(1+4\sigma^2)^2\underset{j =1,\ldots , \m}{\max}{(\boldsymbol{\Sigma}^*_{jj})^2}} \right\} 
\end{align}
for any fixed choice of two indices $1\le j, k\le \m$ and for
all
$\delta \in (0, 40~\underset{j = 1, \ldots , \m}{\max}
\boldsymbol{\Sigma}_{jj}^*)$.
\end{lemma}
\begin{lemma}[\citeauthor{Carbery2001}, \citeyear{Carbery2001}]
  \label{markovlemma}
  Let $\mathcal{X}$ be a Banach space, and let
  $f: \mathbb{R}^\m \to \mathcal{X}$ be a polynomial of degree at most
  $z$.  Suppose $0 < \zeta_1 \leq \zeta_2 < \infty$ and $\mu$ is a
  log-concave probability measure on $\mathbb{R}^\m$. Then
  \begin{equation}
    \left(\int \|f(x)\|^{\zeta_2/z}d\mu(x)
    \right)^{1/\zeta_2}\; \leq\; L\frac{\max(\zeta_2,
      1)}{\max(\zeta_1, 1)}\left(\int
      \|f(x)\|^{\zeta_1/z}d\mu(x) \right)^{1/\zeta_1}, 
  \end{equation} 
  where $L>0$ is an absolute constant.
\end{lemma}
From  this lemma we may derive the following concentration result.
After proving the lemma, we comment on how it is used in the proof of
Corollary~\ref{th2}.   

\begin{lemma}\label{polynomiallemma} 
  Consider a degree $z$ polynomial $f(X) = f(X_1, \ldots , X_\m)$,
  where $X_1, \ldots , X_\m$ are possibly dependent random variables
  with log-concave joint distribution on $\mathbb{R}^\m$.  Let $L > 0$
  be the constant from Lemma~\ref{markovlemma}.  Then, for all
  $\delta$ such that
  \begin{equation}\label{kcondition}
  K := \frac{2}{L}\left(\frac{\delta}{e\sqrt{\Var[f(X)]}} \right)^{1/z} \ge 2,
  \end{equation}
we have, 
\begin{equation}
\Pr[|f(X) - E[f(X)]| > \delta] \;\leq\; \exp\left\{-\frac{2}{L}\left(\frac{\delta}{\sqrt{\Var[f(X)]}}\right)^{1/z} \right\}.
\end{equation}
\end{lemma}
\begin{proof}
Choosing $\zeta_1 = 2z$ and $\zeta_2 = Kz$ in
  Lemma~\ref{markovlemma}, we have 
  \begin{align*}
    E[|f(X) - E[f(X)]|^{K}]^{\frac{1}{K}} 
    &\;\leq\; \left(\frac{LK}{2}\right)^z\sqrt{\Var[f(X)]}. 
  \end{align*}
  Hence, by Markov's inequality, for any $\delta$ satisfying (\ref{kcondition}),
  \begin{align}\label{markov}
    P[|f(X) - E[f(X)]| > \delta] 
    &\;\leq\;  \frac{E[|f(X) -
      E[f(X)]|^{K}]}{\delta^{K}}\\
    &\;\leq\;  
      \left[
      \left(\frac{LK}{2}\right)^z\frac{\sqrt{\Var[f(X)]}}{\delta}
      \right]^K  \\
     & \;=\;  
     \mbox{exp}\{-K\} \\
      & \;=\; \mbox{exp}\left\{-\frac{2}{L}\left(\frac{\delta}{\sqrt{\Var[f(X)]}}\right)^{\frac{1}{z}} \right\},
  \end{align}
  and the proof is complete.
\end{proof}


In the proof of Corollary \ref{th2}, we apply
Lemma~\ref{polynomiallemma} with $\delta=\epsilon_1$ from~(\ref{e1})
and with $\delta=\epsilon_2$ from~(\ref{e2}).  It thus needs to be
checked that condition (\ref{kcondition}) holds in these two cases.
Indeed, the condition holds as long as
\begin{equation}
  \label{mmm} 
  \m \;\geq\; \exp\left\{\frac{2\sqrt{e} - \log
        2}{\tau_2} \right\}.
\end{equation}
To see this, we substitute $\epsilon_1$ and $\epsilon_2$ for $\delta$
in (\ref{kcondition}), take $z = 4$ and $2$ respectively, to find a
term that is lower bounded by $(\tau_2\log \m + \log 2)/e^2$.  Here,
the $1/\sqrt{n}$ factor in $\epsilon_1$ and $\epsilon_2$ cancels out
with the $1/\sqrt{n}$ term generated by the $\sqrt{\Var[f(X)]}$
term in the denominator.  (Recall that in our scenario $f(X)$ is an
empirical average).  The more stringent condition on $\m$ comes from
$\epsilon_2$ and is stated in (\ref{mmm}).  Thus, if~(\ref{mmm})
holds, (\ref{kcondition}) is satisfied.  
Since $\tau_2>3$, the right-hand side
of~(\ref{mmm}) never exceeds 
\[
\exp\left\{\frac{1}{3}(2\sqrt{e} - \log
        2) \right\} < 3.
\]
Hence, in our application of Lemma~\ref{polynomiallemma}, the
condition from~(\ref{kcondition}) holds for $m\ge 3$.

\renewcommand{\theequation}{C.\arabic{equation}}
\renewcommand{\thelemma}{C.\arabic{lemma}}

\section{Experiments}
\label{sec:experiments}

We perform experiments, similar to those found in related work, that
give empirical support for Corollary \ref{th1}.  This corollary treats
Gaussian graphical models for which the sample size $n$ ought to be of
order $\dee^2\log \m$.  We experiment by varying the number of
variables $\m$, the degree $\dee$, and the minimum signal strength.
Following \citet{RavikumarEtAl2011}, we define the `model
complexity' to be
\begin{equation}
  \label{eq:model-complexity}
  C := \frac{4}{\alpha}c_{\boldsymbol{\Gamma}^*} \times
\underset{j}{\max}~ \boldsymbol{\Sigma}^*_{jj}. 
\end{equation}
In addition, we investigate how the sample size $n$ required for
sparsistency for non-negative Gaussian graphical models needs to
depend on $\m$.  All reported results are based on averaging over 100
trials.

\subsection{Gaussian experiments}

We conduct our experiments using three graph structures: (a) a chain,
(b) a 2-D lattice with 4 nearest neighbors, and (c) a star.  We
consider (a) and (b) when varying the number of variables $\m$, in
which case we vary the length of the chain and the number of nodes in
the lattice.  This keeps the degree $\dee$ constant.  The effect that
$\dee$ has on the sample complexity is investigated using stars. 
We let the regularization parameter $\lambda$ scale with
$\sqrt{\log \m/n}$, a choice corroborated by Corollary \ref{th1}.


\subsubsection*{Dependence on number of nodes}\label{numbernodes}

Consider first the case where the underlying conditional independence
graph is a chain of length $\m \in\{64, 100,225,375\}$.  The degree
$\dee$ is always 2, and we choose the tridiagonal precision matrix
$\mathbf{K}^*$ to have entries $\kappa_{jk}^* = 0.3$ if $(j, k) \in E$
and $\kappa_{jj}^* = 1$ for $j = 1, \dots \m$.  Here,
$\alpha$, $c_{\mathbf{K}^*}$ and $c_{\boldsymbol{\Gamma}^*}$ are
constant across all $\m$.

Figure \ref{chain} shows the probability of correct signed support
recovery plotted against the sample size $n$, with different curves
corresponding to different $\m$.  As expected, we see from Figure
\ref{chain}(a) that successful support recovery requires $n$ to grow
with $\m$.  However, upon rescaling $n$ by $1/\log \m$, the curves
overlap as seen in Figure
\ref{chain}(b).

\begin{figure}[t]
\captionsetup[subfigure]{justification=centering}
 \centering
        \begin{subfigure}[b]{0.48\textwidth}
                {\label{chain1} \includegraphics[scale = 0.5]{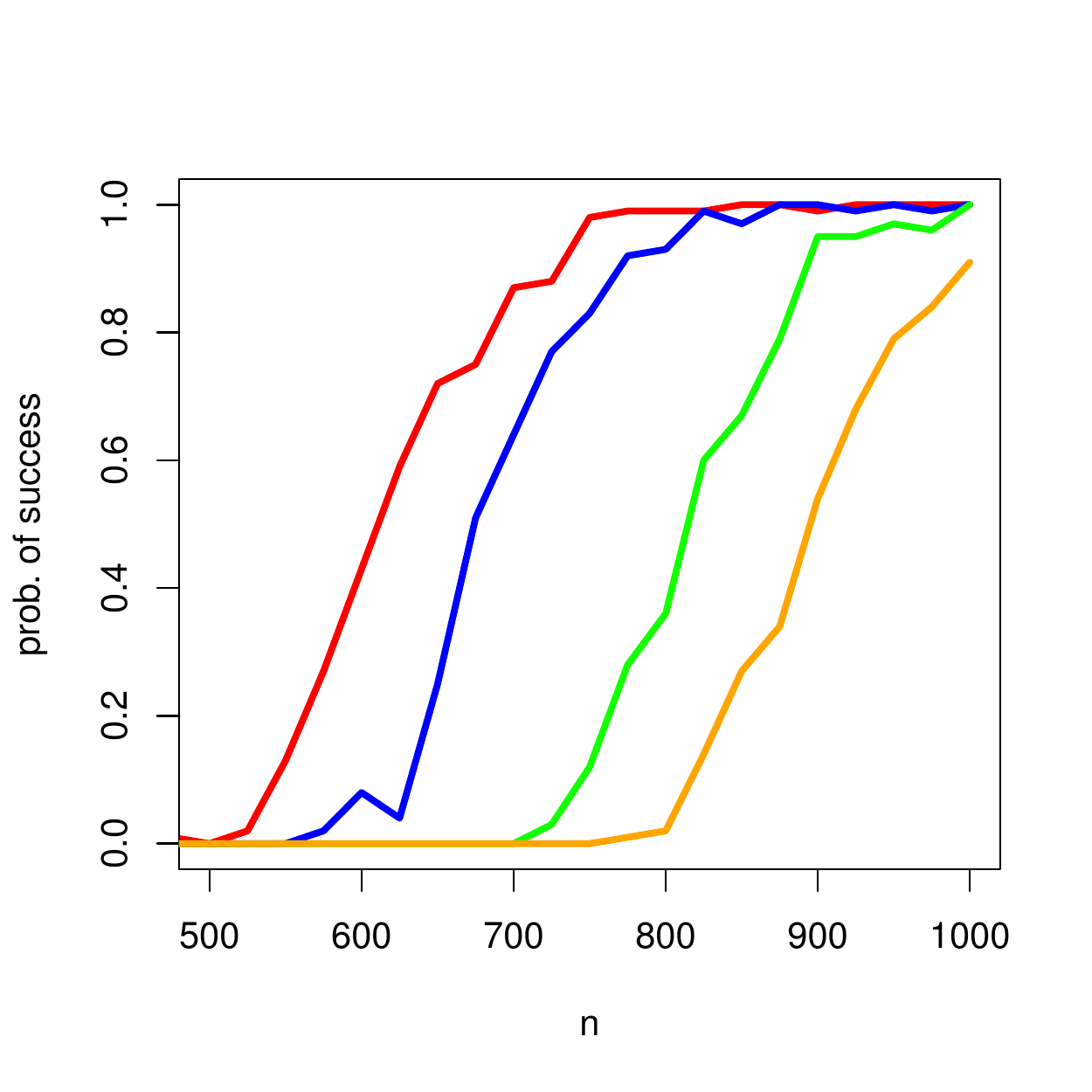}}
                \caption{}
        \end{subfigure}
        \begin{subfigure}[b]{0.48\textwidth}
                {\label{chain2} \includegraphics[scale = 0.5]{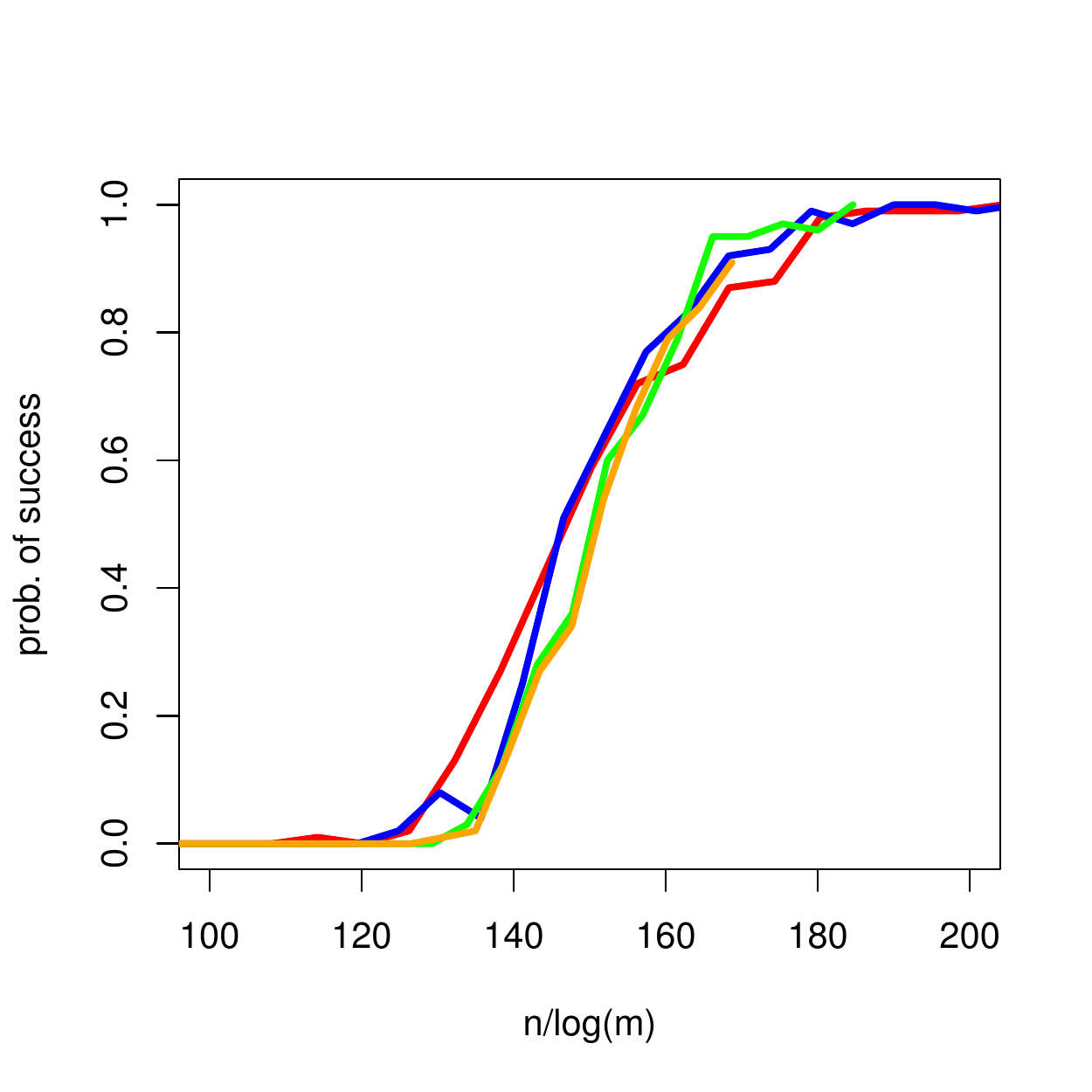}}
                \caption{}
        \end{subfigure}
        \caption{Relative frequencies of signed support recovery for Gaussian
          observations with a conditional independence graph that is a
          chain of varying length $\m$.  Panels (a) and (b) differ
          only in the scaling of the $x$-axis. The colored lines correspond to $m = 64$ (\ref{firstcolour}), $m = 100$ (\ref{secondcolour}), $m = 225$ (\ref{thirdcolour}) and $m = 375$ (\ref{fourthcolour}).}
        \label{chain}
\end{figure}

We repeat the experiment with the 2-D lattice graph with
$\m \in\{64,100,225\}$ nodes.  Each node is connected to four nearest
neighbors such that the degree $\dee$ is always 4.  We choose
$\mathbf{K}^*$ with $\kappa_{jk}^* = 0.2$ for $(j, k) \in E$ and
$\kappa_{jj}^* = 1$ for $j = 1, \dots \m$.  Again, $\alpha$,
$c_{\mathbf{K}^*}$ and $c_{\boldsymbol{\Gamma}^*}$ are constant across
all $\m$.  The results are presented in Figure \ref{lattice}, which
shows curves of recovery probabilities that stack on top of one
another when $n$ by $1/\log \m$. 

We conclude that with $C$ and $\dee$ held constant, the sample size
$n$ needs to scale with $\log \m$ for consistent signed support
recovery.  This is consistent with Corollary \ref{th1}.

\begin{figure}[t]
\captionsetup[subfigure]{justification=centering}
\centering
        \begin{subfigure}[b]{0.48\textwidth}
                {\label{lattice1} \includegraphics[scale = 0.5]{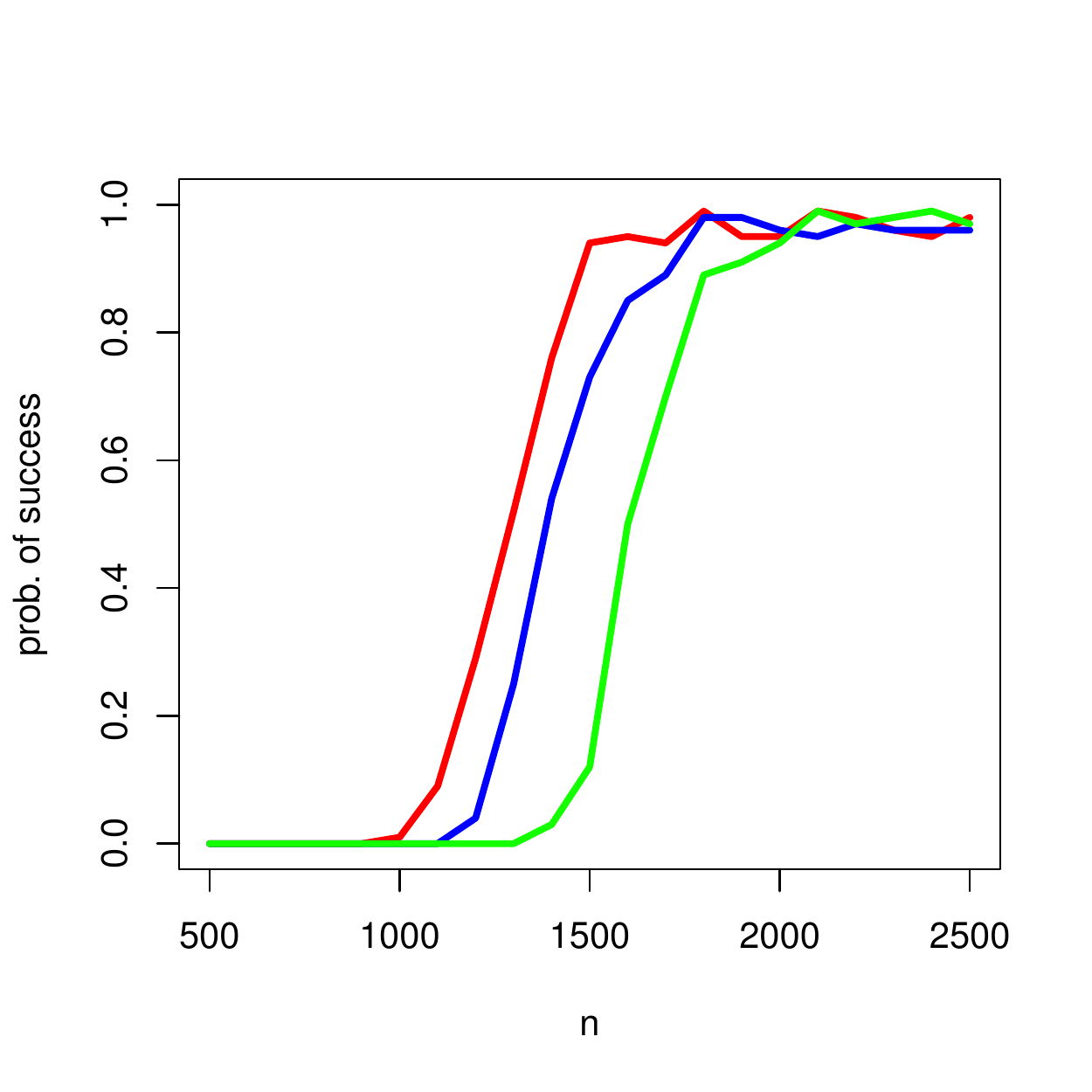}}
                \caption{}
        \end{subfigure}
        \begin{subfigure}[b]{0.48\textwidth}
                {\label{lattice2} \includegraphics[scale = 0.5]{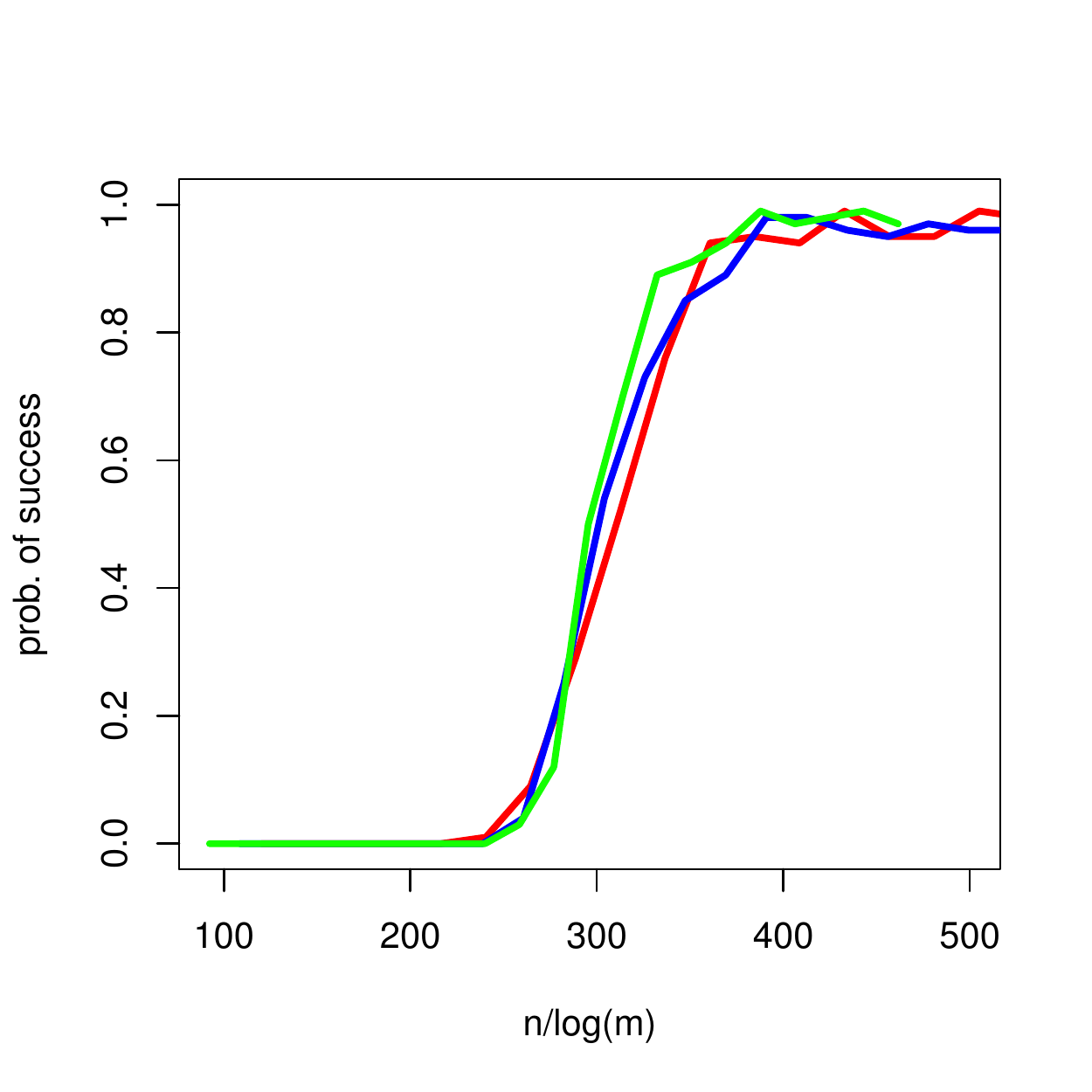}}
                \caption{}
        \end{subfigure}
        \caption{Relative frequencies of signed support recovery for Gaussian
          observations whose conditional independence graph is a
          4-nearest neighbor lattice with $\m$ nodes.  Panels
          (a) and (b) differ 
          only in the scaling of the $x$-axis. The colored lines correspond to $m = 64$ (\ref{firstcolour}), $m = 100$ (\ref{secondcolour}), and $m = 225$ (\ref{thirdcolour}).}
        \label{lattice}
\end{figure}

\subsubsection*{Dependence on node degree}

We now fix the number of nodes to $\m=200$ and vary $\dee$.  We
consider a star graphs with varying hub node degree
$\dee\in\{15,20,25\}$.  The precision matrix $\mathbf{K}^*$ is chosen
such that $\sigma_{jk}^* = 2.5/\dee$ for $(j, k) \in E$, and
$\sigma_{jj}^* = 1$ for $j = 1, \dots \m$.  Now, $\alpha$,
$c_{\mathbf{K}^*}$ and $c_{\boldsymbol{\Gamma}^*}$ are constant across
all $\dee$.

Figure \ref{star} shows the probability of correct signed support
recovery plotted against $n$.  The left panel demonstrates that
correct recovery is more difficult with increasing $\dee$.  Larger $n$
is needed to attain the same success rate.  Upon rescaling $n$ by
$1/\dee^2$ in the right panel, the three curves align.  This validates
Corollary \ref{th1} in that for fixed $\m$, $\alpha$,
$c_{\mathbf{K}^*}$ and $c_{\boldsymbol{\Gamma}^*}$, the sample size
$n$ needs to scale with $\dee^2$ to ensure sign consistency.

\begin{figure}[p]
\captionsetup[subfigure]{justification=centering}
 \centering
        \begin{subfigure}[b]{0.48\textwidth}
                {\label{star1} \includegraphics[scale = 0.5]{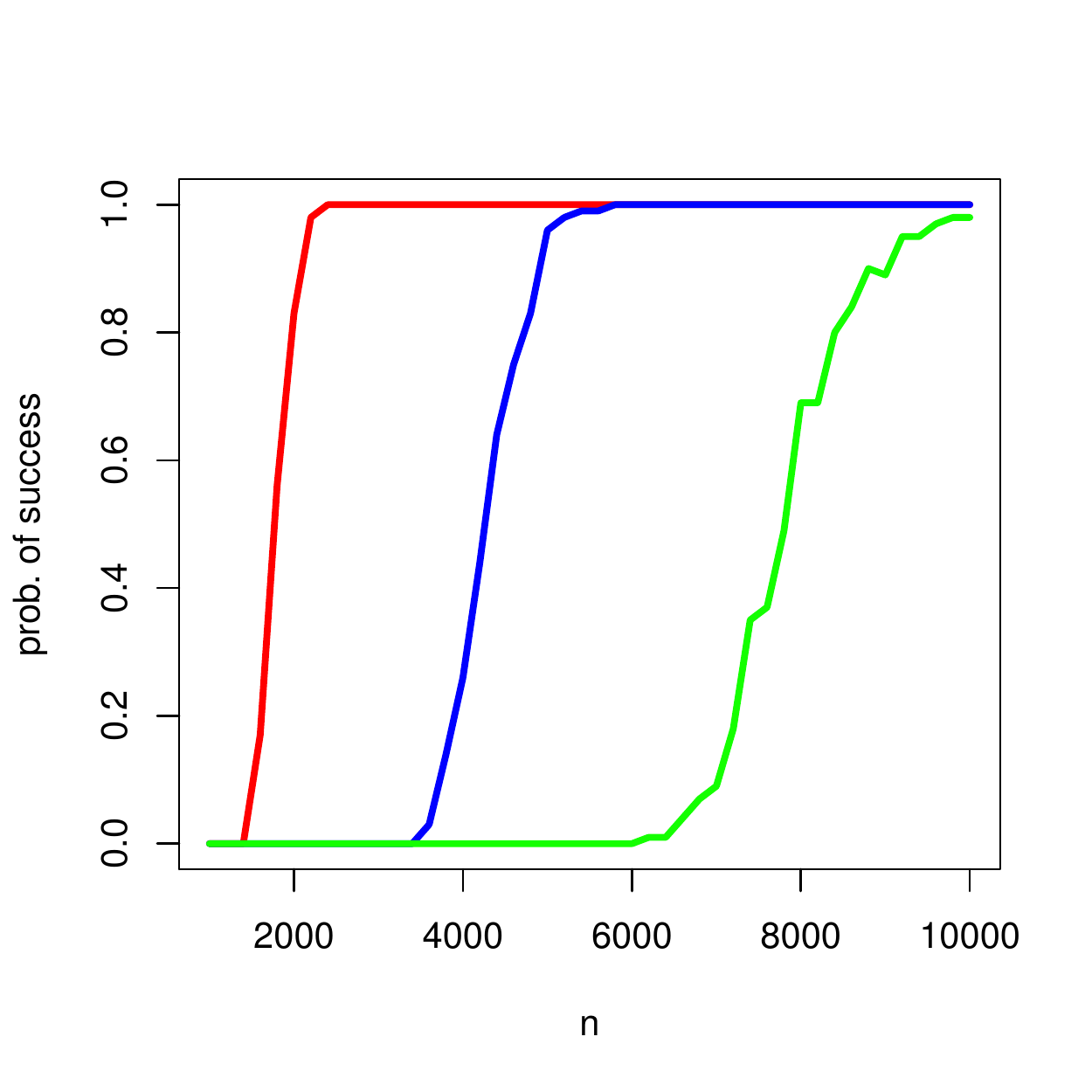}}
                \caption{}
        \end{subfigure}
        \begin{subfigure}[b]{0.48\textwidth}
                {\label{star2} \includegraphics[scale = 0.5]{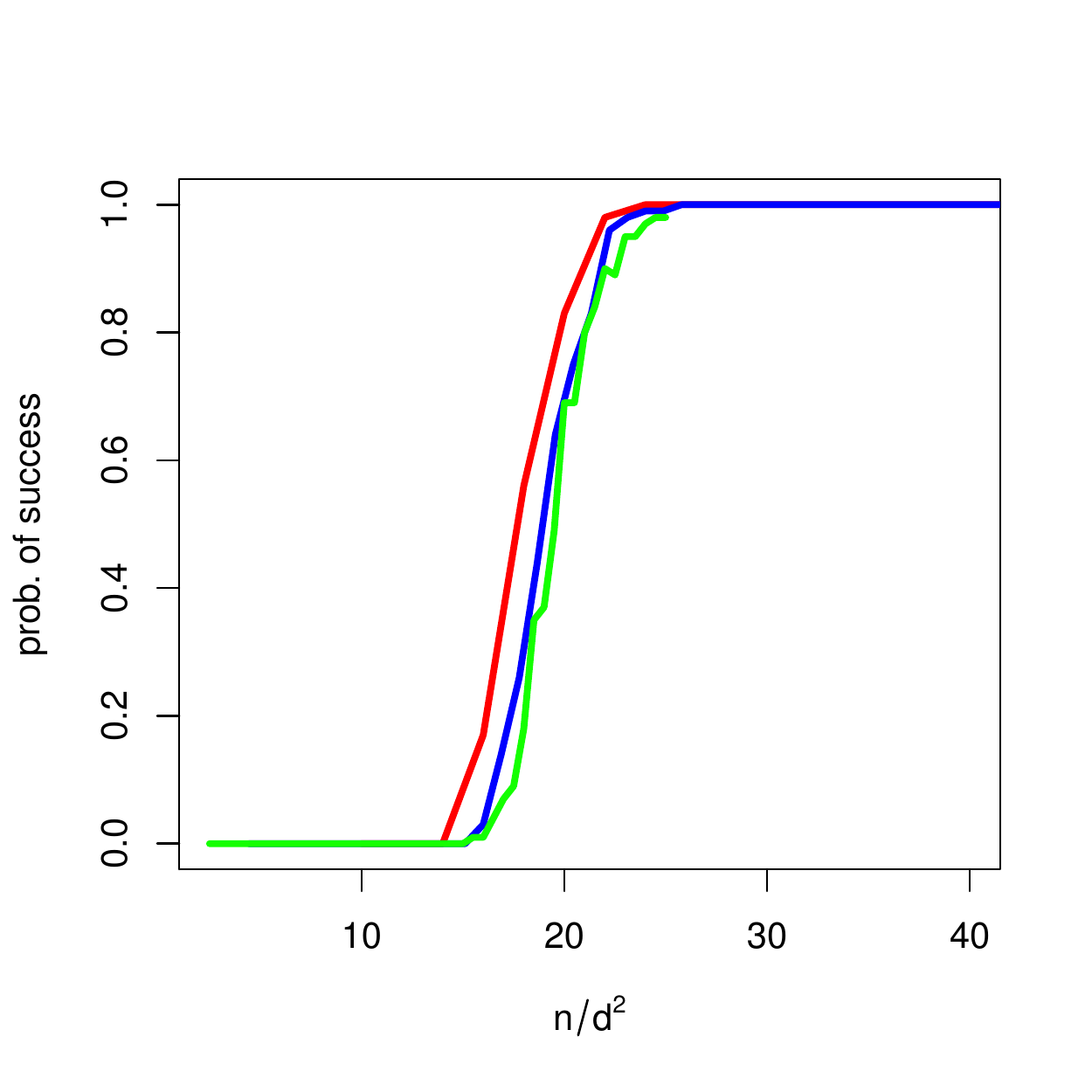}}
                \caption{}
        \end{subfigure}
        \caption{
Relative frequencies of signed support recovery for Gaussian
          observations whose conditional independence graph is a star
          with varying degree $\dee$.  Panels
          (a) and (b) differ 
          only in the scaling of the $x$-axis. The colored lines correspond to $d = 10$ (\ref{firstcolour}), $d = 15$ (\ref{secondcolour}), and $d = 20$ (\ref{thirdcolour}).}
        \label{star}
\end{figure}

\subsubsection*{Dependence on `model complexity'}

We return to the chain-structured graphs considered earlier in this
section.  This time, however, we fix $\m=64$ and $\dee=2$ while
changing the edge strengths $\kappa_{jk}^*$ for $(j, k) \in E$, which
alters $C$ from~(\ref{eq:model-complexity}).  We plot the probability
of correct signed support recovery against $n$ for varying $C$.  In
the resulting Figure \ref{chaincomplex}, the curves shift right as $C$
becomes larger so a larger $n$ is needed to attain the same
probability of correct signed support recovery when $C$ grows.  This
is again consistent with the implications of Corollary \ref{th1}.  We
do not believe that the lower bound we found for $n$ is sharp enough
in terms of its dependence on $\alpha$, $c_{\mathbf{K}^*}$ and
$c_{\boldsymbol{\Gamma}^*}$ to determine the rescaling we must perform
on $n$ to align the curves.

\begin{figure}[p]
\begin{center}
\centerline{\includegraphics[scale = 0.48]{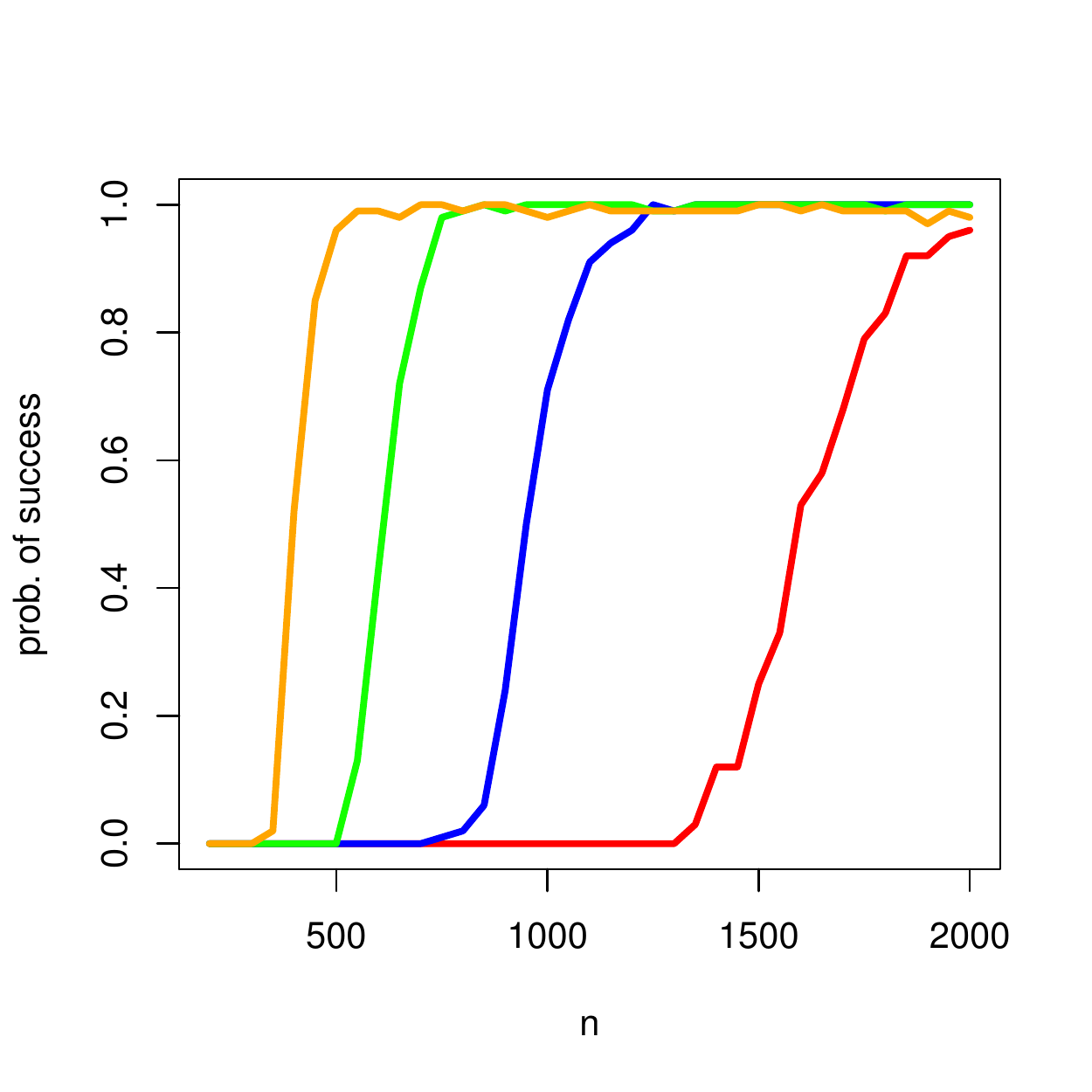}}
\caption{Relative frequencies of signed support recovery for Gaussian
          observations whose conditional independence graph is a
          chain of fixed length.  The different curves correspond to
          different signal strength summarized in the model complexity
          $C$. The colored lines correspond to $C = 857$ (\ref{firstcolour}), $C = 668$ (\ref{secondcolour}), $C = 576$ (\ref{thirdcolour}) and $C = 543$ (\ref{fourthcolour}.)
        }
\label{chaincomplex}
\end{center}
\vskip -0.2in
\end{figure} 

\subsection{Non-negative Gaussian experiments}

\label{numbernodes2}

Finally, we experiment with regularized non-negative score matching
for normal observations truncated to the positive orthant.  According
to Corollary \ref{th2}, a sample size of
$n = \Omega(\dee^2(\log \m)^8)$ is sufficient for signed support
recovery.  The aim of our experiments is to explore to what extent
this scaling is necessary.  Specifically, we will consider exponents
other than 8 for 
$\log\m$.

For our experiments, we revisit the chain-structured graphs from
Section \ref{numbernodes} and choose a triangular matrix
$\mathbf{K}^*$ with $\kappa_{jk}^* = 0.3$ if $(j, k) \in E$ and and
$\kappa_{jj}^* = 1$ for $j = 1, \dots \m$.  The degree $\dee$ is fixed
at 2 and we only vary $\m\in\{20,25,30\}$.  We let the regularization
parameter $\lambda$ to scale with $\sqrt{(\log \m)^8/n}$.
Figure \ref{chain22} plots the probability of correct signed support recovery
against $n$, with different curves for the different values of $\m$.

Panel (a) in Figure \ref{chain22} illustrates that, larger $n$ is
needed account for larger $\m$.  The other three panels have the
$x$-axis rescaled to $n/(\log\m)^a$ for exponents $a\in\{6,7,8\}$.
Panel (b) suggests that $n$ scaling with $(\log\m)^6$ is not
sufficient for support recovery.  Comparing panels (c) and (d),
$(\log\m)^8$ seems more than what is necessary.  It thus appears that
the scaling of the sample size we assumed in Corollary \ref{th2} is
suboptimal but not drastically so.

\begin{figure}
\vskip -0.2in
\captionsetup[subfigure]{justification=centering}
 \centering
        \begin{subfigure}[b]{0.48\textwidth}
                {\includegraphics[scale = 0.5]{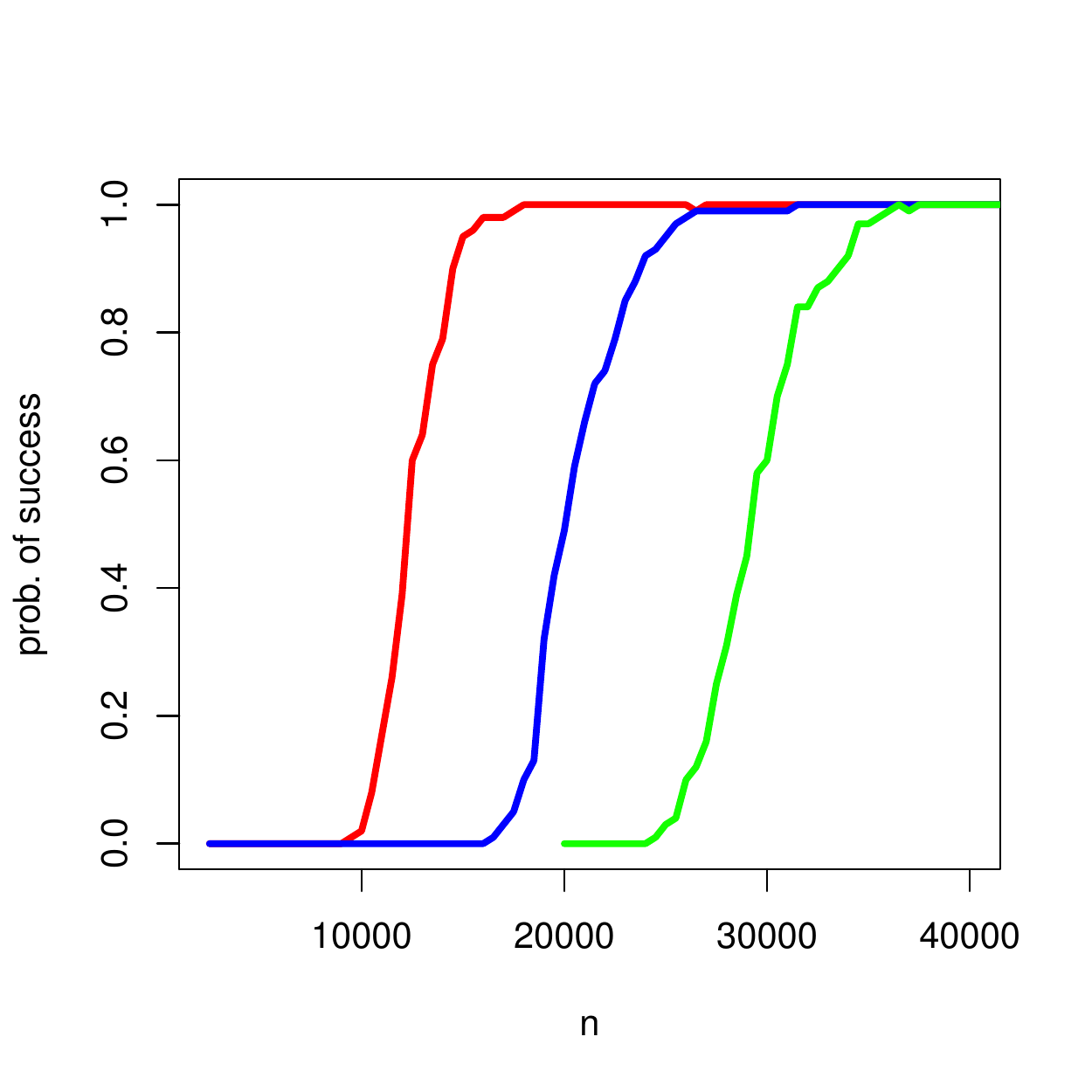}}
                \caption{}
        \end{subfigure}
        \begin{subfigure}[b]{0.48\textwidth}
                {\includegraphics[scale = 0.5]{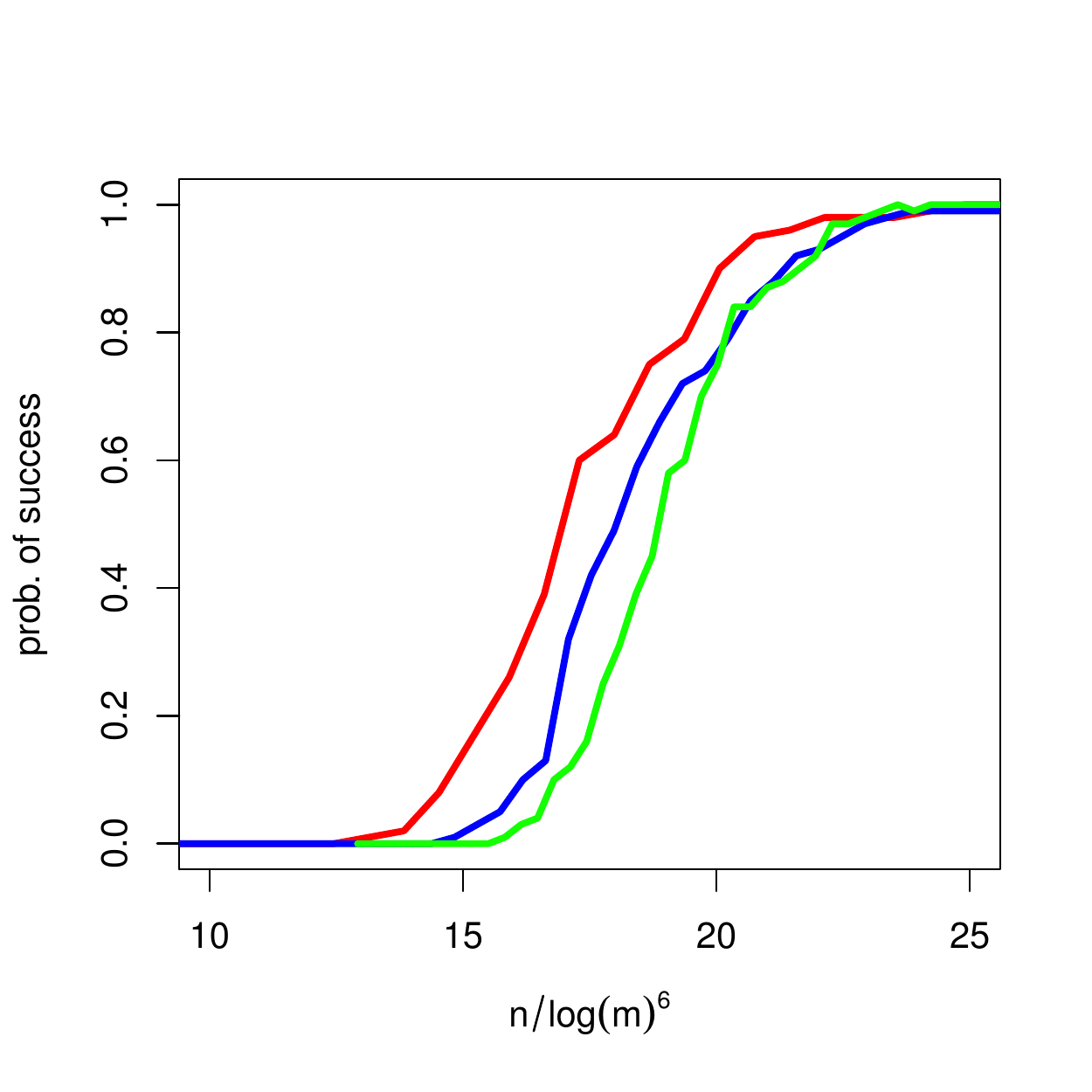}}
                \caption{}
        \end{subfigure}
         \begin{subfigure}[b]{0.48\textwidth}
                {\includegraphics[scale = 0.5]{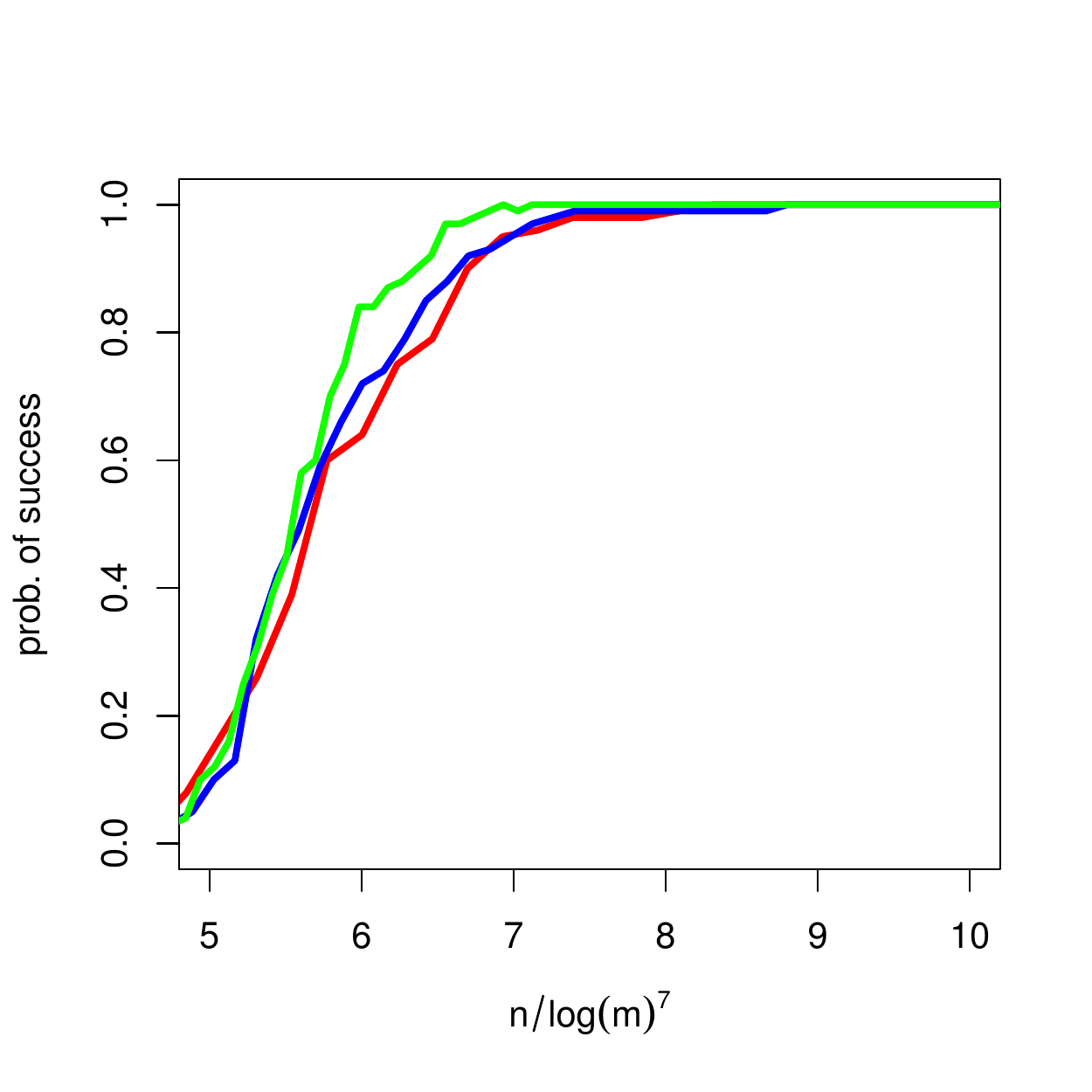}}
                \caption{}
        \end{subfigure}
         \begin{subfigure}[b]{0.48\textwidth}
                {\includegraphics[scale = 0.5]{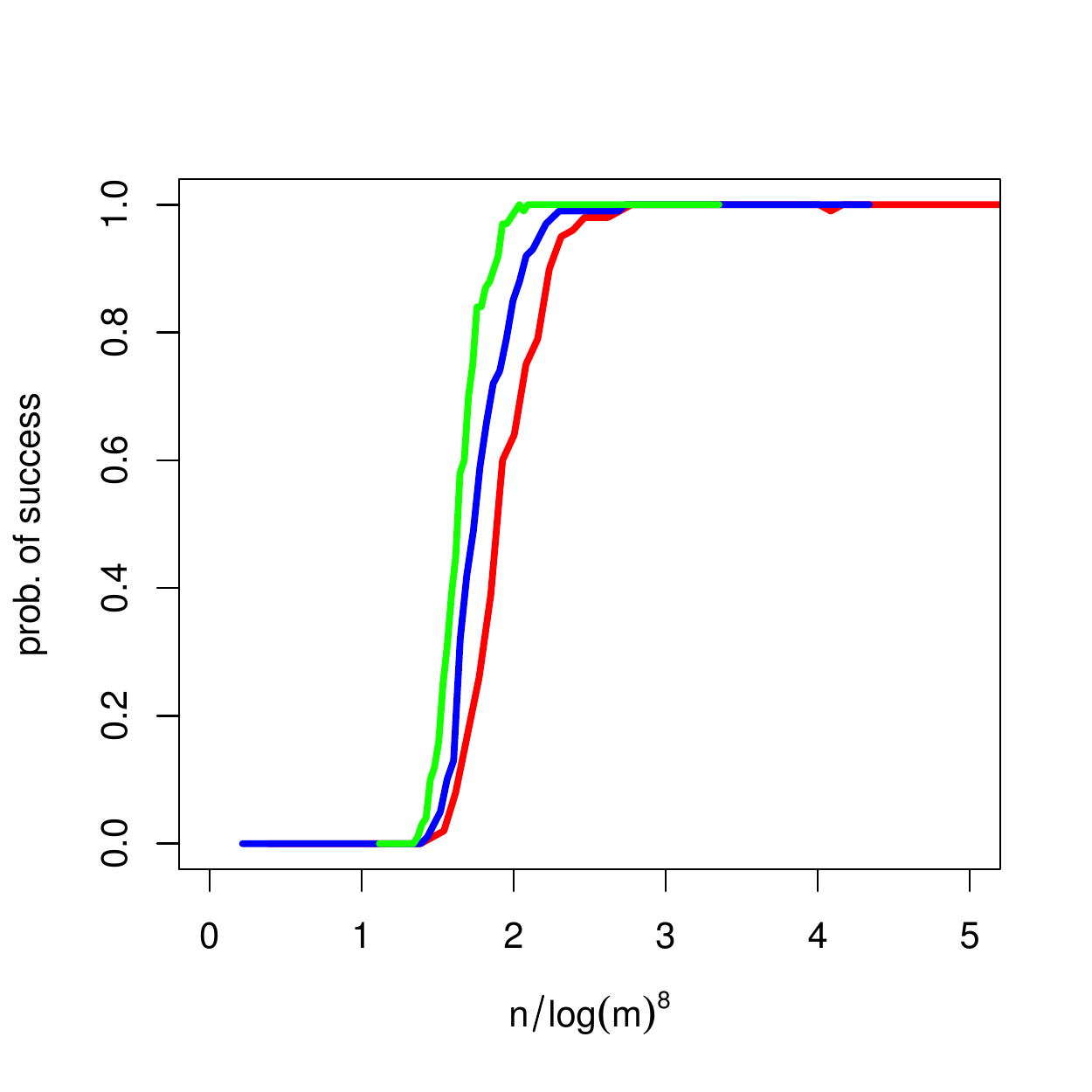}}
                \caption{}
        \end{subfigure}
        \caption{
          Relative frequencies of signed support recovery for
            truncated Gaussian
          observations whose conditional independence graph is a
          chain of varying length $\m$.  The four panels differ only
          in the scaling of the $x$-axis. The colored lines correspond to $m = 20$ (\ref{firstcolour}), $m = 25$ (\ref{secondcolour}), and $m = 30$ (\ref{thirdcolour}).}
        \label{chain22}
\end{figure}

\scalebox{0}{%
\begin{tikzpicture}
    \begin{axis}[hide axis]
        \addplot [
        color=mycolor,
        solid,
        line width=0.9pt,
        forget plot
        ]
        (0,0);\label{scorecolour}
            \end{axis}
\end{tikzpicture}%
}%

\scalebox{0}{%
\begin{tikzpicture}
    \begin{axis}[hide axis]
       \addplot [
        color=mycolor2,
        solid,
        line width=1.2pt,
        forget plot
        ]
        (0,0);\label{spacejamcolour}
            \end{axis}
\end{tikzpicture}%
}%
\scalebox{0}{%
\begin{tikzpicture}
    \begin{axis}[hide axis]
 \addplot [
        color= mycolor3, 
        solid,
        line width=1.2pt,
        forget plot
        ]
        (0,0);\label{glassocolour}
           \end{axis}
\end{tikzpicture}%
}%

\scalebox{0}{%
\begin{tikzpicture}
    \begin{axis}[hide axis]
\addplot [
        color=mycolor4,
        solid,
        line width=1.2pt,
        forget plot
        ]
        (0,0);\label{skepticcolour}
            \end{axis}
\end{tikzpicture}%
}%

\scalebox{0}{%
\begin{tikzpicture}
    \begin{axis}[hide axis]
\addplot [
        color=mycolor5,
        solid,
        line width=1.2pt,
        forget plot
        ]
        (0,0);\label{spacecolour}
            \end{axis}
\end{tikzpicture}%
}%

\scalebox{0}{%
\begin{tikzpicture}
    \begin{axis}[hide axis]
\addplot [
        color=mycolor6,
        solid,
        line width=1.2pt,
        forget plot
        ]
        (0,0);\label{mbcolour}
            \end{axis}
\end{tikzpicture}%
}%

\vspace{-17mm}
\scalebox{0}{%
\begin{tikzpicture}
    \begin{axis}[hide axis]
        \addplot [
        color=red,
        solid,
        line width=0.9pt,
        forget plot
        ]
        (0,0);\label{firstcolour}
            \end{axis}
\end{tikzpicture}%
}%

\scalebox{0}{%
\begin{tikzpicture}
    \begin{axis}[hide axis]
        \addplot [
        color=blue,
        solid,
        line width=0.9pt,
        forget plot
        ]
        (0,0);\label{secondcolour}
            \end{axis}
\end{tikzpicture}%
}%

\scalebox{0}{%
\begin{tikzpicture}
    \begin{axis}[hide axis]
        \addplot [
        color=green,
        solid,
        line width=0.9pt,
        forget plot
        ]
        (0,0);\label{thirdcolour}
            \end{axis}
\end{tikzpicture}%
}%

\scalebox{0}{%
\begin{tikzpicture}
    \begin{axis}[hide axis]
        \addplot [
        color=mycolor3,
        solid,
        line width=0.9pt,
        forget plot
        ]
        (0,0);\label{fourthcolour}
            \end{axis}
\end{tikzpicture}%
}%

\bibliographystyle{imsart-nameyear} 
{\small\bibliography{refsall}}

\end{document}